\documentclass[a4paper,11pt]{article}

\usepackage[english]{babel}
\usepackage[a4paper]{geometry}
\usepackage{amsmath}
\usepackage{amsthm}
\usepackage{amssymb}

\usepackage{color}
\usepackage{hyperref}
\hypersetup{
			colorlinks=true,
			linkcolor=black,
			citecolor=black,
}

\newcommand{\di}{d}
\newcommand{\scal}[2]{\big<#1,#2\big>} 

\let\a=\alpha \let\b=\beta         \let\d=\delta     
             \let\l=\lambda
                          
          \let\ph=\varphi   
\let\ps=\psi   \let\o=\omega     
 \let\D=\Delta

\def\VV{{\cal V}}

\def\NN{{\cal N}}
\def\LL{{\cal L}}
\def\DD{{\cal D}}
\def\UU{{\cal U}}

\def\ff{{\mathfrak f}}

\def\ee{{\underline \varepsilon}}

  \def\v0{{\vec 0}}

\def\bR{\mathbb{R}}
\def\bN{\mathbb{N}}
\def\cU{\mathcal{U}}
\def\cV{\mathcal{V}}
\def\cF{\mathcal{F}}
\def\cG{\mathcal{G}}
\def\cL{\mathcal{L}}
\def\cN{\mathcal{N}}
\def\cE{\mathcal{E}}
\def\cK{\mathcal{K}}
\def\cH{\mathcal{H}}

\def\ph{\varphi}

\def\NNN{\mathbb{N}}  
 
\def\RRR{\mathbb{R}}

\def\indic{\hbox{\raise-2pt \hbox{\indbf 1}}}

\let\dpr=\partial

\let\==\equiv

\let\io=\infty
\let\0=\noindent

\def\*{{\hfill\break\null\hfill\break}}

\def\bmedia#1{{\bigl\langle#1\bigr\rangle}}

\def\tende#1{\,\vtop{\ialign{##\crcr\rightarrowfill\crcr
             \noalign{\kern-1pt\nointerlineskip}
             \hskip3.pt${\scriptstyle #1}$\hskip3.pt\crcr}}\,}
\def\otto{\,{\kern-1.truept\leftarrow\kern-5.truept\to\kern-1.truept}\,}

\def\tr{{\rm tr}}

\def\wt#1{\widetilde{#1}}

\def\sqt[#1]#2{\root #1\of {#2}}

\def\Im{{\rm Im}\,}
\def\lis{\overline}

\def\norml#1#2{{\left|\hskip-.05em\left|#1\right|\hskip-.05em\right|}_{#2}}

\def\ins#1#2#3{\vbox to0pt{\kern-#2 \hbox{\kern#1 #3}\vss}\nointerlineskip}

\def\be{\begin{equation}}
\def\ee{\end{equation}}
\def\bea{\begin{eqnarray}}\def\eea{\end{eqnarray}}
\def\bean{\begin{eqnarray*}}\def\eean{\end{eqnarray*}}
\def\bfr{\begin{flushright}}\def\efr{\end{flushright}}
\def\bc{\begin{center}}\def\ec{\end{center}}
\def\bal{\begin{align}} \def\eal{\end{align}}

\def\ba#1{\begin{array}{#1}} \def\ea{\end{array}}
\def\bd{\begin{description}}\def\ed{\end{description}}
\def\bv{\begin{verbatim}}\def\ev{\end{verbatim}}
\def\non{\nonumber}
\def\Halmos{\hfill\vrule height10pt width4pt depth2pt \par\hbox to \hsize{}}

\newtheorem{theorem}{Theorem}[section]  
\newtheorem{prop}[theorem]{Proposition}
\newtheorem{lemma}[theorem]{Lemma}

\numberwithin{equation}{section}

\def \phis#1{\ph_t^{N}({(#1)}/{2})}  
  
\def \phiq#1{\big( \ph_t^{N}({(#1)}/{2})\,\big)}

\def \dotphis#1{\dot\ph_t^{N}({(#1)}/{2})}  
\def \ddotphis#1{\ddot\ph_t^{N}({(#1)}/{2})}

\def\phn{\varphi^{N}}

\def\norml#1#2{{|\hskip-.05em|#1|\hskip-.05em|}_{#2}}
\newcommand{\expv}[2]{\langle #1,#2 \rangle}

{  

\author{Chiara Boccato, Serena Cenatiempo, Benjamin Schlein \\ \\ Institute of Mathematics, University of Zurich\\ Winterthurerstrasse 190, 8057 Zurich, Switzerland}

\title{Quantum many-body fluctuations around nonlinear Schr\"odinger dynamics}

\begin{document}

\maketitle

\begin{abstract}
We consider the many body quantum dynamics of systems of bosons interacting through a two-body potential $N^{3\beta-1} V (N^\beta x)$, scaling with the number of particles $N$. For $0< \beta < 1$, we obtain a norm-approximation of the evolution of an appropriate class of data on the Fock space. To this end, we need to correct the evolution of the condensate described by the one-particle nonlinear Schr\"odinger equation by means of a fluctuation dynamics, governed by a quadratic generator.
\end{abstract}

\section{Introduction}

In the last years important progress has been achieved in the mathematical understanding of the time-evolution of many body quantum systems. Here, we are going to consider the dynamics of systems of bosons, characterized by permutation symmetric wave functions. 

A bosonic system of $N$ particles moving in three space dimensions can be described on the Hilbert space 
\[ L^2_s (\bR^{3N}) = \left\{ \psi_N \in L^2 (\bR^{3N}) : \| \psi_N \|_2 = 1 \; \text{ and } \psi_N (x_{\pi(1)}, \dots , x_{\pi(N)}) = \psi_N (x_1, \dots , x_N) \right\} \]
The evolution of an initial $\psi_N \in L^2_s (\bR^{3N})$ is governed by the Schr\"odinger equation 
\begin{equation}\label{eq:schr0} i\partial_t \psi_{N,t} = H_N \psi_{N,t} 
\end{equation}
where, on the r.h.s., $H_N$ is the Hamilton operator of the system. Restricting our attention to two-body interactions, the Hamilton operator takes the form  
\begin{equation}\label{eq:ham0} H_N = \sum_{j=1}^N (-\Delta_{x_j} + V_\text{ext} (x_j)) + \lambda \sum^N_{i<j} V (x_i - x_j) 
\end{equation}
where $V_\text{ext}$ is an external potential, $V$ is an interaction and $\lambda \in \bR$ is a coupling constant. Notice that the unique global solution of 
(\ref{eq:schr0}) is given by $\psi_{N,t} = \exp (-i H_N t) \psi_{N,0}$.

For systems of interest in physics and chemistry, the number of particles involved in the dynamics is huge, ranging from $N \simeq 10^3$ up to values of the order $N \simeq 10^{23}$. For this reason, despite the fact that (\ref{eq:schr0}) is a linear equation, it is typically very difficult to extract useful information. One of the main goals of non-equilibrium statistical mechanics is therefore the derivation of effective equations approximating the solution of (\ref{eq:schr0}) in the interesting regimes.  

The simplest non-trivial limit, in which it is possible to obtain an effective approximation of (\ref{eq:schr0}) is the mean field regime, where $N \gg 1$, $|\lambda| \ll 1$ and $N \lambda$ remains fix, of order one. These conditions guarantee that particles interact through a large number of weak collisions, whose total effect is comparable with their inertia. To investigate the time-evolution in the mean field regime, we set $\lambda = 1/N$ in (\ref{eq:ham0}). We obtain the Hamiltonian
\begin{equation}\label{eq:ham0-mf} H_N^\text{mf} = \sum_{j=1}^N (-\Delta_{x_j} + V_\text{ext} (x_j)) + \frac{1}{N} \sum_{i<j}^N V(x_i -x_j) \end{equation}
and we study the corresponding evolution $\psi_{N,t} = \exp (-i H_N^\text{mf} t) \psi_N$. Let us assume that the initial data is approximately factorized, i.e. that $\psi_{N,0} \simeq \ph^{\otimes N}$ for a $\ph \in L^2 (\bR^3)$. Because of the mean-field nature of the interaction, we can expect that, for large $N$, factorization is approximately preserved by the time-evolution. In other words, we can expect that $\psi_{N,t} \simeq \ph_t^{\otimes N}$ for a new $\ph_t \in L^2 (\bR^3)$. Under this assumption, it is simple to show that $\ph_t$ must satisfy the self-consistent Hartree equation
\begin{equation}\label{eq:hartree0} i\partial_t \ph_t = (-\Delta + V_\text{ext}) \ph_t + (V * |\ph_t|^2 ) \ph_t 
\end{equation}
where the cubic nonlinearity reflects the two-body interactions. 

To obtain a precise statement about the convergence of the many-body evolution towards the Hartree dynamics, we introduce the notion of reduced densities. The one-particle reduced density associated with the solution $\psi_{N,t}$ of the Schr\"odinger equation is defined as the non-negative trace-class operator on $L^2 (\bR^3)$, with the integral kernel 
\[ \gamma_{N,t}^{(1)} (x;y) = \int dx_2 \dots dx_N \, \psi_{N,t} (x , x_2, \dots , x_N) \overline{\psi}_{N,t} (y,x_2, \dots , x_N) \]
normalized so that $\tr \gamma_{N,t}^{(1)} = 1$ ($\gamma^{(1)}_{N,t}$ is obtained by taking the partial trace of the orthogonal projection $|\psi_{N,t} \rangle \langle \psi_{N,t}|$ over the degrees of freedom of the last $(N-1)$ particles). 
Similarly, we can also introduce the $k$-particle reduced density $\gamma_{N,t}^{(k)}$ associated with $\psi_{N,t}$, for every $k = 2, 3, \dots , N$. 

It turns out that the reduced densities $\gamma^{(k)}_{N,t}$ provide the appropriate language to discuss convergence of the many body evolution towards the Hartree dynamics in the mean field regime. In fact, under appropriate conditions on the interaction potential $V$ (including the case $V (x) = \pm |x|^{-1}$ of a Coulomb interaction), one can show that, for every family of initial data $\psi_{N,0} \in L^2 (\bR^{3N})$ with $\gamma^{(1)}_{N,0} \to |\ph \rangle \langle \ph|$ (approximate factorization at time $t=0$), we will have 
\begin{equation}\label{eq:conv-red} \gamma_{N,t}^{(1)} \to |\ph_t \rangle \langle \ph_t | 
\end{equation}
as $N \to \infty$, for all fixed $t \in \bR$. Here $\ph_t$ is the solution of the Hartree equation (\ref{eq:hartree0}), with the initial data $\ph_{t=0} = \ph$. In fact (\ref{eq:conv-red}) can be extended to get convergence of the $k$-particle reduced density, for any fixed $k \in \bN$. The first results in the direction of (\ref{eq:conv-red}) have been obtained in \cite{Hepp,GV,Spohn}. More recently, much work went into the proof of (\ref{eq:conv-red}) in the case of singular interaction potentials; see \cite{EY,BGM,ES, 
FK,RS,AGT,KSS,KP,AN,CLS,CP,CH,AFP}.


After identifying the limiting effective dynamics (the one governed by the Hartree equation (\ref{eq:hartree0}), in the mean field regime), it is natural to consider fluctuations around it. To study fluctuations, it is very useful to switch to a representation of the many particle system on the bosonic Fock space
\[ \cF = \bigoplus_{n \geq 0} L^2_s (\bR^{3n}, dx_1 \dots dx_n) \]
On $\cF$, we can describe states with a variable number of particles. The vector $\Psi = \{ \psi_0, \psi_1, \dots \} \in \cF$ describes a state having $n$ particles with probability $\| \psi_n \|_2^2$, for all $n\in \bN$. For $f \in L^2 (\bR^3)$, we let $a^* (f)$ and $a(f)$ denote the usual creation and annihilation operators acting on $\cF$. We also introduce operator-valued distributions $a_x^*, a_x$ creating and, respectively, annihilating a particle at $x$. They satisfy the canonical commutation relations 
\[ [a_x , a_y^* ] = \delta (x-y) , \qquad [a_x, a_y] = [a_x^*, a_y^*] = 0 \,. \]
In terms of these operator-valued distributions, we define the Hamilton operator on $\cF$ by 
\begin{equation}\label{eq:F-ham}\cH^\text{mf}_N = \int dx \, a_x^* (-\Delta_x + V_\text{ext} (x))  a_x + \frac{1}{2N} \int dx dy \, V(x-y) a_x^* a_y^* a_y a_x 
\end{equation}
Since $\cH^\text{mf}_N$ commutes with the number of particles operator 
\[ \cN = \int dx \, a_x^* a_x \]
the corresponding time evolution preserves the number of particles. In particular, if we choose an initial data of the form $\Psi = \{ 0, \dots , 0 , \psi_N, 0, \dots \}$ with exactly $N$ particles, its evolution will coincide precisely with the one generated by (\ref{eq:ham0-mf}).

The advantage of working in the Fock space, rather than in the $N$-particle space $L^2_s (\bR^{3N})$, is the freedom in the choice of the initial data. We are interested in the evolution of coherent initial data. For $f \in L^2 (\bR^3)$, the coherent state with orbital $f$ is given by $W(f) \Omega \in \cF$, where 
$\Omega = \{ 1, 0, 0, \dots \}$ is the vacuum state (with no particles) and 
\[ W(f) = \exp (a^* (f) - a(f)) \]
is the Weyl operator with orbital $f$. Simple computations show that
\begin{equation}\label{eq:WOm} W(f) \Omega = e^{-\| f \|^2_2/2} \left\{ 1, f , \frac{f^{\otimes 2}}{\sqrt{2!}}, \dots \right\} \end{equation}
and that the expected number of particles is given by
\begin{equation}\label{eq:WNW} \langle W(f) \Omega, \cN W(f) \Omega \rangle = \| f \|_2^2 \end{equation}

Motivated by (\ref{eq:WNW}), we study the many body evolution generated by (\ref{eq:F-ham}) for initial coherent states of the form $W(\sqrt{N} \ph) \Omega$, with $\ph \in L^2 (\bR^3)$ such that $\| \ph \|_2 = 1$ (this normalization guarantees that the expected number of particles in $W(\sqrt{N} \ph) \Omega$ is equal to $N$). Since factorization is believed to be approximately preserved by the mean field dynamics, we expect that the evolution of the coherent state $W(\sqrt{N} \ph) \Omega$ can be approximated by the evolved coherent state $W(\sqrt{N} \ph_t) \Omega$, where $\ph_t$ is the solution of the Hartree equation (\ref{eq:hartree0}). 

We define the fluctuation dynamics 
\begin{equation}\label{eq:UN-mf} \cU^\text{mf}_N (t) = W (\sqrt{N} \ph_t)^* e^{-i\cH_N t} W(\sqrt{N} \ph) \end{equation}
and we set $\xi_t = \cU^\text{mf}_N (t) \Omega$. Then, we have
\begin{equation}\label{eq:ident0} 
e^{-i\cH_N^\text{mf} t} W(\sqrt{N} \ph) \Omega = W (\sqrt{N} \ph_t) \xi_t 
\end{equation}
Hence, to prove that the full evolution of the initial coherent state (the l.h.s. of the last equation) is approximately coherent, we need to show that $\xi_t$ is close to the vacuum (in fact, it is enough to prove that the expected number of particles in $\xi_t$ is much smaller than $N$, since the evolved Weyl operator $W(\sqrt{N} \ph_t)$ creates a condensate with  approximately $N$ particles in the state $\ph_t$).
To this end, it is useful to observe that 
\[ i \partial_t \cU^\text{mf}_N (t) = \cL_N^\text{mf} (t) \cU^\text{mf}_N (t) \]
with the generator
\begin{equation}\label{eq:LNmf} 
\begin{split} \cL^\text{mf}_N (t) = \; & \int dx a_x^* (-\Delta_x + V_\text{ext} (x)) a_x + \int dx (V*|\ph_t|^2)(x) a_x^* a_x \\ &+ \int dx dy V(x-y) \ph_t (x) \overline{\ph}_t (y) a_x^* a_y \\ &+ \frac{1}{2} \int dx dy V(x-y) \left[ \ph_t (x) \ph_t (y) a_x^* a_y^* + \text{h.c.} \right] \\ &+ \frac{1}{\sqrt{N}} \int dx dy V(x-y) a_x^* \left[ \ph_t (y) a_y^* + \text{h.c.} \right] a_x \\ &+ \frac{1}{2N} \int dx dy V(x-y) a_x^* a_y^* a_y a_x 
\end{split} 
\end{equation}

Notice that the terms on the third and fourth line do not commute with the number of particles operator $\cN$. This implies that the fluctuation dynamics $\cU^\text{mf}_N (t)$ does not preserve the number of particles (this is of course no surprise; the number of excitations of the condensate is expected to increase during the dynamics). Nevertheless it turns out that the expectation of $\cN$ cannot increase too fast. In fact, using the expression (\ref{eq:LNmf}), one can show that 
\begin{equation}\label{eq:xiN} \langle \xi_t, \cN \xi_t \rangle = \langle \cU^\text{mf}_N (t) \Omega, \cN \cU^\text{mf}_N (t) \Omega\rangle  \leq C e^{K|t|} \end{equation}
uniformly in $N$. This estimate on the growth  of the expectation of $\cN$ in the state $\xi_t$ can be used to show that (\ref{eq:ident0}) remains close to a coherent state, in the sense of the reduced densities. 
If we denote by $\gamma^{(1)}_{N,t}$ the reduced density of the full evolution of $W(\sqrt{N} \ph) \Omega$ (more generally, of a state of the form $W(\sqrt{N} \ph) \xi$, with $\xi$ having only few particles), (\ref{eq:xiN}) allows us to show that  
\begin{equation}\label{eq:rate} \tr \, \left| \gamma^{(1)}_{N,t} - |\ph_t \rangle \langle \ph_t| \right| \leq \frac{C e^{K|t|}}{N} \end{equation}
and that similar estimates hold for the $k$-particles reduced density, for all $k \in \bN$. The study of the dynamics of coherent states has been initiated in \cite{Hepp,GV}. More recently, it has been further developed in \cite{RS,CLS}, leading to a proof of (\ref{eq:rate}).

A part from  bounds like (\ref{eq:rate}) on the rate of convergence of the reduced densities, this approach  also allows us to study the fluctuations around the Hartree dynamics, in the limit of large $N$. {F}rom (\ref{eq:LNmf}), we expect that, for $N \to \infty$, the fluctuation dynamics $\cU^\text{mf}_N (t)$ converges towards a limiting dynamics $\cU^\text{mf}_\infty (t)$, defined by 
\[ i\partial_t \cU^\text{mf}_\infty (t) = \cL^\text{mf}_\infty (t) \cU^\text{mf}_\infty (t) \]
with the generator
\begin{equation}\label{eq:Linfty-mf} \begin{split} \cL^\text{mf}_\infty (t) = \; & \int dx a_x^* (-\Delta_x + V_\text{ext} (x)) a_x + \int dx (V*|\ph_t|^2)(x) a_x^* a_x \\ &+ \int dx dy V(x-y) \ph_t (x) \overline{\ph}_t (y) a_x^* a_y \\ &+ \frac{1}{2} \int dx dy V(x-y) \left[ \ph_t (x) \ph_t (y) a_x^* a_y^* + \text{h.c.} \right]
\end{split} \end{equation}
independent of $N$. In fact, under appropriate assumptions on the interaction potential, one can indeed show that 
\begin{equation}\label{eq:conv-fluc} \| \cU^\text{mf}_N (t) \Omega - \cU^\text{mf}_\infty (t) \Omega \| \leq \frac{Ce^{K|t|}}{\sqrt{N}} \end{equation}
By definition, this implies that 
\begin{equation}\label{eq:fluc-lim} \| e^{-i\cH^\text{mf}_N t} W(\sqrt{N} \ph) \Omega - W (\sqrt{N} \ph_t) \cU^\text{mf}_\infty (t) \Omega \| \leq \frac{Ce^{K|t|}}{\sqrt{N}}  
\end{equation}
Hence, taking into account the limiting fluctuation dynamics $\cU^\text{mf}_\infty (t)$, we provide a norm approximation to the full many body evolution (hence, a stronger approximation compared with the one furnished by the evolved coherent state $W(\sqrt{N} \ph_t) \Omega$, which is only valid in the sense of the reduced densities). The convergence (\ref{eq:conv-fluc}) has already been observed in \cite{Hepp} for smooth interactions and then in \cite{GV} for a larger class of potentials. More recently, it has been established, in a slightly different form, in \cite{GMM1,GMM2,XC}.
Using (\ref{eq:conv-fluc}) and bounds like (\ref{eq:xiN}) on the growth of the expectation of the number of particles (and of its higher moments) with respect to the evolution $\cU_N^\text{mf}$, one can prove a central limit theorem for sums of one-particle observables evolved through the full interacting many-body dynamics; see \cite{BKS,BSS}. 

Instead of considering fluctuations around the Hartree dynamics for coherent initial states on the Fock space, it is possible to analyze directly the mean field evolution in the $N$-particle Hilbert-space $L^2_s (\bR^{3N})$, for approximately factorized initial data. To this end, it is convenient to introduce the time-dependent map
\begin{equation}\label{eq:uNt} u_{N,t} : L^2_s (\bR^{3N}) \to \cF_t \end{equation}
where $\cF_t$ denotes the bosonic Fock space, constructed over the orthogonal complement in $L^2 (\bR^3)$ of the one-dimensional space spanned by the solution $\ph_t$ of (\ref{eq:hartree0}) and $u_{N,t} \psi = \{ \psi^{(0)}, \psi^{(1)}, \dots , \psi^{(N)}, 0 ,0 , \dots \}$ if 
\[ \psi = \psi^{(0)} \ph_t^{\otimes N} + \psi^{(1)} \otimes_s \ph^{\otimes (N-1)} + \dots + \psi^{(N-1)} \otimes_s \ph + \psi^{(N)} \]
where $\otimes_s$ denotes the symmetric tensor product. When applied to the many body evolution $\psi_{N,t} = e^{-i H_N^\text{mf} t} \psi_{N,0}$ of an approximately factorized initial data $\psi_{N,0}$, the isometric map $u_{N,t}$ eliminates the particles in the condensate $\ph_t$ and let us focus on the fluctuations. It has been shown in \cite{LNS}, inspired by ideas developed in the time-independent setting in \cite{LNSS}, that there exists a Fock space unitary evolution $\wt{\cU}_\infty^\text{mf} (t;s)$ with a quadratic generator $\wt{\cL}^\text{mf}_\infty (t)$ such that 
\begin{equation}\label{eq:fluc-fixN} \left\| u_{N,t} e^{-i H_N t} \psi_N - \wt{\cU}_\infty^\text{mf} (t;0) u_{N,0} \psi_N \right\| \leq C N^{-1/2} e^{K|t|}
\end{equation} 
Notice that $\wt{\cL}^\text{mf}_\infty$ is similar but does not coincide with the limiting generator (\ref{eq:Linfty-mf}) (the difference between the two generators is due to the requirement, in the definition of $u_{N,t}$, that fluctuations are orthogonal to $\ph_t$).

A more subtle and physically interesting regime, in which it is possible to approximate the many-body evolution by an effective dynamics, is the Gross-Pitaevskii regime. In the Fock space representation, the Hamilton operator is given by 
\begin{equation}\label{eq:ham-GP} \cH^\text{GP}_N = \int dx \, a_x^* (-\Delta_x + V_\text{ext} (x)) a_x + \frac{1}{2} \int dx dy \, N^2 V(N (x-y)) a_x^* a_y^* a_y a_x 
\end{equation}
where $V \geq 0$ is a smooth, short range potential. It turns out that, in this case, the many-body 
Schr\"odinger evolution can be approximated by the time-dependent Gross-Pitaevskii equation 
\begin{equation}\label{eq:GP0}
i\partial_t \ph_{\text{GP},t} = (-\Delta + V_\text{ext}) \ph_{\text{GP},t} + 8 \pi a_0 \, |\ph_{\text{GP},t}|^2 \ph_{\text{GP},t} 
\end{equation}
where $a_0$ is the scattering length of the potential $V$. Recall that the scattering length is defined through the solution $f$ of the zero-energy scattering equation
\begin{equation}\label{eq:scatt00} \left[ -\Delta + \frac{1}{2} V \right] f = 0 \end{equation}
with the boundary condition $f (x) \to 1$, for $|x| \to \infty$. For $x$ outside the support of $V$, we have 
\[ f(x) = 1 - \frac{a_0}{|x|} \]
where the constant $a_0$ is defined to be the scattering length of $V$. Equivalently, we can define $a_0$ through the integral 
\begin{equation}\label{eq:scatt0} 8 \pi a_0 = \int V(x) f(x) dx \end{equation}

{F}rom the point of view of physics, the Gross-Pitaevskii regime is very different from the mean field limit since here particles interact rarely (only when they are very close, at distances of the order $N^{-1}$) but when they do, the collisions are very strong. As a consequence of the strong interactions among the particles, the many body wave function develops a short range correlation structure which can be described by the solution $f$ of the zero-energy scattering equation (\ref{eq:scatt00}) and is responsible for the emergence of the scattering length in (\ref{eq:GP0}). 

A first derivation of (\ref{eq:GP0}) starting from many-body quantum mechanics has been given in \cite{ESY1,ESY2,ESY3}. Later, an alternative approach has been proposed in \cite{P}. Bounds on the rate of the convergence towards the Gross-Pitaevskii dynamics have been then obtained in \cite{BdOS12}, making use of an appropriately modified version of the coherent states method illustrated above. The main problem one has to face to apply the coherent states approach to the Gross-Pitaevskii regime is the formation of correlations among the particles, which cannot be described by coherent states (the development of a short scale correlation structure in the Gross-Pitaevskii regime has been studied in \cite{EMS,CH2}).
To circumvent this problem, one has to modify the coherent states through  appropriate Bogoliubov transformations. Following \cite{BdOS12}, we define the function 
\[ k^\text{GP}_t (x;y) = - N (1-f (N (x-y))) \ph^{N}_{\text{GP},t} (x) \ph^{N}_{\text{GP},t} (y) \]
where $f$ is the solution of (\ref{eq:scatt00}) 
and where $\ph^{N}_{\text{GP},t}$ is the solution of a modified, $N$-dependent, version of (\ref{eq:GP0}), with the local nonlinearity replaced by an Hartree term, given by convolution with $N^3 V(N.)f(N.)$ (since by (\ref{eq:scatt0}), $N^3 V(N.) f(N.) \to 8\pi a_0 \delta$, it is easy to bound the difference between $\ph^{N}_{\text{GP},t}$ and $\ph_{\text{GP},t}$). Using $k^\text{GP}_t$, we construct the unitary operator 
\[ T_{\text{GP}, t} = \exp \left[ \frac{1}{2}\int dx dy \left( k^\text{GP}_t (x;y) a_x^* a_y^* - \text{h.c.} \right) \right] \]
on $\cF$. $T_{\text{GP},t}$ is a Bogoliubov transformation; it acts on creation and annihilation operators by
\[ \begin{split} T_{\text{GP},t}^* a(f) T_{\text{GP} ,t} = a (\cosh_{k^\text{GP}_t} f) + a^* (\sinh_{k^\text{GP}_t} \overline{f}) \\
T_{\text{GP}, t}^* a^* (f) T_{\text{GP} , t} = a^* (\cosh_{k^\text{GP}_t} f) + a (\sinh_{k^\text{GP}_t} \overline{f}) 
\end{split} 
\]
The operator $T_{\text{GP} , t}$ can be used to implement the short scale correlation structure characterizing the solution of the many-body 
Schr\"odinger equation in the Gross-Pitaevskii limit. 

We consider the initial data  $W(\sqrt{N} \ph) T_{\text{GP},0} \, \Omega$ (more generally, we can consider states of the form $W(\sqrt{N} \ph) T_{\text{GP} ,0} \xi$, with $\xi$  containing only a bounded number of particles). We expect that the many-body evolution of such an initial data still has the same form. To confirm this fact, 
we define the fluctuation vector $\xi_t$ by requiring that 
\begin{equation}\label{eq:xit-GP} 
e^{-i\cH_N t} W(\sqrt{N} \ph) T_{\text{GP} ,0} \Omega = W (\sqrt{N} \phn_{\text{GP},t}) T_{\text{GP}, t} \xi_t \end{equation}

{F}rom (\ref{eq:xit-GP}), to show convergence of the many-body dynamics towards (\ref{eq:GP0}) it is enough to prove that the fluctuation vector $\xi_t$ defined by (\ref{eq:xit-GP}) remains close to the vacuum. Since $\xi_t = \cU^\text{GP}_N (t) \Omega$, with  
\[ \cU^\text{GP}_N (t) = T_{\text{GP},t}^* W^* (\sqrt{N} \phn_{\text{GP},t}) e^{-i\cH_N t} W(\sqrt{N} \ph) T_{\text{GP},0}  \]
the problem reduces to show a bound for the growth of the number of particles with respect to the fluctuation dynamics $\cU^\text{GP}_N$. 

Such a bound has been established in \cite{BdOS12}, making use of certain cancellations in the generator of $\cU^\text{GP}_N$ produced by the introduction of the Bogoliubov transformation $T_{\text{GP},t}$. As a consequence, it was proven in \cite{BdOS12} that the reduced density $\gamma^{(1)}_{N,t}$ of the full evolution of the initial data $W(\sqrt{N} \ph) T_{\text{GP},0} \Omega$ (or, more generally, of initial data having the form $W(\sqrt{N} \ph) T_{\text{GP},0} \xi$, for $\xi \in \cF$ with a bounded expectation for $\cN$ and $\cH_N$) satisfies the bound
\begin{equation}\label{eq:rate-GP} \tr \, \left| \gamma^{(1)}_{N,t} - |\ph_{t} \rangle \langle \ph_{t} | \right| \leq C N^{-1/2} \exp (c_1 \exp (c_2 |t|))
\end{equation}
where $\ph_t$ is the solution of the Gross-Pitaevskii equation (\ref{eq:GP0}). 

In the mean field regime that we discussed above, the coherent states approach could also be used to describe 
fluctuations around the limiting equation. In particular, it allowed us to identify a limiting fluctuation dynamics with a quadratic generator 
and to apply it to obtain a norm bound of the form (\ref{eq:fluc-lim}) for the many-body evolution. After establishing the estimate (\ref{eq:rate-GP}) for the rate of convergence of the one-particle reduced density, it is therefore natural to ask whether we can use the same approach to describe fluctuations around the Gross-Pitaevskii equation in the limit of large $N$.  Unfortunately, it turns out that the in the Gross-Pitaevskii regime, one cannot approximate the fluctuation dynamics $\cU_N^\text{GP}$ by a quadratic evolution in norm. Although one can control its effect on the growth of the number of particles (needed to prove (\ref{eq:rate-GP})), the cubic and quartic components of the generator of $\cU_N^\text{GP}$ (cubic and quartic in the creation and annihilation operators) are not negligible in the limit of large $N$. 

Instead of considering fluctuations of the time-evolution around the time-dependent Gross-Pitaevskii equation, it is also possible to approach this problem from a static, time-independent, point of view. To this end, one can trap the system in a finite volume (either by imposing boundary conditions or by turning on an external potential) and one can study the difference between the many-body ground state energy and the minimum of the Gross-Pitaevskii energy functional or, more generally, the energy of low lying excitations (the fact that Gross-Pitaevskii theory describes, in leading order, the ground state properties of the many body system has been established in \cite{LSY,LS}).  
In the mean field setting, this program has been carried out in \cite{S,GS,LNSS,DN}, where it was proven that the excitation spectrum is determined by a quadratic Hamiltonian similar to (\ref{eq:Linfty-mf}). This suggests that, in the mean field regime, a good approximation for the many body ground state has the form $W(\sqrt{N} \ph) T \Omega$, where $\ph$ minimizes the Hartree energy functional and where $T$ is a Bogoliubov transformation (the exponential of a quadratic expression in creation and annihilation operators), needed to diagonalize the quadratic Hamiltonian. Similarly, good approximation for low-lying excited states have the form $W(\sqrt{N_0} \ph) T a^* (g_1) \dots a^* (g_k) \Omega$, for appropriate $k \in \bN$, $N_0 = N - k$ and orbital $g_1, \dots , g_k$ orthogonal to $\ph$ (in fact, to produce states with a fixed number of particles, it is better to work with a map $u_{N}$, defined similarly to (\ref{eq:uNt}), rather than with the Weyl operator $W (\sqrt{N} \ph)$; see \cite{LNSS} for details). 

Although the excitation spectrum in the Gross-Pitaevskii regime should still be close to the 
spectrum of a quadratic Hamiltonian, a good approximation for the ground state cannot have 
the form $W(\sqrt{N} \ph) T \Omega$ for a Bogoliubov transformation $T$ (analogously, excited states cannot be approximated by vectors like $W(\sqrt{N} \ph) T a^* (g_1) \dots a^* (g_k) \Omega$). This follows from \cite{ESY9}, where it has been shown (in fact, in a more general setting) that the minimum of the energy over all states of the form $W(\sqrt{N} \ph) T \Omega$ (with $T$ being the exponential of a quadratic expression) remains strictly above the true ground state energy, with an error of order one (an upper bound to the correct ground state energy, up to an error that, in the Gross-Pitaevskii regime, vanishes in the limit of large $N$, has been obtained in \cite{YY}; this result is consistent with the Lee-Huang-Yang prediction). 

In this paper, we are going to consider an intermediate regime, lying between the mean field and the Gross-Pitaevskii limits. For $0 < \beta < 1$, we define the Hamilton operator 
\begin{equation}\label{eq:Hbeta} \cH_N = \int dx \nabla_x a_x^* \nabla_x a_x + \frac{1}{2N} \int dx dy \, N^{3\beta} V (N^\beta (x-y)) a_x^* a_y^* a_y a_x \end{equation}
acting on the bosonic Fock space $\cF$. To simplify a bit the computations, we neglect here the external potential (but it would be easy to modify our analysis to include one). The Hamiltonian (\ref{eq:Hbeta}) can be thought of as an interpolation between the mean field Hamiltonian (\ref{eq:F-ham}), obtained with $\beta =0$, and the Gross-Pitaevskii Hamiltonian (\ref{eq:ham-GP}), recovered with $\beta = 1$. 

For $0< \beta < 1$, the many body evolution develops weaker correlations, compared with the Gross-Pitaevskii regime. As a consequence, on the level of the reduced densities, the many body evolution generated by (\ref{eq:Hbeta}) can be approximated, in the limit $N \to \infty$ and for all $0 < \beta < 1$, by the nonlinear Schr\"odinger equation
\begin{equation}\label{eq:NLS0} i\partial_t \ph_t = -\Delta \ph_t + b_0 |\ph_t|^2 \ph_t 
\end{equation}
with $b_0 = \int V(x) dx$ (notice that $8\pi a_0 \leq b_0$ for all short range and repulsive $V$).

While (\ref{eq:NLS0}) is enough if we are interested in the limiting behavior of the reduced densities, to study fluctuations and to obtain a norm approximation for the many body evolution we need a more precise ansatz, taking into account the (weak) two-body correlations. Instead of working with the solution of the zero-energy scattering equation (\ref{eq:scatt00}), we find more convenient here to fix $\ell > 0$ and to consider the ground state $f_{N,\ell}$ of the Neumann problem associated with the potential $N^{-1+3\beta} V(N^\beta .)$ on the sphere of radius $\ell$ centered at the origin. In other words, we choose $f_{N,\ell}$ as the solution of the eigenvalue problem (\ref{eq:Nf}) associated with the smallest possible eigenvalue $\l_{N,\ell}$, normalized so that $f_{N,\ell} (x) = 1$ for $|x| = \ell$ and continued to $\bR^3$ by requiring that $f_{N,\ell} (x) = 1$ for all $|x| \geq \ell$.  
We use $f_{N,\ell}$ to describe correlations among particles in the condensate. Accordingly, we consider the $N$-dependent Hartree equation 
\begin{equation}\label{eq:NLSN} 
i\partial \phn_t = -\Delta \phn_t + (N^{3\beta} V(N^\beta .) f_{N,\ell} *|\phn_t|^2 ) \phn_t \, .
\end{equation}
As $N \to \infty$ and for all $0 < \beta < 1$, $\phn_t$ approaches the solution of the nonlinear equation (\ref{eq:NLS0}).  We will see, however, that (\ref{eq:NLSN}) furnishes a better approximation for the dynamics of the condensate wave function, because, through the factor $f_{N,\ell}$, it takes into account the correlations among the particles (which, despite being weak, are not negligible in the analysis of the fluctuations).

Our goal is to study the fluctuations around  (\ref{eq:NLSN}), to prove that their dynamics has a quadratic generator in the limit of large $N$, and to use the limiting fluctuation dynamics (with the quadratic generator) to obtain a norm approximation of the many body evolution generated by (\ref{eq:Hbeta}). 


First of all, we need to take care of the correlation structure. We proceed similarly as in \cite{BdOS12}, introducing a family of Bogoliubov transformations. Using the Neumann ground state $f_{N,\ell}$, we set $\o_{N,\ell} = 1 - f_{N,\ell}$ (so that $\o_{N,\ell} (x) = 0$ for all $|x| \geq \ell$) and we define
\begin{equation}\label{eq:kNell0} k_{N,t} (x;y) = -N \o_{N,\ell} (x-y) (\phn_t ((x+y)/2))^2  
\end{equation}
where $\phn_t$ is the solution of (\ref{eq:NLSN}). It turns out (see Lemma \ref{lm:propomega}) that $\| k_{N,t} \|_2$ is bounded, uniformly in $N$, and therefore (\ref{eq:kNell0}) is the integral kernel of a Hilbert-Schmidt operator that we denote again by $k_{N,t}$. Using $k_{N,t}$, we define the Bogoliubov transformation
\begin{equation}\label{eq:TN} T_{N,t}  = \exp \left[ \frac{1}{2}\int dx dy \, k_{N,t} (x;y) a_x^* a_y^* - \text{h.c.} \right] \end{equation}
We consider initial data of the form $W(\sqrt{N} \ph) T_{N,0} \xi_N$, with a $\xi_N$ ``close'' to the vacuum $\Omega$ (in the sense that the expectation and the variance of the number of particles and of the kinetic energy operator in the state $\xi_N$ can be bounded uniformly in $N$). We write the evolution of such initial data as 
\[ e^{-i\cH_N t} W(\sqrt{N} \ph) T_{N,0} \xi_N = W(\sqrt{N} \phn_t) T_{N,t} \xi_{N,t} \]
where 
\[ \xi_{N,t} = \cU_N (t;0) \xi_N \]
with the fluctuation dynamics
\[ \cU_N (t;0) = T_{N,t}^* W^* (\sqrt{N} \phn_t) e^{-it \cH_N} W (\sqrt{N} \ph) T_{N,0} \]
Our goal is to approximate $\cU_N (t;0)$ by an evolution with a quadratic generator. 

To this end, we define the phase 
\begin{equation}\begin{split}\label{eq:etaN}
\eta_N (t) =\;& N \int dx dy N^{3\beta} V(N^\beta(x-y)) (1/2 - f_{N,\ell} (x-y)) |\phn_t (x)|^2 |\phn_t (y)|^2\\
&+\int dx dy |\nabla_x \sinh_{k_{N,t}} (x,y)|^2 + \int dx (N^{3\beta} V(N^\beta .) *|\phn_t|^2) (x) \langle s^N_x, s^N_x \rangle \\ &+ \int dx dy N^{3\beta} V(N^\beta (x-y)) \phn_t (x) \bar{\ph}^{N}_t (y) \langle s^N_x,s^N_y \rangle \\
 &+\text{Re } \int dx dy \, N^{3\beta} V(N^\beta (x-y)) \phn_t (x) \phn_t (y) \langle s^N_x , c^N_y \rangle\\
 &+\frac{1}{2N} \int dx dy \, N^{3\beta} V (N^\beta (x-y)) \left[ |\langle s^N_x , c^N_y \rangle|^2  +  |\langle s^N_x , s^N_y \rangle|^2   + \langle s^N_y , s^N_y \rangle \langle s^N_x ,s^N_x \rangle \right] \end{split}
\end{equation}
and the time-dependent and $N$-dependent quadratic generator
\begin{equation}\label{eq:L2N0} 
\begin{split}
\cL_{2,N} (t) = \; & (i\partial_tT_{N,t}^*)T_{N,t} + \cL^{(K)}_{2,N} (t) + \cL^{(V)}_{2,N} (t) \\
&+\frac{N}{2} \int dx dy \, \omega_{N,\ell} (x-y) \\ &\hspace{.3cm} \times \left[( \phn_t ((x+y)/2) \Delta \phn_t ((x+y)/2) + |\nabla \phn_t ((x+y)/2)|^2)  a_x^* a_y^*+\text{h.c.}\right] \\
&+ N \lambda_{\ell, N} \int dx dy \, {\bf 1} (|x-y| \leq \ell) \left[ (\phn_t ((x+y)/2))^2 a_x^* a_y^* + \text{h.c.} \right]  
\end{split}
\end{equation}
with 
\begin{equation}\label{eq:LV} \begin{split} 
\cL_{2,N}^{(V)} (t) = \; & \int dx \,  (N^{3\beta} V(N^\beta .) *|\phn_t|^2) (x) \\ &\hspace{1cm} \times \left[ a^* (c_x^{N}) a (c_x^{N}) + a^* (s_x^N) a (s_x^N) + a^* 
(c_x^{N}) a^* (s_x^N) + a (s_x^N) a(c_x^{N}) \right] \\
&+ \int dx dy N^{3\beta} V(N^\beta (x-y)) \phn_t (x) 
\bar{\ph}^N_t (y) \\ &\hspace{1cm} \times \left[ a^* (c_x^{N}) a(c_y^N) + a^* (s_y^N) a 
(s_x^N) + a^* (c_x^{N}) a^* 
(s_y^N) + a (s_x^N) a(c_y^N) \right] \\
&+ \frac{1}{2} \int dx dy N^{3\beta} V(N^\beta (x-y)) 
\phn_t (x) \phn_t (y) \\ &\hspace{1cm} \times \left[ a^* (c_x^{N}) a(s_y^N) + a^* (c_y^N) a(s_x^N) + a(s_x^N) a (s_y^N) \right] \\
&+ \frac{1}{2} \int dx dy N^{3\beta} V(N^\beta (x-y)) 
\bar{\ph}^N_t (x) \bar{\ph}^N_t (y) \\ &\hspace{1cm} \times \left[ a^* (s_y^N) a(c_x^{N})+ a^* (s_x^N) a (c_y^N) + a^* (s_x^N) a^* (s_y^N) \right] \\
&+ \frac{1}{2} \int dx dy N^{3\beta} V(N^\beta (x-y))
\phn_t (x) \phn_t (y) \left[ a^* (p_x^N) a_y^* 
+ a^* (c_x^{N}) a^* (p_y^N) \right] \\ 
&+ \frac{1}{2} \int dx dy N^{3\beta} V(N^\beta (x-y)) 
\bar{\ph}^N_t (x) \bar{\ph}^N_t (y) \left[ a (p_x^N) a_y + a (c_x^{N}) a (p_y^N) \right] 
\end{split} \end{equation}
and
\begin{equation}\label{eq:LK} \begin{split} 
\cL_{2,N}^{(K)} (t) = \; &\int dx \nabla_x a_x^* \nabla_x a_x + \int dx \, \Big[ a^*_x  a(-\Delta_x p_x^N) + a^* (-\Delta_x p_x^N) a_x +  a^* (\nabla_x p_x^N) a(\nabla_x p_x^N)  \\
&\hspace{1cm}  + \nabla_x a^* (k_x) \nabla_x a(k_x) 
+ a^* (-\Delta_x r_x^N) a(k_x)  + a^* (s_x^N) a(-\Delta_x r_x^N) \\ &\hspace{1cm}   + a^* (-\Delta_x p_x^N) a^* 
(k_x)  + a(k_x) a (-\Delta p_x^N) + a^*_x a^* (-\Delta_x r_x^N)  \\ &\hspace{1cm} + a (-\Delta_x r_x^N) a_x + a^* (p_x^N) a^* (-\Delta_x r_x^N) + a(-\Delta_x r_x^N) a(p_x^N) 
\Big] 
\end{split} \end{equation}
Here we used the shorthand notation $c^N_x (y) = \cosh_{k_{N,t}} (y,x)$, $s^N_x (y) = \sinh_{k_{N,t}} (y,x)$, $p^N_x (y) = (\cosh_{k_{N,t}} - 1) (y,x)$ and  $r^N_x (y) = (\sinh_{k_{N,t}} - k_{N,t}) (y,x)$. 

We denote by $\cU_{2,N} (t;s)$ the time evolution generated by $\cL_{2,N}$. In other words, $\cU_{2,N} (t;s)$ is the two-parameter group of unitary operators solving the Schr\"odinger equation 
\begin{equation}\label{eq:U2N-def} i\partial_t \cU_{2,N} (t;s) = \cL_{2,N} (t) \cU_{2,N} (t;s) \end{equation}
with the initial condition $\cU_{2,N} (s;s) =1$ for all $s \in \bR$. The existence of such an evolution $\cU_{2,N}$ can be established as in \cite{GV} (notice here that all kernels entering the generator $\cL_{2,N} (t)$ are smooth for every finite $N$; they only develop singularities as $N \to \infty$). 

In the next theorem, our first main result, we show that, when acting on appropriate vectors $\xi_N$, $\cU_N (t;0)$ can be approximated by $\cU_{2,N} (t;0)$, up to the (physically uninteresting) phase $\eta_N$. 
\begin{theorem}\label{thm:main1}
Let $V \geq 0$ be smooth, spherical symmetric and compactly supported. Fix $0 < \beta < 1$ and let $\alpha = \min (\beta /2 , (1-\beta)/2)$. Let $\ph \in H^4 (\bR^3)$. Fix $\ell > 0$. Consider a sequence $\xi_N \in \cF$ such that $\| \xi_N \| =1$ and  
\begin{equation}\label{eq:ass-xi} \langle \xi_N , \left[\cN^2 + \cK^2 + \cV_N \right] \xi_N \rangle \leq C \end{equation}
uniformly in $N$. Here we introduced the notation 
\begin{equation}\label{eq:cVN} \cK = \int dx \nabla_x a_x^* \nabla_x a_x, \quad \quad \cV_N = \frac{1}{2N} \int dx dy N^{3\beta} V(N^\beta (x-y)) a_x^* a_y^* a_y a_x  \end{equation}
for the kinetic and the potential energy operators entering in $\cH_N = \cK + \cV_N$. Then there are $C,c_1, c_2 > 0$ such that 
\[ \begin{split} 
&\left\| e^{-i \cH_N t} W(\sqrt{N} \ph) \, T_{N,0} \, \xi_N - e^{-i \int_0^t  \eta_N (s) ds} \, W (\sqrt{N} \phn_t) T_{N,t} \, \cU_{2,N} (t) \xi_N \right\|^2 \\ &\hspace{8cm} \leq C N^{-\alpha} \exp (c_1 \exp (c_2 |t|)) \end{split} \]
for all $t \in \bR$ and all $N$ large enough. 
\end{theorem}

Notice that a result similar to Theorem \ref{thm:main1} has been recently obtained in \cite{GM}, however only for $0 < \beta < 1/3$. The main difference with respect to \cite{GM} is the fact that here we already introduce the Bogoliubov transform $T_{N,0}$ at time $t=0$; this allows us to cover all $\beta < 1$. 

Theorem \ref{thm:main1} shows that, in the limit of large $N$, the dynamics of the fluctuations around the nonlinear Schr\"odinger equation (\ref{eq:NLSN}) can be approximated by the evolution $\cU_{2,N} (t;s)$ having the quadratic generator (\ref{eq:L2N0}). It is also possible to approximate the dynamics of the fluctuations by a limiting evolution, again with a quadratic generator, but now independent of $N$. To this end, we start by noticing that, in the limit of large $N$, the integral kernel (\ref{eq:kNell0}) approaches the limit
\begin{equation}\label{eq:kell} k_t (x;y) = - \o_\ell^\text{asymp} (x-y) \ph_t^2 ((x+y)/2) \end{equation}
where $\ph_t$ is the solution of the $N$-independent nonlinear Schr\"odinger equation (\ref{eq:NLS0}) while 
\begin{equation}\label{eq:oasym0} \o_\ell^\text{asymp} (x) = \left\{ \begin{array}{ll} 
\frac{b_0}{8\pi} \left[ \frac{1}{|x|} - \frac{3}{2\ell} + \frac{x^2}{2\ell^3} \right] \quad & \text{for } |x| \leq \ell \\
0 & \text{otherwise} \end{array} \right. 
\end{equation}
is the pointwise limit of $N \o_{N,\ell} (x)$ for $N \to \infty$. Like $k_{N,t}$, also $k_t \in L^2 (\bR^3 \times \bR^3)$ is the integral kernel of a Hilbert-Schmidt operator on $L^2 (\bR^3)$. It defines the Bogoliubov transformation
\begin{equation}\label{eq:Tt} T_t = \exp \left[\frac{1}{2} \int dx dy \, k_t (x;y) a_x^* a_y^* - \text{h.c.} \right] \end{equation}

With this notation, we can formally take the limit $N \to \infty$ in (\ref{eq:L2N0}) and define the time-dependent quadratic generator 
\begin{equation}\label{eq:L2inf0} \begin{split} \cL_{2,\infty} (t) = \; & (i\partial_t T_{t}^*) T_{t} + \cL_{2}^{(K)} (t) + \cL^{(V)}_2 (t) \\ &+ \frac{1}{2} \int dx dy \, \omega_{\ell}^\text{asymp} (x-y) \\ &\hspace{1cm} \times \left\{ \left[ \ph_t ((x+y)/2) \Delta \ph_t ((x+y)/2) + |\nabla \ph_t ((x+y)/2)|^2 \right] a_x^* a_y^* + \text{h.c.} \right\} \\ 
&+ \frac{3b_0}{8\pi \ell^3} \int dx dy {\bf 1} (|x-y| \leq \ell) \left[ \ph_t^2 ((x+y)/2) a_x^* a_y^* + \text{h.c.} \right] \end{split} 
\end{equation}
where $\cL^{(K)}_2$ and $\cL^{(V)}_2$ are defined as $\cL_{2,N}^{(K)}$ and $\cL^{(V)}_{2,N}$, but with $\phn_t$ and $k_{N,t}$ replaced by $\ph_t$ and $k_t$ (and thus with $c^N_x, s^N_x, p_x^N, r^N_x$ replaced by $c_x, s_x, p_x, r_x$, defined analogously). 

We denote by $\cU_{2,\infty} (t;s)$ the two-parameter unitary group generated by (\ref{eq:L2inf0}). In other words, we define $\cU_{2,\infty}$ through the 
Schr\"odinger equation 
\[ i\partial_t \cU_{2,\infty} (t;s) = \cL_{2,\infty} (t) \cU_{2,\infty} (t;s) \]
with the initial condition $\cU_{2,\infty} (s;s) = 1$. Also here, the existence of such an evolution $\cU_{2,\infty}$ is guaranteed by the work \cite{GV} (the only singularity entering $\cL_{2,\infty}$ is the Coulomb-type singularity of $\o_\ell^\text{asymp}$). 
\begin{theorem}\label{thm:main2}
Let $V \geq 0$ be smooth, spherical symmetric and compactly supported. Fix $0 < \beta < 1$ and set $\alpha = \min (\beta/2 , (1-\beta)/2)$. Let $\ph \in H^4 (\bR^3)$. Fix $\ell > 0$. Consider a sequence $\xi_N \in \cF$ such that  $\| \xi_N \| =1$ and 
\[ \langle \xi_N , \left[\cN^2 + \cK^2 + \cV_N \right] \xi_N \rangle \leq C \]
uniformly in $N$. Then there exist $C,c_1, c_2 > 0$ such that 
\[ \begin{split} 
&\left\| e^{-i \cH_N t} W(\sqrt{N} \ph) T_{N,0} \xi_N - e^{-i \int_0^t \eta_N (s) ds} \, W (\sqrt{N} \phn_t) T_{N,t} \, \cU_{2,\infty} (t) \xi_N \right\|^2 \\ &\hspace{8cm} \leq C N^{-\alpha} \exp (c_1 \exp (c_2 |t|)) \end{split} \]
for all $t \in \bR$ and all $N$ large enough. 
\end{theorem}
 
Theorem \ref{thm:main2} shows that, for all $0< \beta < 1$, the dynamics of the fluctuations around the nonlinear Schr\"odinger evolution (\ref{eq:NLSN}) is governed by the quadratic generator (\ref{eq:L2inf0}).

%
%
%
%
%
%
%
%
%

\bigskip 
 
{\it Acknowledgements.} The authors would like to acknowledge support by the Swiss National Science Foundation through the SNF Project ``Effective equations from quantum dynamics''. 
S. Cenatiempo acknowledges the support of MIUR through the FIR grant 2013 ``Condensed Matter in Mathematical Physics'' (code RBFR13WAET). C. Boccato and B. Schlein also gratefully acknowledge support by the CRC-1060 ``The Mathematics of emergent effects''.

\section{The Fluctuation Dynamics}

We work on the bosonic Fock space
\[ \cF = \mathbb{C} \oplus \bigoplus_{n \geq 1} L^2_s (\bR^{3n}) \]
where $L^2_s (\bR^{3n})$ denotes the space of all $\psi_n \in L^2 (\bR^{3n})$ such that 
\[ \psi_n (x_{\pi 1} , \dots, x_{\pi n}) = \psi_n (x_1, \dots , x_n) \]
for all permutations $\pi \in S_n$. We use the notation 
$\Omega = \{ 1, 0, 0 \dots \}$ for the vacuum vector in $\cF$. 

For $g \in L^2 (\bR^3)$, we introduce the creation operator $a^* (g)$ and its adjoint, the annihilation operator $a(g)$, by 
\[ \begin{split} (a^* (g) \Psi)^{(n)} (x_1, \dots , x_n) &= \frac{1}{\sqrt{n}} \sum_{j=1}^n g (x_j) \psi^{(n-1)} (x_1, \dots , x_{j-1}, x_{j+1}, \dots , x_n) \\
(a(g) \Psi)^{(n)} (x_1, \dots , x_n) &= \sqrt{n+1} \int dx \, \overline{g (x)} \, \psi^{(n+1)} (x, x_1, \dots , x_n) \end{split} \]
for all $\Psi = \{ \psi^{(n)} \}_{n\in \bN} \in \cF$. Creation and annihilation satisfy canonical commutation relations 
\[ [ a(f), a^* (g) ] = \langle f,g \rangle , \quad [ a(f) , a(g)] = [ a^* (f), a^* (g) ] = 0 \]

We will also make use of the operator-valued distributions $a_x^*, a_x$, defined so that 
\[ a^* (g) = \int dx g(x) a_x^*, \quad a(g) = \int dx \overline{g} (x) a_x \]

Although creation and annihilation operators are unbounded operators, they can be bounded with respect to the number of particles operator, defined by 
\[ (\cN \Psi)^{(n)} = n \psi^{(n)} \] 
for all $\Psi = \{ \psi^{(n)} \}_{n\in \bN} \in \cF$ or, equivalently, in terms of the operator valued distributions $a_x, a_x^*$, by 
\[ \cN = \int dx \, a_x^* a_x \, . \]
Using this last expression, it is easy to check that
\[\begin{split} \| a(g) \Psi \| &\leq \| g \|_2 \| \cN^{1/2} \Psi \|,  \\
\| a^* (g) \Psi \| &\leq \| g \|_2 \| (\cN+1)^{1/2} \Psi \| \end{split} \]

On $\cF$, we are interested in the time-evolution generated by the Hamilton operator
\begin{equation}\label{eq:HN-F} 
\cH_N = \int dx \nabla_x a_x^* \nabla_x a_x + \frac{1}{2N} \int dx dy \, N^{3\beta} V (N^\beta (x-y)) a_x^* a_y^* a_y a_x 
\end{equation}
for $0 < \beta < 1$ and for a smooth potential $V \geq 0$ with spherical symmetry and with compact support. Notice that $\cH_N$ commutes with the number of particles operator $\cN$; hence, the time evolution preserves the number of particles. On the $n$-particle sector, (\ref{eq:HN-F}) acts as 
\[ \cH_N |_{\cF_n} = \sum_{j=1}^n -\Delta_{x_j} + \frac{1}{N} \sum_{i<j}^n N^{3\beta} V(N^\beta (x-y)) \]

We are going to study the time-evolution generated by (\ref{eq:HN-F}) on coherent states. For $g \in L^2 (\bR^3)$, we define the Weyl operator 
\begin{equation}\label{eq:BHC} W(g) = \exp (a^* (g) - a(g)) = e^{-\| g \|_2^2/2} e^{a^* (g)} e^{-a(g)} \end{equation} 
Weyl operators are unitary $W^* (g) = W(-g) = W^{-1} (g)$ and they act on creation and annihilation by shifts, i.e. for any $f,g \in L^2 (\bR^3)$, we have 
\begin{equation}\label{eq:W-shift}\begin{split} 
W(g)^* a(f) W(g) &= a(f) + \langle f, g \rangle, \\
W (g)^* a^* (f) W(g) &= a^* (f) + \langle g,f \rangle \end{split} \end{equation}


For $\ph \in L^2 (\bR^3)$ with $\| \ph \|_2 = 1$, the coherent state $W(\sqrt{N} \ph) \Omega$ describes a condensate with an average of $N$ particles, all described by the orbital $\ph$ (recall (\ref{eq:WOm}).
To obtain a norm approximation of the many body evolution, we have to implement the correct short scale correlation structure on top of $W(\sqrt{N} \ph) \Omega$. To this end, we fix $\ell > 0$ and we define $f_{N,\ell}$ to be the ground state of the Neumann problem 
\begin{equation}\label{eq:Nf} \left[ -\Delta + \frac{1}{2N} N^{3\beta} V(N^{\beta} x) \right] f_{N,\ell} = \lambda_{N,\ell} f_{N,\ell} 
\end{equation}
on the sphere $|x| \leq \ell$, normalized so that $f_{N,\ell} (x) = 1$ for $|x| = \ell$. In the next lemma, whose proof is deferred to Appendix \ref{s:scattering}, we collect some important properties of the solution of (\ref{eq:Nf}).

\begin{lemma}\label{lm:propomega} Let $V$ be smooth, positive, spherically symmetric and compactly supported with $b_0 = \int V dx$. Let $f_{N, \ell}$ be the ground state of the Neumann problem 
\be
\big(-\D + \frac{1}{2} N^{3\b-1} V(N^{\b}\cdot)\big) f_{N, \ell} = \l_{N,\ell} f_{N, \ell}\label{rescaled_scattChi}
\ee
on the sphere of radius $\ell$, with the boundary conditions
\[
f_{N, \ell}(x) =1 \qquad   \dpr_r f_{N, \ell}(x)=0\,
\]
for all $x \in \bR^3$ with $|x| = \ell$. For $N$ sufficiently large (such that $RN^{-\beta} < \ell$) we have: 
\begin{enumerate}
\item [i)]
\begin{equation}\label{energy}
\left| \l_{N, \ell} - \frac{3 b_0}{8\pi N \ell^{3}} \right| \leq \frac{C}{N^{2-\beta}} \end{equation}
\item [ii)] There is a constant $0<c_0<1$ such that, for all $|x| \leq \ell$,
\begin{align}
c_0 \leq f_{N, \ell}(x) \leq 1   \label{S.ii}
\end{align}
\item [iii)] Let $\o_{N, \ell}=1-f_{N, \ell}$. There exists a constant $C >0$ such that, for all $|x| \leq \ell$,  
\begin{align}
\o_{N, \ell}(x) \leq \frac{C}{N(|x|+N^{-\b})} \qquad |\nabla\o_{N, \ell}(x)| \leq \frac{C}{N(|x|^2+N^{-2\b})} \label{S.iii}
\end{align}
\end{enumerate}
\end{lemma}

We continue $f_{N,\ell}$ to a function on $\bR^3$ by setting $f_{N,\ell} (x) = 1$ for all $|x| \geq \ell$ and we define $\o_{N,\ell} = 1 - f_{N,\ell}$ (so that $\o_{N,\ell} (x) = 0$ for all $|x| \geq \ell$). To generate correlations in the condensate $W(\sqrt{N} \ph) \Omega$ at time $t=0$, we define 
\begin{equation}\label{eq:kN0} k_{N,0} (x;y) = -N \o_{N,\ell} (x-y) \ph^2 ((x+y)/2) 
\end{equation}
With Lemma \ref{lm:propomega} it is easy to check that $\| k_{N,0} \|_2 \leq C$, uniformly in $N$ (the constant depends on $\ell$, though). Hence, (\ref{eq:kN0}) is the integral kernel of a Hilbert-Schmidt operator, which we will denote again by $k_{N,0}$. We use $k_{N,0}$ to define the Bogoliubov transformation 
\[ T_{N,0} = \exp \left[ \frac 1 2\int dx dy k_{N,0} (x;y) a_x^* a_y^* - \text{h.c.} \right] \] 
acting as a unitary operator on $\cF$.  

We are going to study the time evolution of initial data of the form $W(\sqrt{N} \ph) T_{N,0} \xi_N$, for a sequence $\xi_N \in \cF$ satisfying (\ref{eq:ass-xi}) (this assumption guarantees that our initial data are dominated by the Weyl operator $W(\sqrt{N} \ph)$, creating a condensate with an average of $N$ particles in the state $\ph$, and by the Bogoliubov transformation $T_{N,0}$, creating the short scale correlation structure). 

To construct an ansatz to approximate the full evolution of such an initial data, we proceed as follows. First of all, we let the condensate wave function $\ph$ evolve, according to the Hartree equation
\begin{equation}\label{eq:NLSN1}
i\partial_t \phn_t = -\Delta \phn_t + (N^{3\beta} V(N^\beta .)f_{N,\ell} * |\phn_t|^2) \phn_t 
\end{equation}
with the initial data $\phn_{t=0} = \ph$. Some standard properties of this equation (including global well-posedness in the energy space $H^1 (\bR^3)$ and propagation of higher Sobolev regularity) are reviewed in Proposition \ref{prop:ph}. 

On top of the evolved condensate, we add the short scale correlation structure, appropriately modified to take into account the variation of the condensate wave function. Similarly to (\ref{eq:kN0}),
we define 
\begin{equation}\label{eq:kNt} 
k_{N,t} (x;y) = - N \o_{N,\ell} (x-y) (\phn_t ((x+y)/2))^2 
\end{equation}
Like at time $t=0$, the function $k_{N,t}$ turns out to be the integral kernel of a Hilbert-Schmidt operator.
We use it to define the family of Bogoliubov transformations 
\begin{equation}\label{eq:bogNt} T_{N,t} = \exp \left[ \frac 12 \int dx dy k_{N,t} (x;y) a_x^* a_y^* - \text{h.c.} \right] 
\end{equation}
Observe that the unitary operators $T_{N,t}$ act on creation and annihilation operators by
\begin{equation} \begin{split}\label{eq:chsh} T_{N,t}^* a(g) T_{N,t} &= a (\cosh_{k_{N,t}} g) + a^* (\sinh_{k_{N,t}} \overline{g}) \\
 T_{N,t}^* a^* (g) T_{N,t} &= a^* (\cosh_{k_{N,t}} g) + a (\sinh_{k_{N,t}} \overline{g})
\end{split}
\end{equation}
where 
\[ \begin{split} \cosh_{k_{N,t}} &= \sum_{n \geq 0} \frac{1}{(2n)!} \,(k_{N,t}  \overline{k}_{N,t})^n \, , \; \qquad    \sinh_{k_{N,t}} = \sum_{n \geq 0} \frac{1}{(2n+1)!} (k_{N,t} \overline{k}_{N,t})^n k_{N,t}\, . \end{split} \]

We define the fluctuation vector $\xi_{N,t}$ 
at time $t$, requiring that
\begin{equation}\label{eq:iden1} e^{-i \cH_N t} W(\sqrt{N} \ph) T_{N,0} \xi_{N} = W(\sqrt{N} \phn_t) T_{N,t} \xi_{N,t} \end{equation}
Equivalently, we have $\xi_{N,t} = \cU_N (t;0) \xi_{N}$ with the fluctuation dynamics 
\begin{equation}\label{eq:flucd} \cU_N (t;s) = T_{N,t}^* W^* (\sqrt{N} \phn_t) e^{-i (t-s) \cH_N} W(\sqrt{N} \phn_s) T_{N,s} 
\end{equation}

The fluctuation vector $\xi_{N,t}$ measures the distance from the modified coherent state $W(\sqrt{N} \phn_t) T_{N,t} \Omega$. Following the analysis of  \cite{BdOS12}, one can prove a bound of the form
\[ \begin{split} \langle \xi_{N,t} , \cN \xi_{N,t} \rangle &= \langle \cU_N (t;0) \xi_N , \cN \cU_N (t;0) \xi_N \rangle \\ &\leq C \exp (c_1 \exp (c_2 |t|)) \langle \xi_N, \left[ \cN + \cN^2 / N + \cH_N \right] \xi_N \rangle \end{split} \]
for the growth of the expectation of the number of particles with respect to the fluctuation dynamics. As explained in the introduction, this bound provides an estimate of the form (\ref{eq:rate-GP}) for the convergence of the reduced density associated with the full many body evolution (\ref{eq:iden1}).

Here, we want to go one step further, and prove a norm-approximation for (\ref{eq:iden1}), using (\ref{eq:NLSN1}) for the evolution of the condensate, (\ref{eq:bogNt}) for the description of the correlation structure, and a unitary evolution with a quadratic generator to approximate the fluctuation dynamics. 
To this end, we will use the following proposition.
\begin{prop}\label{prop:UN-UN2}
Let $V \geq 0$ be smooth, spherical symmetric and compactly supported. Fix $0 < \beta < 1$ and set $\alpha = \min (\beta/2 , (1-\beta)/2)$. Let $\ph \in H^4 (\bR^3)$. Fix $\ell > 0$. Consider a sequence $\xi_N \in \cF$ such that  $\| \xi_N \| =1$ and 
\[ \langle \xi_N , \left[\cN^2 + \cK^2 + \cV_N \right] \xi_N \rangle \leq C \]
uniformly in $N$. Let $\cU_N$ be the fluctuation dynamics defined in (\ref{eq:flucd}), and $\cU_{2,N}$
be the quadratic evolution defined in (\ref{eq:U2N-def}). Then there exist $C,c_1, c_2 > 0$ such that 
\[ \left\| \cU_N (t;0) \xi_N - e^{-i \int_0^t \eta_N (s) ds} \, \cU_{2,N} (t;0) \xi_N \right\|^2  \leq C N^{-\alpha} \exp (c_1 \exp (c_2 |t|)) \]
for all $t \in \bR$ and all $N$ large enough. 
\end{prop}

Making use of this proposition, we can easily prove Theorem \ref{thm:main1}.

\begin{proof}[Proof of Theorem \ref{thm:main1}]
We notice that
\[ \begin{split} 
&\left\| e^{-i\cH_N t} W(\sqrt{N} \ph) T_{N,0} \xi_N - e^{-i \int_0^t \eta_N (s) ds} \, W (\sqrt{N} \phn_t) T_{N,t} \cU_{2,N} (t;0) \xi_N 
\right\|^2 \\ &\hspace{2cm} =  \left\| W (\sqrt{N} \phn_t) T_{N,t} \left[ \cU_N (t;0) - e^{-i \int_0^t \eta_N (s) ds} \, \cU_{2,N} (t;0) \right] \xi_N \right\|^2 \\ &\hspace{2cm} = \left\| \left[ \cU_N (t;0) - e^{-i \int_0^t \eta_N (s) ds} \, \cU_{2,N} (t;0) \right] \xi_N \right\|^2 
\leq C N^{-\alpha} \exp (c_1 \exp (c_2 |t|)) \end{split} \]
where, in the last inequality, we used Proposition \ref{prop:UN-UN2}. 
\end{proof}

To prove Proposition \ref{prop:UN-UN2} we will compute the generator of the fluctuation dynamics $\cU_N$ and we will compare it with the quadratic generator (\ref{eq:L2N0}) of $\cU_{2,N}$. This is the content of the next two sections.

\section{The Generator of the Fluctuation Dynamics}\label{s:est}

{F}rom the definition of the fluctuation dynamics (\ref{eq:flucd}), we find 
\[ i\partial_t \cU_N (t;s) = \cL_N (t) \cU_N (t;s) \]
with the generator 
\begin{equation}\label{eq:cLN} \begin{split}  
\cL_N (t) = \; &(i\partial_t T_{N,t}^*) T_{N,t} + T_{N,t}^* \left[ \left(i\partial_t W^* (\sqrt{N} \phn_t) \right) W (\sqrt{N} \phn_t) \right. \\ &\hspace{4cm} \left. + W^* (\sqrt{N} \phn_t) \cH_N W (\sqrt{N} \phn_t) \right] T_{N,t} 
\end{split} 
\end{equation}

The goal of this section is to prove the following theorem, describing the properties of the generator $\cL_N (t)$. 
\begin{theorem}\label{thm:LN2}
Let $V \geq 0$ be smooth, spherical symmetric and compactly supported. Fix $0 < \beta < 1$ and set $\alpha = \min (\beta/2 , (1-\beta)/2)$. Let $\ph \in H^4 (\bR^3)$. Fix $\ell > 0$. Let $\cU_N$ be the fluctuation dynamics defined in (\ref{eq:flucd}). Then
\[ \cL_N (t) = \eta_N (t) + \cL_{2,N} (t) + \cV_N + \cE_N (t) \]
where the phase $\eta_N (t)$ and the quadratic generator $\cL_{2,N} (t)$ are given by (\ref{eq:etaN}) and (\ref{eq:L2N0}), respectively, $\cV_N$ is defined in (\ref{eq:cVN}),  
and where the error term $\cE_N (t)$ satisfies
\begin{equation}\label{eq:est-cE1} \begin{split} 
\pm \cE_N (t) \leq \delta \cV_N + C \cN^2 /N + C_\delta e^{K|t|} (\cN+1) \\
\pm [\cN, \cE_N (t)] \leq \delta \cV_N + C \cN^2 /N + C_\delta e^{K|t|} (\cN+1) \\
\pm \dot{\cE}_N (t) \leq \delta \cV_N + C \cN^2 /N + C_\delta e^{K|t|} (\cN+1) \end{split} \end{equation}
and 
\begin{equation}\label{eq:est-cE2} \left| \langle \psi_1, \cE_N (t) \psi_2 \rangle \right| \leq C N^{-\alpha} e^{K|t|} \left[ \langle \psi_1, (\cK+ \cN + 1) \psi_1 \rangle + \langle \psi_2 , (\cK^2 + (\cN+1)^2) \psi_2 \rangle \right] \end{equation}
for all $\psi_1, \psi_2 \in \cF$. Moreover, the quadratic generator $\cL_{2,N} (t)$ satisfies the bounds 
\begin{equation}\label{eq:est-cL2N} \begin{split} 
\pm (\cL_{2,N} (t) - \cK) &\leq C e^{K|t|} (\cN+1), \qquad \hspace{.25cm} (\cL_{2,N} (t) - \cK)^2 \leq C e^{K|t|} (\cN+1)^2 \\
\pm \left[ \cN , \cL_{2,N} (t) \right] &\leq C e^{K|t|} (\cN+1),  \qquad \hspace{.07cm} \pm \left[ \cN^2 , \cL_{2,N} (t) \right] \leq C e^{K|t|} (\cN+1)^2 \\
\pm \dot{\cL}_{2,N} (t) &\leq C e^{K|t|} (\cN+1) ,\qquad 
\hspace{1.1cm} |\dot{\cL}_{2,N} (t)|^2 \leq C e^{K|t|} (\cN+1)^2
\end{split} 
\end{equation}
\end{theorem}

To prove Theorem \ref{thm:LN2}, we compute the different terms on the r.h.s. of (\ref{eq:cLN}). 
{F}rom (\ref{eq:W-shift}) we obtain 
\[ (i\partial_t W^* (\sqrt{N} \phn_t)) W (\sqrt{N} \phn_t)  + W^* (\sqrt{N} \phn_t) \cH_N W (\sqrt{N} \phn_t) = \sum_{j=0}^4 \cG_N^{(j)} (t) \]
with 
\[ \begin{split} \cG_N^{(0)} (t) = \; & N \int dx dy N^{3\beta} V(N^\beta (x-y)) (1/2 - f_{N,\ell} (x-y)) |\phn_t (x)|^2 |\phn_t (y)|^2 \\
\cG_N^{(1)} (t) = \; &\sqrt{N} \left[ a^* ((N^{3\beta} V(N^\beta .) \omega_{N,\ell}  * |\phn_t|^2) \phn_t) \right. \\ &\hspace{4cm} \left.+ a ((N^{3\beta} V(N^\beta .) \omega_{N,\ell}  * |\phn_t|^2) \phn_t) \right]
 \\
\cG_N^{(2)} (t) = &\; \int dx \nabla_x a_x^* \nabla_x a_x + \int dx \, (N^{3\beta} V(N^\beta .) * |\phn_t|^2 )(x) a_x^* a_x \\ &+ \int dx dy \, N^{3\beta} V(N^\beta (x-y)) \phn_t (x) \overline{\ph}^N_t (y) a_x^* a_y \\ &+ \frac{1}{2} \int dx dy \, N^{3\beta} V(N^\beta (x-y)) \left[ \phn_t (x) \phn_t (y) a_x^* a_y^* + \overline{\ph}^N_t (x) \overline{\ph}^N_t (y) a_x a_y \right] \\  
\cG_N^{(3)} (t) = &\; \frac{1}{\sqrt{N}} \int dx N^{3\beta} V(N^\beta (x-y)) a_x^* \left( \phn_t (y) a_y^* + \overline{\ph}^N_t (y) a_y \right) a_x \\ 
\cG_N^{(4)} (t) = &\; \frac{1}{2N} \int dx dy N^{3\beta} V(N^\beta (x-y)) a_x^* a_y^* a_y a_x 
\end{split} 
\] 
{F}rom (\ref{eq:cLN}), we conclude that
\[ \cL_N (t) = \left[ i\partial_t T_{N,t}^* \right] T_{N,t} + \sum_{j=0}^4 T_{N,t}^* \cG_N^{(j)} (t) T_{N,t} \]
Next, we analyze the contributions $T_{N,t}^* \cG_N^{(j)} (t) T_{N,t}$ to the generator $\cL_N (t)$ separately. 

For $j=0$, we have 
\begin{equation}\label{eq:G0N}
T_{N,t}^* \cG^{(0)}_N (t) T_{N,t} = N \int dx dy N^{3\beta} V(N^\beta (x-y)) (1/2 - f_{N,\ell} (x-y)) |\phn_t (x)|^2 |\phn_t (y)|^2
\end{equation}
For $j=1$, we just write 
\begin{equation}\label{eq:G1N}
T_{N,t}^* \cG^{(1)}_N (t) T_{N,t} = \sqrt{N} \, T_{N,t}^* \left[  a^* ((N^{3\beta} V(N^\beta .) \omega_{N,\ell} * |\phn_t|^2) \phn_t) + \text{h.c.} \right] T_{N,t}
\end{equation}
since at the end the main part of this term will be canceled. 

In the next four subsections, we study the terms $T_{N,t}^* \cG_{N}^{(j)} (t) T_{N,t}$, for $j=2,3,4$ and the term $(i\partial_t T_{N,t}^*) T_{N,t}$. At the end, in subsection \ref{s:pf-thm3}, we combine all results to prove Theorem \ref{thm:LN2}.


\subsection{Analysis of $T_{N,t}^* \cG^{(2)}_N (t) T_{N,t}$}
\label{ss:q-terms}

We denote by 
\[ \cK = \int dx \, \nabla_x a_x^* \nabla_x a_x \]
the kinetic energy operator. We will use the shorthand notation $c_{N,t} = \cosh_{k_{N,t}}$, $s_{N,t} = \sinh_{k_{N,t}}$, $p_{N,t} = c_{N,t} - 1$, $r_{N,t} = s_{N,t} - k_{N,t}$ and, for a fixed $x \in \bR^3$, $c_x^{N}  (y) = c_{N,t} (y;x)$, $s_x^N (y) = s_{N,t} (y;x)$, $p_x^N (y) = p_{N,t} (y;x)$ and $r_x^N (y) = r_{N,t} (y;x)$.

The goal of this subsection is to prove the following proposition, which is a consequence of Proposition \ref{prop:L2K} and Proposition \ref{prop:L2hat} below. 
\begin{prop}\label{prop:quadr}
Under the assumptions of Theorem \ref{thm:LN2}, we have 
\begin{equation}\label{eq:G2N} \begin{split}
 T_{N,t}^* \cG_N^{(2)} T_{N,t} = &\int dx dy |\nabla_x \sinh_{k_{N,t}} (x,y)|^2 + \int dx (N^{3\beta} V(N^\beta .) *|\phn_t|^2) (x) \langle s_x^N, s_x^N \rangle \\ &+ \int dx dy N^{3\beta} V(N^\beta (x-y)) \phn_t (x) \bar{\ph}^N_t (y) \langle s_x^N,s_y^N \rangle \\
 &+\text{Re } \int dx dy \, N^{3\beta} V(N^\beta (x-y)) \phn_t (x) \phn_t (y) \langle s_x^N , c_y^N \rangle  \\
  &+ \int dx dy \, (-\Delta_x k_{N,t} (x,y)) a_x^* a_y^* + \int dx dy \, (-\Delta_x k_{N,t} (x,y)) a_x a_y \\ 
&+ \frac{1}{2} \int dx dy N^{3\beta} V(N^\beta (x-y)) \left[ \phn_t (x) \phn_t (y) a^*_x a^*_y + \bar{\ph}^N_t (x) \bar{\ph}^N_t (y) a_x a_y \right]  \\
  &+ \cL^{(K)}_{2,N} (t) + \cL^{(V)}_{2,N} (t)  \end{split} 
\end{equation}
with $\cL_{2,N}^{(K)}$ and $\cL_{2,N}^{(V)}$ given in (\ref{eq:LK}) and, respectively, in (\ref{eq:LV}). Putting
\begin{equation}\label{eq:L2N2}  
\cL^{(2)}_{2,N} (t) = \cL^{(K)}_{2,N} (t) + \cL^{(V)}_{2,N} (t) 
\end{equation}
we find 
\[ \begin{split} 
\pm (\cL^{(2)}_{2,N} (t) - \cK) &\leq C e^{K|t|} (\cN+1), \qquad \hspace{.22cm} (\cL^{(2)}_{2,N} (t) - \cK)^2 \leq C e^{K|t|} (\cN+1)^2 \\
\pm \left[ \cN , \cL^{(2)}_{2,N} (t) \right] &\leq C e^{K|t|} (\cN+1), \qquad  \pm \left[ \cN^2 , \cL^{(2)}_{2,N} (t) \right] \leq C e^{K|t|} (\cN+1)^2 \\
\pm \dot{\cL}^{(2)}_{2,N} (t) &\leq  C e^{K|t|} (\cN+1), \qquad \hspace{1cm} (\dot{\cL}^{(2)}_{2,N} (t))^2 \leq C e^{K |t|} (\cN+1)^2 
\end{split} \]
\end{prop}

To prove Proposition \ref{prop:quadr}, we split
\begin{equation}\label{eq:G2-split} T_{N,t}^* \cG^{(2)}_{N} (t) T_{N,t} = T_{N,t}^* \cK T_{N,t} + T_{N,t}^* \cG_{N}^{(2,V)} (t) T_{N,t} \end{equation}
with  
\begin{equation} \label{eq:L2-split}
\begin{split}  \cG_N^{(2,V)} (t) = \; &\int dx \, (N^{3\beta} V(N^\beta .) * |\phn_t|^2 )(x) a_x^* a_x \\ &+ \int dx dy \, N^{3\beta} V(N^\beta (x-y)) \phn_t (x) \bar{\ph}^N_t (y) a_x^* a_y \\ &+ \frac{1}{2} \int dx dy \, N^{3\beta} V(N^\beta (x-y)) \left[ \phn_t (x) \phn_t (y) a_x^* a_y^* + \bar{\ph}^N_t (x) \bar{\ph}^N_t (y) a_x a_y \right]  \end{split} 
\end{equation}

We start by analyzing $T_{N,t}^* \cK T_{N,t}$.  {F}rom (\ref{eq:chsh}), we find 
\[ \begin{split} T_t^* \cK T_t = \;& \int dx dy  \,| \nabla_x \sinh_{k_t} (y,x)|^2 \\ & + \int dx dy \, (-\Delta_x k_{N,t} (x,y)) a_x^* a_y^* + \int dx dy \, (-\Delta_x k_{N,t} (x,y)) a_x a_y \\ &+ \cL_{2,N}^{(K)} (t) \end{split} \]
with $\cL_{2,N}^{(K)}$ as given in (\ref{eq:LK}).
%

\begin{prop}\label{prop:L2K}
Under the assumptions of Theorem \ref{thm:LN2}, the operator $\cL_{2,N}^{(K)} (t)$ defined in (\ref{eq:LK}) satisfies the bounds  
\begin{equation}\label{eq:prop21} \begin{split}
\pm (\cL_{2,N}^{(K)} (t) - \cK) &\leq C e^{K|t|} (\cN + 1),  \qquad \hspace{.4cm} (\cL_{2,N}^{(K)} - \cK)^2 \leq C e^{K|t|} (\cN +1)^2 \\
\pm \left[ \cN , \cL_{2,N}^{(K)}  \right] &\leq  C e^{K|t|} (\cN +1), \qquad \pm [\cN^2 , \cL_{2,N}^{(K)} (t) ] \leq C e^{K|t|} (\cN+1)^2 \\
\pm \dot{\cL}_{2,N}^{(K)} (t) &\leq C e^{K|t|} (\cN+1), \qquad \hspace{.7cm} \left[ \dot{\cL}_{2,N}^{(K)} (t) \right]^2 \leq C e^{K|t|} (\cN+1)^2 \end{split} \end{equation}
\end{prop}

The proof of Proposition \ref{prop:L2K} is based on the estimates contained in the next lemma. 
\begin{lemma}\label{lm:L2K}
Let $j_1, j_2 \in L^2 (\bR^3 \times \bR^3)$ and denote $j_{i,x} (z) = j_i (z,x)$ for $i=1,2$. Set 
\[ \begin{split} 
A_1 &= \int dx a^\sharp (j_{1,x}) a_x \\
A_2 &= \int dx a^\sharp (j_{1,x}) a^\sharp (j_{2,x}) \end{split} \]
Then, we have
\begin{equation}\label{eq:A123} 
\begin{split} 
|\langle \psi, A_1 \psi \rangle | &\leq C \| j_1 \|_2 \| (\cN+1)^{1/2} \psi \|^2 \\
|\langle \psi, A_2 \psi \rangle | &\leq C \| j_1 \|_2 \| j_2 \|_2 \| (\cN+1)^{1/2} \psi \|^2 
\end{split} 
\end{equation} 
for all $\psi \in \cF$. Moreover, 
\begin{equation}\label{eq:A123-2} \begin{split} A_1^* A_1 + A_1 A_1^*   &\leq C \| j_1 \|_2^2 \, (\cN+1)^2  
\\
A_2^* A_2 + A_2 A_2^*  & \leq C \| j_1 \| \| j_2 \|_2  \, (\cN+1)^2
\end{split} \end{equation}
Furthermore, let
\[ A_3 = \int dx a^* (\nabla_x k^N_x) a(\nabla_x k^N_x) \]
where we put $k^N_x (y)=k_{N,t}(y,x)$, with $k_{N,t}$ as defined in (\ref{eq:kNt}), with initial data $\ph_{N,t=0} = \ph \in H^4 (\bR^3)$. Then we have
\begin{equation}\label{eq:A4}
A_3 \leq C e^{K|t|} (\cN+1), \qquad  A_3^2 \leq C e^{K|t|} (\cN+1)^2 
\end{equation}
and also the bound
\begin{equation}\label{eq:A4dot-bd}  \dot{A}_3^* \dot{A}_3 + \dot{A}_3 \dot{A}_3^* \leq C e^{K|t|} (\cN+1)^2 \end{equation}
for the time derivative 
\begin{equation}\label{eq:A4dot} \dot{A}_3 = \int dx a^* (\nabla_x \dot{k}^N_x) a (\nabla_x k^N_x) + \int dx a^* (\nabla_x k^N_x) a(\nabla_x \dot{k}^N_x) 
\end{equation}
\end{lemma}

\begin{proof}
We start with (\ref{eq:A123}). To this end, we observe that
\[ \begin{split} | \langle \psi, A_1 \psi \rangle | &\leq \int dx \| a^\sharp (j_{1,x}) \psi \| \| a_x \psi \| \\ &\leq \int dx \| j_{1,x} \|_2 \| (\cN+1)^{1/2} \psi \| \|a_x \psi \| \leq \| j_1 \|_2 \|(\cN+1)^{1/2} \psi \|^2 \end{split} \]
by Cauchy-Schwarz. The bound for $A_2$ follows similarly.
Now, we show (\ref{eq:A123-2}). We begin with
\[ \begin{split} \langle \psi, A_1^* A_1 \psi \rangle  &= \int dx dy \langle \psi, a_x^* a^\sharp (j_{1,x}) a^{\sharp} (j_{1,y}) a_y \psi \rangle \\ &\leq \int dx dy \| a^\sharp (j_{1,x}) a_x \psi \| \| a^\sharp (j_{1,y}) a_y \psi \| \\ &\leq \int dx dy \| j_{1,x} \|_2 \| j_{2,y} \|_2 \| a_x \cN^{1/2} \psi \| \| a_y \cN^{1/2} \psi \| \\ &\leq \left[ \int dx dy \| j_{1,x} \|_2^2 \| a_y \cN^{1/2} \psi \|^2 \right]^{1/2} \left[ \int dx dy \| j_{1,y} \|_2^2 \| a_x \cN^{1/2} \psi \|^2 \right]^{1/2} \\ &\leq \| j_1 \|_2^2 \| \cN \psi \|^2   \end{split} \] 
Analogously 
\[ \begin{split} 
\langle \psi, A_1 &A_1^* \psi \rangle \\ = \; & \int dx dy \langle \psi, a^\sharp (j_{1,x}) a_x a_y^* a^\sharp (j_{1,y}) \psi \rangle \\ \leq \; & \int dx dy \langle \psi, a^\sharp (j_{1,x}) a_y^* a_x a^\sharp (j_{1,y}) \psi \rangle + \int dx \langle \psi, a^\sharp (j_{1,x}) a^\sharp (j_{1,x}) \psi \rangle \\
\leq \; &\int dx dy \| a_y a^\sharp (j_{1,x}) \psi \| \| a_x a^\sharp (j_{1,y}) \psi \| + \| j_{1} \|^2 \| (\cN+1)^{1/2} \psi \|^2  \end{split} \]
which implies, as above, that $\langle \psi, A_1 A_1^* \psi \rangle \leq 2 \| j_1 \|^2_2 \| (\cN+1) \psi \|^2$. The estimate for $A_2^* A_2 + A_2 A_2^*$ follows analogously. 

We switch now to the term $A_3$. We write it as
\[ A_3 = \int dx dy dz \, \nabla_x k_{N,t} (y,x) \nabla_x \bar{k}_{N,t} (z,x) a_z^* a_y = \int dy \, a^* (u^N_y) a_y \]
with 
\[ u^N_y (z) = u_{N,t} (z,y) = \int dx \, \nabla_x k_{N,t} (y,x) \nabla_x \bar{k}_{N,t} (z,x)  \]
To bound the $L^2 (\bR^3 \times \bR^3)$ norm of $u_{N,t}$, we notice that 
\[ \begin{split} \| u_{N,t} \|_2^2&=\int dz dy\Big|\int dx \, \nabla_x k_{N,t} (y,x) \nabla_x \bar{k}_{N,t} (z,x)\Big|^2\\
&=\int dz dy dx_1dx_2 \, \nabla_{x_1} \bar k_{N,t} (y,x_1) \nabla_{x_1} {k}_{N,t} (z,x_1)\nabla_{x_2} k_{N,t} (y,x_2) \nabla_{x_2} \bar{k}_{N,t} (z,x_2)\\
\end{split}\]
Since
\[ |\nabla_x k_{N,t} (x;z)| \leq C \frac{{\bf 1} (|x-z| \leq \ell)}{|x-z|^2} \left[ |\phn_t ((x+z)/2)|^2 + |\nabla \phn_t ((x+z)/2)|^2 \right]  \]
we find
\[\begin{split} \|u_{N,t} \|_2^2 \leq \; &\| \phn_t \|_{H^3}^6 \int dx_1 dx_2 dy \, \left[ |\phn_t ((x_1+y)/2)|^2 + |\nabla \phn_t ((x_1+y)/2)|^2 \right] \\ &\hspace{2cm} \times \frac{{\bf 1} (|x_1-y| \leq \ell) {\bf 1} (|x_2 -y| \leq \ell)}{|x_1 - y|^2 |x_2 - y|^2} \int dz \frac{1}{|x_1 - z|^2 |z - x_2|^2} \\
\leq \; & \| \phn_t \|_{H^3}^6 \| \phn_t \|_{H^1}^2 \int dx_1 dx_2 \frac{{\bf 1} (|x_1| \leq \ell) {\bf 1} (|x_2| \leq \ell)}{|x_1|^2 |x_2|^2 |x_1 - x_2|} \\ \leq \; &C e^{K|t|} \end{split} \]
Hence, (\ref{eq:A4}) follows from (\ref{eq:A123}) and (\ref{eq:A123-2}). As for (\ref{eq:A4dot-bd}), we proceed similarly, remarking that 
\[ \begin{split} \tilde{u}^N_1 (y,z) &= \int dx \nabla_x \dot{k}_{N,t} (y,x) \nabla_x \bar{k}_{N,t} (z,x) \qquad \text{and } \\
\tilde{u}_2^N (y,z) &= \int dx \nabla_x k_{N,t}(y,x) \nabla_x \dot{\bar{k}}_{N,t} (z,x) \end{split} \]
are such that $\| \tilde{u}^N_i \|_2 \leq C e^{K|t|}$ for $i=1,2$ (one has to observe that one can avoid to take the terms $\nabla \dot{\ph}^N_t$ in the $L^\infty$-norm).  
\end{proof}

\begin{proof}[Proof of Proposition \ref{prop:L2K}] 
Since \[ \| k_{N,t} \|_2, \| s_{N,t} \|_2, \| p_{N,t} \|_2, \| r_{N,t} \|_2, \| \nabla_x p_{N,t} \|_2, \| \Delta_x p_{N,t} \|_2, \| \Delta_x r_{N,t} \|_2 \leq C e^{K|t|} \] are all finite, uniformly in $N$ (see section \ref{s:kernel}), we note that the terms in (\ref{eq:LK}), with the exception of the kinetic energy operator $\cK$, have the form of one of the operators $A_1, A_1^*, A_2, A_3$ analyzed in Lemma \ref{lm:L2K}. {F}rom (\ref{eq:A123}), (\ref{eq:A4}) it follows immediately that 
\[ \pm (\cL_{2,N}^{(K)} (t) - \cK) \leq C e^{K|t|} (\cN+1) \]
With Cauchy-Schwarz, we can bound the square of $(\cL_{2,N}^{(K)} - \cK)$ by the sum of terms like $A_i A_i^*$ and $A_i^* A_i$, with $i=1,2,3$. Lemma \ref{lm:L2K} implies therefore that 
\[ (\cL_{2,N}^{(K)} - \cK)^2 \leq C e^{K|t|} (\cN+1)^2 \] 

Since the commutators of $\cN$ with the terms on the r.h.s. of (\ref{eq:LK}) are either zero, or proportional to the same terms, we also obtain the bounds on the second line of (\ref{eq:prop21}).

Finally, since (see Section \ref{s:kernel}) \[ \| \dot{k}_{N,t} \|_2, \| \dot{s}_{N,t} \|_2, \| \dot{p}_{N,t} \|_2, \| \dot{r}_{N,t} \|_2, \| \nabla_x \dot{p}_{N,t} \|_2, \|
\Delta_x \dot{p}_{N,t} \|_2, \| \Delta_x \dot{r}_{N,t} \|_2 \leq C e^{K|t|} \] we also obtain the estimates on the last line of (\ref{eq:prop21}). 
%
\end{proof}

Next, we analyze the second term in (\ref{eq:G2-split}). We have 
\[ \begin{split}
T_t^* &\cG_{N}^{(2,V)} (t) T_t \\ = \; &
\int dx (N^{3\beta} V(N^\beta .) *|\phn_t|^2) (x) \langle s_x^N, s_x^N \rangle \\ &+\int dx dy N^{3\beta} V(N^\beta (x-y)) \phn_t (x) \bar{\ph}^N_t (y) \langle s_x^N,s_y^N \rangle \\ &+\text{Re } \int dx dy \, N^{3\beta} V(N^\beta (x-y)) \phn_t (x) \phn_t (y) \langle s_x^N , c_y^N \rangle \\ 
&+\frac{1}{2} \int dx dy N^{3\beta} V(N^\beta (x-y)) \left[ \phn_t (x) \phn_t (y) a^*_x a^*_y + \bar{\ph}^N_t (x) \bar{\ph}^N_t (y) a_x a_y \right] \\ &+ \cL^{(V)}_{2,N} (t) \end{split} \]
with $\cL_{2,N}^{(V)}$ given in (\ref{eq:LV}). The properties of the generator $\cL_{2,N}^{(V)}$ are analyzed in the next proposition, which together with Proposition \ref{prop:L2K}, implies Proposition \ref{prop:quadr}.
\begin{prop}\label{prop:L2hat}
Under the assumptions of Theorem \ref{thm:LN2}, the operator $\cL_{2,N}^{(V)} (t)$ defined in (\ref{eq:LV}) satisfies the bounds 
\[ \begin{split}
\pm \cL_{2,N}^{(V)} (t) &\leq C e^{K|t|} (\cN+1), \qquad  \hspace{.7cm} (\cL_{2,N}^{(V)})^2 (t) \leq C e^{K|t|} (\cN+1)^2 \\
\pm [\cN, \cL_{2,N}^{(V)} (t)] &\leq C e^{K|t|} (\cN+1), \qquad \pm[ \cN^2 , \cL_{2,N}^{(V)} (t) ] \leq C e^{K|t|} (\cN+1)^2 \\
\dot{\cL}_{2,N}^{(V)} (t) &\leq C e^{K|t|} (\cN+1), 
\qquad \hspace{.85cm} |\dot{\cL}_{2,N}^{(V)} (t)|^2 \leq C e^{K|t|} (\cN+1)^2 \end{split} \]
\end{prop}
To prove Proposition \ref{prop:L2hat}, we will make use of bounds contained in the following lemma. 
\begin{lemma}\label{lm:L2hat}
Let $j_1, j_2 \in L^2 (\bR^3 \times \bR^3)$. We consider the operators 
\[ \begin{split} 
A_{1,1} &= \int dx dy N^{3\beta} V(N^\beta (x-y)) \ph_1 (x) \ph_2 (y)  a^\sharp (j_{1,x}) a^\sharp (j_{2,y}) \\
A_{1,2} &= \int dx dy N^{3\beta} V(N^\beta (x-y)) \ph_1 (x) \ph_2 (y)  a^\sharp (j_{1,x}) a_y  \\
A_{1,3} &= \int dx dy N^{3\beta} V(N^\beta (x-y)) \ph_1 (x) \ph_2 (y)  a^*_x a_y \end{split} \]
and 
\[ \begin{split} 
A_{2,1} &= \int dx dy N^{3\beta} V(N^\beta (x-y)) \ph_1 (x) \ph_2 (x) a^\sharp (j_{1,y}) a^\sharp (j_{2,y}) \\
A_{2,2} &= \int dx dy N^{3\beta} V(N^\beta (x-y)) \ph_1 (x) \ph_2 (x)   a^\sharp (j_{1,y}) a_y  \\
A_{2,3} &= \int dx dy N^{3\beta} V(N^\beta (x-y)) \ph_1 (x) \ph_2 (x) a^*_y a_y \end{split} \]
Then we have
\begin{equation}\label{eq:Aj123} 
\begin{split} 
| \langle \psi, A_{j,1} \psi \rangle | &\leq  C  \, \| j_1 \|_2 \| j_2 \|_2 \| \ph_1 \|_{H^2} \| \ph_2 \|_{H^2} \| (\cN+1)^{1/2} \psi \|^2 \\
| \langle \psi, A_{j,2} \psi \rangle | &\leq  C   \| j_1 \|_2 \| \ph_1 \|_{H^2} \| \ph_2 \|_{H^2} \| (\cN+1)^{1/2} \psi \|^2 \\
|\langle \psi, A_{j,3} \psi \rangle | & \leq C    \| \ph_1 \|_{H^2} \| \ph_2 \|_{H^2} \| (\cN+1)^{1/2} \psi \|^2 
\end{split}  
\end{equation}
for $j=1,2$ and all $\psi \in \cF$. Moreover,
\begin{equation}\label{eq:Aj123-2} \begin{split} 
A_{j,1}A_{j,1}^* + A_{j,1}^* A_{j,1} &\leq C  \| j_1\|^2_2 \|j_2 \|_2^2 \| \ph_1 \|^2_{H^2} \| \ph_2 \|^2_{H^2} \, (\cN+1)^2 \\
A_{j,2}A_{j,2}^* + A_{j,2}^* A_{j,2} &\leq C   \| j_1\|^2_2 \| \ph_1 \|^2_{H^2} \| \ph_2 \|^2_{H^2} \, (\cN+1)^2 \\
A_{j,3}A_{j,3}^* + A_{j,3}^* A_{j,3} &\leq C   \| \ph_1 \|^2_{H^2} \| \ph_2 \|^2_{H^2} \, (\cN+1)^2 
\end{split} \end{equation}
for $j=1,2$ and all $\psi \in \cF$. 
\end{lemma}

\begin{proof}
We start with the bounds (\ref{eq:Aj123}), for the case $j=1$. We have
\[ \begin{split} | \langle \psi , A_{1,1} \psi \rangle | &\leq \int dx dy N^{3\beta} V(N^\beta (x-y)) |\ph_1 (x)| |\ph_2 (y)| \| a^\sharp (j_{1,x}) \psi \| \| a^\sharp (j_{2,y}) \psi \| \\ &\leq \| \ph_1 \|_\infty \| \ph_2 \|_\infty  \int dx dy N^{3\beta} V(N^\beta (x-y)) \| j_{1,x} \|_2 \| j_{2,y} \|_2 \| (\cN+1)^{1/2} \psi \|^2 \\ 
&\leq \| \ph_1 \|_{H^2} \|\ph_2 \|_{H^2}  \| (\cN+1)^{1/2} \psi \|^2\,  \left[ \int dx dy N^{3\beta} V(N^\beta (x-y)) \| j_{1,x} \|_2^2 \right]^{1/2} \\ &\hspace{4cm} \times  \left[ \int dx dy N^{3\beta} V(N^\beta (x-y)) \| j_{2,y} \|^2_2 \right]^{1/2} \\ &\leq C  \| j_1 \|_2 \| j_2 \|_2 \| \ph_1 \|_{H^2} \| \ph_2 \|_{H^2} \| (\cN+1)^{1/2} \psi \|^2 \end{split} \]
Moreover,
\[ \begin{split} 
|\langle \psi , A_{1,2} \psi \rangle | &\leq \int dx dy N^{3\beta} V(N^\beta (x-y)) |\ph_1 (x)||\ph_2 (y)| \| a^{\sharp} (j_{1,x}) \psi \| \| a_y \psi \| \\ &\leq C \| \ph_1 \|_\infty \| \ph_2 \|_\infty \left[ \int dx dy N^{3\beta} V(N^\beta (x-y)) \| j_{1,x} \|^2_2 \| (\cN+1)^{1/2} \psi \|^2 \right]^{1/2} \\ &\hspace{2cm} \times  \left[ \int dx dy N^{3\beta} V(N^\beta (x-y)) \| a_y \psi \|^2 \right]^{1/2}  \\ &\leq C \| j_1 \| \| \ph_1 \|_{H^2} \| \ph_2 \|_{H^2} \| (\cN+1)^{1/2} \psi \|^2 \end{split} \]
Finally, 
\[ \begin{split} | \langle \psi, A_{1,3} \psi \rangle | &\leq \int dx dy N^{3\beta} V(N^\beta (x-y)) |\ph_1 (x)| |\ph_2 (y)| \| a_y \psi \| \| a_x \psi \| \\ &\leq \| \ph_1 \|_\infty \| \ph_2 \|_\infty \int dx dy  N^{3\beta} V(N^\beta (x-y)) \| a_y \psi \|^2 \\ &\leq C e^{K|t|} \| \cN^{1/2} \psi \|^2 \end{split} \]
The terms $A_{2,j}$, with $j=1,2,3$ can be bounded similarly; we skip here the straightforward 
details. We switch now to the estimates (\ref{eq:Aj123-2}). Again, we consider the case $j=1$. We have
\[ \begin{split} 
\langle \psi, A_{1,1} A_{1,1}^* \psi \rangle \leq \; &\int dx dy dz dw N^{3\beta} V(N^\beta (x-y)) N^{3\beta} V(N^\beta (z-w)) \\ &\hspace{.2cm} \times |\ph_1 (x)| |\ph_2 (y)| |\ph_1 (z)| |\ph_2 (w)| \| a^\sharp (j_{1,x}) a^\sharp (j_{2,y}) \psi \| \| a^\sharp (j_{1,z}) a^\sharp (j_{2,w}) \psi \| \\ \leq \; &C \| \ph_1 \|^2_\infty \| \ph_2 \|_\infty^2 \int dx dydz dw N^{3\beta} V(N^\beta (x-y)) N^{3\beta} V(N^\beta (z-w)) \\ &\hspace{.2cm} \times \| j_{1,x} \|_2 \| j_{2,y} \|_2 \| j_{1,z} \|_2 \| j_{2,w} \|_2 \| (\cN+1) \psi \|^2 \\ \leq \; &C \| \ph_1 \|_{H^2}^2 \| \ph_2 \|_{H^2}^2 \| (\cN+1) \psi \|^2  \\ &\hspace{.2cm} \times \int dx dy dz dw  N^{3\beta} V(N^\beta (x-y))  N^{3\beta} V(N^\beta (z-w)) \| j_{1,x} \|^2_2 \| j_{2,w} \|^2_2 \\ \leq \; &C  \| j_1 \|^2 \| j_2 \|^2 \| \ph_1 \|_{H^2}^2 \| \ph_2 \|_{H^2}^2 \, \| (\cN+1) \psi \|^2 \end{split} \]
The contribution $A_{1,1}^* A_{1,1}$ can be estimated exactly in the same way. Let us now consider the term $A_{1,3} A_{1,3}^*$. We find
\[ \begin{split} \langle \psi & , A_{1,3} A_{1,3}^* \psi \rangle \\ = \; & \int dx dy dz dw N^{3\beta} V(N^\beta (x-y)) N^{3\beta} V(N^\beta (z-w)) \\ &\hspace{5cm} \times  \ph_1 (x) \ph_2 (y) \overline{\ph}_1 (z) \overline{\ph}_2 (w) \langle \psi , a_x^* a_y a_w^* a_z \psi \rangle \\ = \; & \int dx dy dz dw N^{3\beta} V(N^\beta (x-y)) N^{3\beta} V(N^\beta (z-w)) \\ &\hspace{5cm} \times   \ph_1 (x) \ph_2 (y) \overline{\ph}_1 (z) \overline{\ph}_2 (w) \langle \psi, a_x^* a_w^* a_y a_z \psi \rangle \\ &+ \int dx dy dz N^{3\beta} V(N^\beta (x-y)) N^{3\beta} V(N^\beta (z-y)) \ph_1 (x) |\ph_2 (y)|^2 \overline{\ph}_1 (z) \langle \psi, a_x^* a_z  \psi \rangle \end{split} \]
which leads us, by Cauchy-Schwarz, to $\langle \psi, A_{1,3} A_{1,3}^* \psi \rangle \leq C \| \ph_1 \|_{H^2}^2 \| \ph_2 \|_{H^2}^2 \| (\cN+1) \psi \|^2$. Since $A_{1,3}^*$ equals $A_{1,3}$, with $\ph_1$ and $\ph_2$ exchanged, the operator $A_{1,3}^* A_{1,3}$ can be bounded in the same way. Also the estimates for the terms $A_{1,2} A_{1,2}^*$, $A_{1,2}^* A_{1,2}$ and for the squares of the operators $A_{2,j}$, with $j=1,2,3$, can be obtained similarly. 
\end{proof}

\begin{proof}[Proof of Proposition \ref{prop:L2hat}]
Writing $\cosh_{k_{N,t}} = 1 + p_{N,t}$, and recalling from Lemma \ref{l:BoundsK} that $\sinh_{k_{N,t}}, p_{N,t} \in L^2 (\bR^3 \times \bR^3)$, with a norm bounded uniformly in $N$, we notice that $\mathcal{L}_{2,N}^{(V)}$ is a sum of terms, each of them having the form of one of the operators $A_{i,1}, A_{i,2}, A^*_{i,2}, A_{i,3}$, for an $i=1,2$. Since the solution $\phn_t$ of (\ref{eq:NLSN1}) satisfies $\| \phn_t \|_{H^2} < C e^{K|t|}$, the bounds (\ref{eq:Aj123}) in Lemma \ref{lm:L2hat} immediately imply that 
\[ \pm \cL^{(V)}_{2,N} (t)\leq C e^{K|t|} (\cN+1) \]
and (since by the assumptions on $\ph$ we have $\| \dot{\ph}^N_t \|_{H^2} \leq C e^{K |t|}$) also that 
\[ \begin{split} \pm [\cN, \cL^{(V)}_{2,N} (t) ] &\leq C e^{K|t|} (\cN+1), \qquad \quad \pm \dot{\cL}^{(V)}_{2,N} \leq  C e^{K|t|} (\cN+1) \, . \end{split} \]

{F}rom (\ref{eq:Aj123-2}), on the other hand, we find
\[ \begin{split} (\cL^{(V)}_{2,N})^2 &\leq C e^{K|t|} (\cN+1)^2 \\  (\dot{\cL}^{(V)}_{2,N})^2 &\leq C e^{K|t|} (\cN+1)^2 \\ \pm [\cN^2, \cL^{(V)}_{2,N} (t)] &\leq C e^{K|t|} (\cN+1)^2 \end{split} 
\] 
\end{proof}

\subsection{Analysis of $T_{N,t}^* \cG^{(3)}_N (t) T_{N,t}$}

We have
\[ \begin{split} 
T_t^* \cG_{N}^{(3)} T_t = \; & 
\frac{1}{\sqrt{N}} \int dx dy N^{3\beta} V(N^\beta (x-y)) \phn_t (y) \\ &\hspace{2cm} \times (a^* (c_x^{N}) + a(s_x^N)) (a^* (c_y^N) + a (s_y^N)) (a(c_x^{N}) + a^* (s_x^N)) \\ 
&+ \frac{1}{\sqrt{N}} \int dx dy N^{3\beta} V(N^\beta (x-y)) \bar{\ph}^N_t (y) \\ &\hspace{2cm} \times (a^* (c_x^{N}) + a(s_x^N)) (a (c_y^N) + a^* (s_y^N)) (a(c_x^{N}) + a^* (s_x^N))  
\end{split} \]
Writing the terms in normal order and decomposing 
\[ \langle s_x^N, c_y^N \rangle = \bar{s}_{N,t} (y,x) + \langle s_x^N, p_y^N \rangle = \bar{k}_{N,t} (y,x) + \bar{r}_{N,t} (y,x) + \langle s_x^N , p_y^N \rangle \]
we arrive at 
\begin{equation}\label{eq:G3N} T_t^* \cG_{N}^{(3)} T_t =  \frac{1}{\sqrt{N}} \int dx dy N^3 V(N(x-y)) \phn_t (y) \bar{k}_{N,t} (x,y) T_t^* a_x T_t+\text{h.c.}  + \cE_{3,N} (t)
\end{equation}
with 
\begin{equation}\label{eq:cE3N} 
\begin{split} 
\cE_{3,N} (t) = \; &\frac{1}{\sqrt{N}} \int dx dy N^{3\beta} V(N^\beta (x-y)) \phn_t (y) \\ &\hspace{1cm} \times \left[ a^* (c_x^{N}) a^* (c_y^N) a^* (s_x^N) + a^* (c_x^{N}) a^* (c_y^N) a(c_x^{N}) + a^* (c_x^{N}) a^* (s_x^N) a(s_y^N) \right.  \\ &\hspace{1.5cm}  + a^* (c_y^N) a^* (s_x^N) a (s_x^N)  + a^* (c_x^{N}) a(s_y^N) a (c_x^{N})   + a^* (c_y^N) a (s_x^N) a(c_x^{N})  \\ &\hspace{1.5cm} \left. + a^* (s_x^N) a(s_x^N) a(s_y^N) + a(s_x^N) a(s_y^N) a(c_x^{N})  \right]   
\\ &+ \frac{1}{\sqrt{N}} \int dx dy \, N^{3\beta} V(N^\beta (x-y)) \phn_t (y) \langle s_y^N ,s_x^N \rangle  (a^* (c_x^{N}) + a(s_x^N))   \\ &+ \frac{1}{\sqrt{N}} \int dx dy \, N^{3\beta} V(N^\beta (x-y)) \phn_t (y) \| s_x^N \|^2 (a^* (c_y^N) + a (s_y^N)) 
\\ &+ \frac{1}{\sqrt{N}} \int dx dy N^{3\beta} V(N^\beta (x-y)) \phn_t (y) \left[ \langle s_y^N, p_x^N \rangle + \bar{r}_{N,t} (x,y)  \right]\\ &\hspace{8cm} \times  (a (c_x^{N}) + a^* (s_x^N)) \\ &+ \text{h.c.}  
\end{split} 
\end{equation}
The properties of $\cE_{3,N} (t)$ are established in the next proposition. 
\begin{prop}\label{prop:cL3N}
Under the assumptions of Theorem \ref{thm:LN2}, we find, for every $\delta > 0$, a constant $C_\delta > 0$ such that 
\begin{equation}\label{eq:cL3N1} \begin{split} \pm \cE_{3,N} (t) &\leq \delta \cV_N + C \cN^2 / N + C_\delta e^{K|t|} (\cN+1) \\
\pm [\cN , \cE_{3,N} (t)] &\leq \delta \cV_N + C \cN^2 / N  + C_\delta e^{K|t|} (\cN+1) \\
\pm \dot{\cE}_{3,N} (t) &\leq \delta \cV_N + C \cN^2 /N + C_\delta e^{K|t|} (\cN+1)
\end{split}
\end{equation}
Moreover, we have
\begin{equation}\label{eq:cL3N2} 
\begin{split}
\left| \langle \psi_1, \cE_{3,N} (t) \psi_2 \rangle \right| &\leq \frac{Ce^{K|t|}}{N^{(1-\beta)/2}} \left[  \langle \psi_1, (\cK+\cN + 1) \psi_1 \rangle + \langle \psi_2, (\cK^2 + (\cN+1)^2 ) \psi_2 \rangle \right] 
\end{split} 
\end{equation}
for all $\psi_1, \psi_2 \in \cF$.
\end{prop}

After applying the Cauchy-Schwarz inequality, we will estimate the cubic terms in $\cE_{3,N} (t)$ by quadratic and quartic terms. The quadratic terms can be controlled with Lemma \ref{lm:L2hat}. To bound the quartic terms, we will need the following three lemmas.  
\begin{lemma}\label{lm:LN4}
Let $V \in L^1 (\bR^3)$, $V \geq 0$ and $0 < \beta < 1$. Let $j_1, j_2 \in L^2 (\bR^3 \times \bR^3)$ with 
\[ M_i := \max \left[ \sup_x \int dy |j_i (x,y)|^2 , \sup_y \int dx |j_i (x,y)|^2 \right] < \infty \]
for $i=1,2$. Then we have
\[ \begin{split} \int dx dy N^{3\beta} V(N (x-y)) \| a^\sharp (j_{1,x}) a^\sharp (j_{2,y}) \psi \|^2 &\leq C \min (M_1 \|j_2 \|^2_2 , M_2 \| j_1 \|_2^2 ) \| (\cN+1) \psi \|^2 \\
\int dx dy N^{3\beta} V(N (x-y)) \| a^\sharp (j_{1,x}) a_y \psi \|^2 &\leq C M_1 \| (\cN+1) \psi \|^2 \end{split} \]
for all $\psi \in \cF$. These inequalities remain true if both operators applied to $\psi$ on the l.h.s. act on the same variable, i.e.  
\[ \begin{split} 
\int dx dy N^{3\beta} V(N (x-y)) \| a^\sharp (j_{1,x}) a^\sharp (j_{2,x}) \psi \|^2 &\leq C \min (M_1 \|j_2 \|^2_2 , M_2 \| j_1 \|_2^2 ) \| (\cN+1) \psi \|^2 
\\
\int dx dy N^{3\beta} V(N (x-y)) \| a^\sharp (j_{1,x}) a_x \psi \|^2 &\leq C M_1 \| (\cN+1) \psi \|^2 \end{split} \]
for all $\psi \in \cF$. 
\end{lemma}

\begin{proof}
We observe that 
\[ \begin{split} \int dx dy N^{3\beta}& V(N (x-y)) \| a^\sharp (j_{1,x}) a^\sharp (j_{2,y}) \psi \|^2 \\
&\leq \int dx dy N^{3\beta} V(N (x-y))\|j_{1,x}\|_2^2 \|j_{2,y}\|_2^2\| (\cN+1)\psi \|^2 \\
&\leq C \min (M_1 \|j_2 \|^2_2 , M_2 \| j_1 \|_2^2 ) \| (\cN+1) \psi \|^2 \\ \end{split} \]
and that 
\[ \begin{split}\int dx dy N^{3\beta}& V(N (x-y)) \| a^\sharp (j_{1,x}) a_y \psi \|^2  \\
&\leq \int dx dy N^{3\beta} V(N (x-y))\|j_{1,x}\|_2^2\| a_y\cN^{1/2}\psi \|^2 \\
&\leq C M_1  \| (\cN+1) \psi \|^2 \\ \end{split} \]
The last two bounds are obtained similarly.
\end{proof}

In the next lemma, we control quartic terms where the arguments of creation and annihilation operators is the kernel $\cosh_{k_{N,t}}$. 
\begin{lemma}\label{lm:LN4-cc}
Let $0 < \beta < 1$, $V \in L^1 (\bR^3)$, $V \geq 0$, and $\cV_N$ be defined as in (\ref{eq:cVN}). Let $k_{N,t}$ be defined as in (\ref{eq:kNt}) and, as usual, let $c_x^{N} (y) =\cosh_{k_{N,t}} (y,x)$. Then, we have
\[ \frac{1}{2N} \int dx dy N^{3\beta} V(N^\beta (x-y)) a^* (c_x^{N}) a^* (c_y^N) a (c_y^N) a(c_x^{N}) = \cV_N + \widetilde{\cE} (t) \]
where, for every $\delta > 0$, we can find a constant $C_\delta > 0$ such that  
\[ \widetilde{\cE} (t) \leq \delta \cV_N + C_\delta  (\cN+1)^2 / N  \]
\end{lemma}
\begin{proof}
The lemma follows expanding $ a(c^N_x)=a_x+ a(p^N_x)$ and $ a(c^N_y)=a_y+ a(p^N_y)$ and applying Lemma \ref{lm:LN4} and Lemma \ref{l:BoundsK}.
\end{proof}

Finally, we need a bound for the expectation of $\cV_N$ in terms of the square of the kinetic energy operator. This is the content of the next lemma.
\begin{lemma}\label{lm:VNK2}
Let $V\in L^1 (\bR^3)$, $V \geq 0$ and $V$ be radially symmetric, $0 < \beta < 1$. Then 
\[ \int dx dy N^{3\beta} V(N^\beta (x-y)) \| a_x a_y \psi \|^2 \leq C \int dx dy \| \nabla_x a_x \nabla_y a_y \psi \|^2 + C \int dx dy \| \nabla_x a_x a_y \psi \|^2 \]
\end{lemma}
\begin{proof}
We define 
\begin{equation}\label{eq:deff} g (x) =-\frac{1}{4\pi}\int dy \frac{1}{|x-y|} V(y) 
\end{equation}
Then we have $\Delta g = V$ and therefore 
\[ \begin{split} 
\int dx &dy N^{3\beta} V(N^\beta (x-y)) \| a_x a_y \psi_2 \|^2 \\ = \; &\int dx dy N^{\beta} \nabla_x \cdot \nabla_y \left[ g (N^\beta (x-y)) \right] \langle a_x a_y \psi_2 , a_x a_y \psi_2 \rangle \\ = \; &\int dx dy N^\beta g (N^\beta (x-y)) \left[ \langle \nabla_x a_x \nabla_y a_y \psi_2, a_x a_y \psi_2 \rangle + \langle \nabla_x a_x a_y , a_x \nabla_y a_y \psi_2 \rangle  + \text{h.c.} \right] \\ \leq \; &N^\beta \left[ \int dx dy \| \nabla_x a_x \nabla_y a_y \psi_2 \|^2 \right]^{1/2}  \left[ \int dx dy |g (N^\beta (x-y))|^2 \| a_x a_y \psi_2 \|^2 \right]^{1/2} \\ &+  
N^\beta \int dx dy |g (N^\beta (x-y))| \| \nabla_x a_x a_y \psi_2 \|^2 \\  \leq \; &C \int dx dy \| \nabla_x a_x \nabla_y a_y \psi_2 \|^2  +  C \int dx dy \| \nabla_x a_x a_y \psi_2 \|^2 \end{split} \] 
Here we used the fact that (from Newton's theorem) $|g(x)| \leq C |x|^{-1}$.
\end{proof}

\begin{proof}[Proof of Proposition \ref{prop:cL3N}]
Let us decompose the operator $\cE_{3,N} (t)$ as 
\[ \cE_{3,N} (t) =  L (t) + C_1 (t) + C_2 (t)  \]
with 
\begin{equation}\label{eq:LC1C2}\begin{split} 
L (t) = & \frac{1}{\sqrt{N}} \int dx dy N^{3\beta} V( N^\beta (x-y)) \phn_t (y) \langle s_y^N, s_x^N \rangle (a^* (c_x^{N}) + a(s_x^N))  \\
&+ \frac{1}{\sqrt{N}} \int dx dy N^{3\beta} V( N^\beta (x-y)) \phn_t (y) \| s_x^N \|^2  (a^* (c_y^N) + a(s_y^N))  \\
&+ \frac{1}{\sqrt{N}} \int dx dy N^{3\beta} V( N^\beta (x-y)) \phn_t (y) ( \langle s_y^N, p_x^N \rangle + \bar{r}_{N,t}(x,y)) (a (c_x^{N}) + a^*(s_x^N)) \\ &+ \text{h.c.} \\
C_1 (t) = \; &\frac{1}{\sqrt{N}} \int dx dy N^{3\beta} V(N^\beta (x-y)) \phn_t (y) \\ & \times \left[ a^* (c_x^{N}) a^* (s_x^N) a(s_y^N) + a^* (c_y^N) a^* (s_x^N) a(s_x^N) + a^* (c_x^{N}) a(s_y^N) a(c_x^{N}) \right. \\  &\hspace{.4cm} \left. + a^* (c_y^N) a(s_x^N) a(c_x^{N}) + a^* (s_x^N) a(s_x^N) a (s_y^N) + a(s_x^N) a(s_y^N) a(c_x^{N}) \right] + \text{h.c.} \\
C_2 (t) = \; & \frac{1}{\sqrt{N}} \int dx dy N^{3\beta} V(N^\beta (x-y)) \phn_t (y) \\ &\hspace{1cm} \times\left( a^* (c_x^{N}) a^* (c_y^N) a^* (s_x^N) + a^* (c_x^{N}) a^* (c_y^N) a (c_x^{N}) \right) + \text{h.c.} \end{split} \end{equation} 

Let us begin with the linear terms in $L (t)$. Using Cauchy-Schwarz inequality and  Lemma \ref{lm:L2hat} (together with the estimates in Appendix \ref{s:kernel} for the Hilbert-Schmidt operators $\sinh_{k_{N,t}}, p_{N,t}, r_{N,t}$), it is easy to check that
\[ \Big| \langle \psi_1, L(t) \psi_2 \rangle \Big|
\leq C N^{-1/2} e^{K|t|} \left( \langle \psi_1, (\cN+1)  \psi_1\rangle  + \langle \psi_2 , (\cN +1) \psi_2 \rangle \right) \]
and, taking $\psi_1 = \psi_2$, that 
\[ \begin{split} 
\pm L (t) &\leq C N^{-1/2} e^{K|t|} (\cN+1) \\
\pm [ \cN, L(t) ] &\leq C N^{-1/2} e^{K|t|} (\cN+1), \qquad \pm \dot{L} (t) \leq C N^{-1/2} e^{K|t|} (\cN+1) \end{split} \]

Next, we control the terms in $C_1 (t)$. Decomposing $\cosh_{k_{N,t}} = 1 + p_{N,t}$ and applying Cauchy-Schwarz, it is clear that all  contributions to $\langle \psi_1, C_1 (t) \psi_2 \rangle$ can be bounded by a sum of terms of the form
\[ \begin{split}
\frac{1}{\sqrt{N}} &\left[ \int dx dy N^{3\beta} V(N^\beta (x-y)) |\phn_t (y)|^2 \| a^\sharp (j_{1,x}) \psi_1 \|^2 \right]^{1/2} \\ &\hspace{4cm} \times \left[ \int dx dy N^{3\beta} V(N^\beta (x-y)) \| a^\sharp (j_{2,x}) a (j_{3,y}) \psi_2 \|^2 \right]^{1/2} \, ,
\\
\frac{1}{\sqrt{N}} &\left[ \int dx dy N^{3\beta} V(N^\beta (x-y)) |\phn_t (y)|^2 \| a^\sharp (j_{1,x}) \psi_1 \|^2 \right]^{1/2} \\ &\hspace{4cm} \times \left[ \int dx dy N^{3\beta} V(N^\beta (x-y)) \| a^\sharp (j_{2,x}) a_y \psi_2 \|^2 \right]^{1/2} \, ,
\\
\frac{1}{\sqrt{N}} &\left[ \int dx dy N^{3\beta} V(N^\beta (x-y)) |\phn_t (y)|^2 \| a_x \psi_1 \|^2 \right]^{1/2} \\ &\hspace{4cm} \times \left[ 
 \int dx dy N^{3\beta} V(N^\beta (x-y)) \| a^\sharp (j_{2,x}) a^\sharp (j_{3,y}) \psi_2 \|^2 \right]^{1/2} \, ,
\\
\frac{1}{\sqrt{N}} &\left[ \int dx dy N^{3\beta} V(N^\beta (x-y)) |\phn_t (y)|^2 \| a_x \psi_1 \|^2 \right]^{1/2} \\ &\hspace{4cm} \times \left[ 
 \int dx dy N^{3\beta} V(N^\beta (x-y)) \| a^\sharp (j_{2,x}) a_y \psi_2 \|^2 \right]^{1/2} 
\end{split} \]
with appropriate $j_1, j_2, j_3 \in L^2 (\bR^3 \times \bR^3)$ satisfying $\| j_i \|_2 \leq C e^{K|t|}$ and 
\[ M_i = \max \left[ \sup_x \int dy |j_i (x,y)|^2 , \sup_y \int dx |j_i (x,y)|^2 \right] \leq C e^{K|t|} \]
for all $i=1,2,3$. Using Lemma \ref{lm:LN4}, we conclude therefore that 
\be \begin{split} 
|\langle \psi_1, C_1 (t) \psi_2 \rangle | &\leq C N^{-1/2} e^{K|t|} \| (\cN+1)^{1/2} \psi_1 \| \| (\cN+1) \psi_2 \| \\
& \leq C N^{-1/2} e^{K|t|} \left[ \langle \psi_1, (\cN + 1) \psi_1 \rangle + \langle \psi_2, (\cN+1)^2 \psi_2 \rangle \right] 
\label{eq:C1-bd} \end{split} \ee
Taking $\psi_1 = \psi_2$ the first inequality shows that, for all $\delta > 0$, there exists a constant $C_\delta > 0$ with  
\[ \begin{split} \pm C_1 (t) &\leq \delta \, \cN^2 /N + C_\delta e^{K|t|} (\cN+1) \\
\pm [\cN, C_1 (t) ] &\leq  \delta \, \cN^2 /N + C_\delta e^{K|t|} (\cN+1) \\
\pm \dot{C}_1 (t) &\leq \delta \, \cN^2 /N + C_\delta e^{K|t|} (\cN+1)  \end{split} \] 

Finally, we study the term $C_2 (t)$. Using Cauchy-Schwarz, we obtain   
\[ \begin{split} \Big| \frac{1}{\sqrt{N}} \int dx dy &N^{3\beta} V (N^\beta (x-y)) \phn_t (y) \langle \psi, a^* (c_x^{N}) a^* (c_y^N) a(c_x^{N}) \psi \rangle \Big| \\  \leq \; &\frac{\kappa}{N} \int dx dy  N^{3\beta} V(N^\beta (x-y)) \| a (c_x^{N}) a(c_y^N) \psi \|^2 
\\ &+ \frac{1}{\kappa} \int dx dy N^{3\beta} V(N^\beta (x-y)) |\phn_t (y)|^2 \| a (c_x^{N}) \psi \|^2 
\end{split} \]
for every $\kappa > 0$. With Lemma \ref{lm:LN4-cc}, we find, for every $\delta > 0$, a constant $C_\delta > 0$ such that 
\[\begin{split}  \pm \frac{1}{\sqrt{N}} \int dx dy N^{3\beta} V (N^\beta (x-y)) & \left[ \phn_t (y) a^* (c_x^{N}) a^* (c_y^N) a(c_x^{N}) +\text{h.c.} \right]  \\ &\hspace{2cm} \leq \delta \, \cV_N + C\cN^2 / N + C_\delta e^{K|t|} (\cN +1) \end{split} \]
Similarly, one can show that, for all $\delta > 0$, there exists $C_\delta > 0$ such that 
\[ \begin{split} 
\pm \frac{1}{\sqrt{N}} \int dx dy N^{3\beta} V (N^\beta (x-y)) &\left[ \phn_t (y) a^* (c_x^{N}) a^* (c_y^N) a^* (s_x^N) + \text{h.c.} \right] \\ &\hspace{2cm} \leq \delta \cV_N + C \cN^2 / N +  C_\delta e^{K|t|} (\cN+1)  \end{split} \]
We conclude that, for all $\delta > 0$ there exists $C_\delta > 0$ with 
\[ \begin{split} \pm C_2 (t) &\leq \delta \cV_N + C \cN^2 / N + C_\delta e^{K|t|} (\cN+1) \\
\pm \left[ \cN, C_2 (t) \right] &\leq \delta \cV_N + C\cN^2/N + C_\delta e^{K|t|} (\cN+1) \\
 \pm \dot{C}_2 (t) &\leq \delta \cV_N + C\cN^2/N + C_\delta e^{K|t|} (\cN+1)  \end{split} \]

To prove (\ref{eq:cL3N2}) for the term $C_2 (t)$, we need to proceed differently. We observe that, for any $\psi_1, \psi_2 \in \cF$, 
\begin{equation}\label{eq:C2i}
\begin{split} \Big| \frac{1}{\sqrt{N}} \int dx &dy N^{3\beta} V(N^\beta (x-y)) \phn_t (y) \langle \psi_1, a^* (c_x^{N}) a^* (c_y^N) a(c_x^{N}) \psi_2 \rangle \Big|   \\  \leq \; &\left| \frac{1}{\sqrt{N}} \int dx dy N^{3\beta} V(N^\beta (x-y)) \phn_t (y) \langle \psi_1, a^* (c_x^{N}) a^*_y a(c_x^{N})\psi_2 \rangle \right| \\ &+ \left| \frac{1}{\sqrt{N}} \int dx dy N^{3\beta} V(N^\beta (x-y)) \phn_t (y) \langle \psi_1, a^* (c_x^{N}) a^* (p_y^N) a(c_x^{N})\psi_2 \rangle \right| \end{split} \end{equation}
To bound the second term, we observe that
\[ \begin{split} 
\Big| \frac{1}{\sqrt{N}} &\int dx dy N^{3\beta} V(N^\beta (x-y)) \phn_t (y) \langle \psi_1 , a^* (c_x^{N}) a^* (p_y^N) a(c_x^{N}) \psi_2 \rangle \Big| \\ \leq \; &\frac{1}{\sqrt{N}} \int dx dy N^{3\beta} V(N^\beta (x-y)) |\phn_t (y)| \| p_y^N \|_2 \| a(c_x^{N}) \psi_1 \| \|a (c_x^{N}) (\cN+1)^{1/2} \psi_2 \| \\ \leq \; &\frac{Ce^{K|t|}}{\sqrt{N}} \left[ \langle  \langle \psi_1, \cN \psi_1 \rangle + \langle \psi_2 , (\cN+1)^2 \psi_2 \rangle \right] \end{split} \]
As for the first term on the r.h.s. of (\ref{eq:C2i}),
we recall the definition (\ref{eq:deff}) of the function $g$ with $\Delta g = V$. Putting $v = \nabla g$, we find that $\nabla \cdot v = V$. Hence
\[ \begin{split}
 \frac{1}{\sqrt{N}} \int dx & dy N^{3\beta} V(N^\beta (x-y)) \phn_t (y) \langle \psi_1, a^* (c_x^{N}) a^*_y a(c_x^{N})\psi_2 \rangle \\ = \; & \frac{1}{\sqrt{N}} \int dx dy N^{2\beta} \nabla_y \cdot \left[ v(N^\beta (x-y)) \right] \phn_t (y) \langle \psi_1, a^* (c_x^{N}) a^*_y a(c_x^{N}) \psi_2 \rangle  \\ = \; & \frac{1}{\sqrt{N}} \int dx dy N^{2\beta} v (N^\beta (x-y)) \cdot \nabla \phn_t (y) \langle \psi_1, a^* (c_x^{N}) a^*_y a(c_x^{N}) \psi_2 \rangle \\ &+ \frac{1}{\sqrt{N}} \int dx dy \phn_t (y) N^{2\beta} v (N^\beta (x-y)) \cdot \langle a (c_x^{N}) \nabla_y a_y \psi_1, a(c_x^{N}) \psi_2 \rangle
\end{split} \]
Therefore,
\begin{equation}\label{eq:C2ii} \begin{split} 
\Big| \frac{1}{\sqrt{N}} &\int dx dy N^{3\beta} V(N^\beta (x-y)) \phn_t (y) \langle \psi_1, a^* (c_x^{N}) a^*_y a(c_x^{N})\psi_2 \rangle \Big| \\ \leq \; &\frac{1}{\sqrt{N}} \int dx dy N^{2\beta} |v (N^\beta (x-y)) | \| a(c_x^{N}) (\cN+1)^{1/2} \psi_2 \| \\ & \times \left[ |\nabla \phn_t (y)| \| a(c_x^{N})  (\cN+1)^{-1/2} a_y \psi_1 \| + |\phn_t (y)| \| a (c_x^{N}) (\cN+1)^{-1/2} \nabla_y a_y \psi_1 \| \right]
\end{split} 
\end{equation}
The contribution of the first term in the parenthesis can be estimated by  
\[ \begin{split} \int dx &dy N^{2\beta} |v (N^\beta (x-y)) | |\nabla \phn_t (y)| \| a(c_x^{N})  (\cN+1)^{-1/2} a_y \psi_1 \| \| a(c_x^{N}) (\cN+1)^{1/2} \psi_2 \| \\ \leq \; &C e^{K|t|} \,  N^{2\beta} \left[ \int dx dy \| a(c_x^{N})  (\cN+1)^{-1/2} a_y \psi_1 \|^2 \right]^{1/2}\\ &\hspace{3cm} \times  \left[ \int dx dy |v (N^\beta (x-y))|^2 \| a(c_x^{N}) (\cN+1)^{1/2} \psi_2 \|^2 \right]^{1/2}  \\ \leq \; &C e^{K|t|} N^{\beta/2} \| \cN^{1/2} \psi_1 \| \| (\cN+1) \psi_2 \| 
\end{split} \]
Here we used the bounds $\| v \|_\infty \leq C \| V \|_3$ and $\| v \|_2 \leq C \| V \|_{6/5}$, that can be proven with the help of the Hardy-Littlewood-Sobolev inequality. 

As for the second term in the parenthesis on the r.h.s. of (\ref{eq:C2ii}), we have
\[ \begin{split} 
\int dx dy N^{2\beta} | v (N^\beta &(x-y))| |\phn_t (y)| \| a (c_x^{N}) (\cN+1)^{-1/2} \nabla_y a_y \psi_1 \| \| a(c_x^{N}) (\cN+1)^{-1/2} \psi_2 \| \\ \leq \; & C e^{K|t|} N^{2\beta} \left[ \int dx dy \| a (c_x^{N}) (\cN+1)^{-1/2} \nabla_y a_y \psi_1 \|^2 \right]^{1/2} \\ &\hspace{2cm} \times  \left[ \int dx dy |v (N^\beta (x-y))|^2 \| a (c_x^{N}) (\cN+1)^{1/2} \psi_2 \|^2 \right]^{1/2} \\ \leq \; & C e^{K|t|} N^{\beta/2} \left[ \langle \psi_1, \cK \psi_1 \rangle + \langle \psi_2, (\cN+1)^2 \psi_2 \rangle \right] \end{split} \]
We find that
\begin{equation}\label{eq:C2iii} \begin{split} 
\Big| \frac{1}{\sqrt{N}} \int dx &dy N^{3\beta} V(N^\beta (x-y)) \phn_t (y) \langle \psi_1, a^* (c_x^{N}) a^* (c_y^N) a(c_x^{N})\psi_2 \rangle \Big| \\ &\leq C e^{K|t|}  N^{(\beta-1)/2} \left[ \langle \psi_1, (\cK + \cN + 1) \psi_1 \rangle + \langle \psi_2, (\cN+1)^2 \psi_2 \rangle \right]  \end{split} 
\end{equation}
Since the other term entering in the definition of $C_2(t)$ in (\ref{eq:LC1C2}) can be bounded similarly and the hermitian conjugates can be treated using Lemma 3.10, we conclude that
\[ \left| \langle \psi_1, C_2 (t) \psi_2 \rangle \right| \leq C e^{K|t|} N^{(\beta-1)/2} \left[ \langle \psi_1, (\cK + \cN + 1) \psi_1 \rangle + \langle \psi_2, (\cK + \cN + 1)^2 \psi_2 \rangle \right] 
\]
\end{proof}

\subsection{Analysis of $T_{N,t}^* \cG^{(4)}_N (t) T_{N,t}$}

Finally, we consider the conjugation of the quartic terms. Expanding the products, and writing all 
terms in normal order, we obtain
\begin{equation}\label{eq:G4N} \begin{split} 2T_t^* &\cG_{N}^{(4)} (t) T_t \\ = \;&\frac{1}{N} \int dx dy \, N^{3\beta} V (N^\beta (x-y)) \left[ |\langle s_x^N , c_y^N \rangle|^2  +  |\langle s_x^N , s_y^N \rangle|^2   + \langle s_y^N , s_y^N \rangle \langle s_x^N ,s_x^N \rangle \right] \\
&+ \frac{1}{N} \int dx dy \, N^{3\beta} V(N^\beta (x-y)) \left[ \langle c_y^N, s_x^N \rangle a^* (c_x^{N}) a^* (c_y^N) + \langle s_x^N,c_y^N \rangle a(c_y^N) a(c_x^{N}) \right] \\ 
&+ 2\cV_N + \cE_{4,N} (t) \end{split} 
\end{equation}
where we defined $\cE_{4,N} (t) = \cE_{4,N}^{(2)} (t) + \cE_{4,N}^{(4)} (t)$ with 
\[\begin{split} 
\cE^{(2)}_{4,N} (t) = \; &\int dx dy \, N^2 V (N (x-y)) \\ & \times \Big[ 2 \langle s_y^N, s_y^N \rangle a^* (c_x^{N})  a^* (s_x^N) + 2 \langle s_y^N, s_y^N \rangle a (s_x^N) a (c_x^{N}) 
+ 2 \langle s_y^N, s_x^N \rangle a^* (c_x^{N})  a^* (s_y^N)   \\ &\hspace{.5cm}
+ 2\langle s_x^N , s_y^N \rangle a (s_y^N) a (c_x^{N}) + 2 \langle s_y^N, s_y^N \rangle   a^* (c_x^{N}) a (c_x^{N}) 
+ 2 \langle s_y^N , s_x^N \rangle a^* (c_x^{N}) a(c_y^N) 
  \\ &\hspace{.5cm} + 2 \langle s_x^N , s_y^N \rangle  a^* (s_x^N)  a (s_y^N) 
+ 2 \langle s_y^N , s_y^N \rangle  a^* (s_x^N)  a (s_x^N) +  \langle c_y^N , s_x^N \rangle  a^* (c_x^{N})   a (s_y^N)   \\ &\hspace{.5cm} + \langle s_x^N, c_y^N \rangle a^* (s_y^N) a (c_x^{N}) 
+  \langle s_x^N , c_y^N \rangle a^* (s_y^N) a^* (s_x^N)   + \langle c_y^N , s_x^N \rangle a (s_x^N) a (s_y^N)  
  \\ &\hspace{.5cm} + \langle s_x^N , c_y^N \rangle a^* (s_x^N) a(c_y^N)  + \langle c_y^N, s_x^N \rangle  a^* (c_y^N) a (s_x^N) \Big]
\end{split} \]
and 
\begin{equation}\label{eq:decoVN} 
\cE^{(4)}_{4,N} (t) = \; \cE_{4,N}^{(4,1)} (t) + \cE_{4,N}^{(4,2)} (t) 
\end{equation}
Here, the term  
\[ \begin{split} 2N \cE_{4,N}^{(4,1)} (t) = \; &\int dx dy N^{3\beta} V(N^\beta (x-y)) \\ & \times \Big[ a^* (p_x^N) a^* (c_y^N) a(c_y^N) a(c_x^{N})  + a_x^* a^* (p_y^N) a (c_y^N) a(c_x^{N}) \\ &\hspace{4cm} + a_x^* a_y^* a (p_y^N) a(c_x^{N}) + a_x^* a_y^* a_y a(p_x^N)  \Big] 
\end{split} \]
contains the contributions arising from $a^* (c_x^{N}) a^* (c_y^N) a(c_y^N) a(c_x^{N})$, after removing $\cV_N$, while  
\begin{equation}\label{eq:E442} \begin{split}
2&N\cE_{4,N}^{(4,2)} (t) \\ = & \int dx dy \, N^{3\beta} V (N^\beta (x-y)) \Big[ a^* (c_x^{N}) a^* (c_y^N) a^* (s_y^N) a^* (s_x^N) + a^* (c_x^{N}) a^* (c_y^N) a^* (s_y^N) a (c_x^{N})  \\ & + a^* (c_x^{N}) a^* (c_y^N) a^* (s_x^N) a(c_y^N)  + a^* (c_x^{N})  a^* (s_y^N) a^* (s_x^N) a (s_y^N) + a^* (c_x^{N}) a^* (s_y^N) a (s_y^N) a (c_x^{N})  \\ &  + a^* (c_x^{N})  a^* (s_x^N)  a (s_y^N) a(c_y^N) + a^* (c_x^{N}) a (s_y^N) a (c_y^N) a (c_x^{N})  + a^* (c_y^N) a^* (s_y^N) a^* (s_x^N)  a (s_x^N)  \\ & + a^* (c_y^N) a^* (s_y^N)  a (s_x^N) a (c_x^{N}) +  a^* (c_y^N)  a^* (s_x^N) a (s_x^N) a(c_y^N)  + a^* (c_y^N)  a (s_x^N)  a (c_y^N) a (c_x^{N})  \\ &  + a^* (s_y^N)   a^* (s_x^N) a (s_x^N) a (s_y^N) + a^* (s_y^N)  a (s_x^N) a (s_y^N) a (c_x^{N})  + a^* (s_x^N) a (s_x^N) a (s_y^N) a(c_y^N)   \\ &  + a (s_x^N) a (s_y^N) a (c_y^N) a (c_x^{N}) \Big] \end{split} 
\end{equation}

\begin{prop}\label{prop:G4N}
Under the assumptions of Theorem \ref{thm:LN2}, there exists, for all $\delta > 0$, a constant $C_\delta > 0$ such that 
\[ \begin{split} \pm \cE_{4,N} (t) &\leq \delta \cV_N +  C_\delta e^{K|t|} (\cN+1)^2 / N + C_\delta N^{-1+\beta} e^{K|t|}(\cN+1) \\
\pm \left[ \cN, \cE_{4,N} (t) \right] & \leq  \delta 
\cV_N +  C_\delta e^{K|t|} (\cN+1)^2 / N + C_\delta N^{-1+\beta} e^{K|t|}(\cN+1) \\
\pm \dot{\cE}_4 (t) &\leq  \delta \cV_N +  C_\delta e^{K|t|} (\cN+1)^2 / N + C_\delta N^{-1+\beta} e^{K|t|}(\cN+1)\end{split} \]
Furthermore,  
\begin{equation}\label{eq:VN12}  \begin{split} |\langle \psi_1 , \cE_{4,N} (t) \psi_2 \rangle | &\leq C N^{-(1-\beta)/2} e^{K|t|} \Big[ \langle \psi_1, \cV_N \psi_1 \rangle +  \langle \psi_1 , ((\cN + 1)^2 / N) \psi_1 \rangle \\ &\hspace{3cm} + \langle \psi_1 , (\cN +1)\psi_1 \rangle +  \langle \psi_2, (\cK + \cN+1)^2 \psi_2 \rangle \Big] 
\end{split} 
\end{equation}
for all $\psi_1, \psi_2 \in \cF$. 
\end{prop}

\begin{proof}
We start with the terms in $\cE_{4,N}^{(2)} (t)$. We write
\[ \cE_{4,N}^{(2)} (t) = \cE_{4,N}^{(2,1)} (t) + \cE_{4,N}^{(2,2)} (t) \]
with 
\[ \begin{split} \cE_{4,N}^{(2,1)} (t) = &\frac{1}{N} \int dx dy \, N^{3\beta} V (N^\beta (x-y)) \\ &\times \Big[ 2 \langle s_y^N, s_y^N \rangle a^* (c_x^{N})  a^* (s_x^N) + 2 \langle s_y^N, s_y^N \rangle a (s_x^N) a (c_x^{N}) 
+ 2 \langle s_y^N, s_x^N \rangle a^* (c_x^{N})  a^* (s_y^N)   \\ &\hspace{.5cm}
+ 2\langle s_x^N , s_y^N \rangle a (s_y^N) a (c_x^{N}) + 2 \langle s_y^N, s_y^N \rangle   a^* (c_x^{N}) a (c_x^{N}) 
+ 2 \langle s_y^N , s_x^N \rangle a^* (c_x^{N}) a(c_y^N) 
  \\ &\hspace{.5cm} + 2 \langle s_x^N , s_y^N \rangle  a^* (s_x^N)  a (s_y^N) + 2 \langle s_y^N , s_y^N \rangle  a^* (s_x^N)  a (s_x^N) \Big] \\
\cE_{4,N}^{(2,2)} (t) = &\frac{1}{N} \int dx dy \, N^{3\beta} V (N^\beta (x-y))  \Big[ \langle c_y^N , s_x^N \rangle  a^* (c_x^{N})   a (s_y^N)   \\ &\hspace{.5cm} + \langle s_x^N, c_y^N \rangle a^* (s_y^N) a (c_x^{N}) 
+  \langle s_x^N , c_y^N \rangle a^* (s_y^N) a^* (s_x^N)   + \langle c_y^N , s_x^N \rangle a (s_x^N) a (s_y^N)  
  \\ &\hspace{.5cm} + \langle s_x^N , c_y^N \rangle a^* (s_x^N) a(c_y^N)  + \langle c_y^N, s_x^N \rangle  a^* (c_y^N) a (s_x^N) \Big]
\end{split} 
\]
Using Lemma \ref{lm:L2hat} and the bounds for the Hilbert-Schmidt operator $\sinh_{k_{N,t}}$, we easily find that 
\[ \begin{split} \left| \langle \psi_1, \cE_{4,N}^{(2,1)} (t) \psi_2 \rangle \right| &\leq C N^{-1} e^{K|t|} \| (\cN+1)^{1/2} \psi_1 \| \| (\cN+1)^{1/2} \psi_2 \| \\ &\leq C N^{-1} e^{K|t|} \left[ \langle \psi_1, (\cN+1) \psi_1 \rangle + \langle \psi_2 , (\cN+1) \psi_2 \rangle  \right]  \end{split} \]
Taking $\psi_1 = \psi_2$, we immediately obtain  
\[ \begin{split} \pm \cE_{4,N}^{(2,1)} (t) &\leq C N^{-1} e^{K|t|} (\cN +1)  \\ 
\pm \left[ \cN, \cE_{4,N}^{(2,1)} (t) \right] &\leq C N^{-1} e^{K|t|} (\cN+1) \\
\pm \dot{\cE}_4^{(2,1)} (t) &\leq C N^{-1} e^{K|t|} (\cN+1) \end{split} \]

Let us now consider the terms in $\cE_{4,N}^{(2,2)} (t)$. They all contain the inner product $\langle c_x^{N}, s_y^N \rangle = k_{N,t} (x,y) + r_{N,t} (x,y) + \langle p_x^N, s_y^N \rangle$. Since $|r_{N,t}(y,x)|$ and $|\langle p_x^N, s_y^N \rangle|$ are bounded, uniformly in $N$, the contributions arising from these terms can be dealt with, as we did above, for the term $\cE_{4,N}^{(2,1)} (t)$. For the contributions proportional to $k_{N,t} (x,y)$, we use the bound $|k_{N,t} (x,y)| \leq C |x-y|^{-1}$. Paying the price of an additional factor $N^\beta$, we can deal with these terms as we did above, just replacing the potential $V(x)$ with $V(x)/|x|$ (which is still integrable, by assumption). We conclude that
\[ \begin{split} |\langle \psi_1, \cE_{4,N}^{(2,2)} (t) \psi_2 \rangle | &\leq C N^{-(1-\beta)} e^{K|t|} \| (\cN+1)^{1/2} \psi_1 \| \| (\cN+1)^{1/2} \psi_2 \| \\ &\leq C N^{-(1-\beta)} e^{K|t|} \left[ \langle \psi_1, (\cN+1) \psi_1 \rangle + \langle \psi_2 , (\cN+1) \psi_2 \rangle  \right]  \end{split} \]
As usual, we also obtain
\[ \begin{split} \pm \cE_{4,N}^{(2,2)} (t) &\leq C N^{-(1-\beta)} e^{K|t|}  (\cN +1) \\
\pm \left[ \cN , \cE_{4,N}^{(2,2)} (t) \right] &\leq C N^{-(1-\beta)} e^{K|t|} (\cN +1) \\
\pm \dot{\cE}_4^{(2,2)} (t) &\leq C N^{-(1-\beta)} e^{K|t|} (\cN +1) \end{split} \]

Finally, we consider the quartic terms in $\cE_{4,N}^{(4)} (t)$. Recall the decomposition (\ref{eq:decoVN}). {F}rom Lemma \ref{lm:LN4-cc}, we find for every $\delta > 0$ a constant $C_\delta > 0$ such that
\[ \begin{split} 
\pm \cE_{4,N}^{(4,1)} (t) &\leq \delta \cV_N  + C_\delta e^{K|t|} (\cN+1)^2 / N \\
\pm \dot{\cE}_4^{(4,1)} (t) &\leq \delta  \cV_N  + C_\delta e^{K|t|} (\cN+1)^2 / N \end{split} \]
(notice that $[\cN, \cE_{4,N}^{(4,1)}] =0$). In order to show that $\cE_{4,N}^{(4)}$ satisfies the estimate (\ref{eq:VN12}), we remark that, from Lemma~\ref{lm:VNK2},
\[ \begin{split} |\langle \psi_1, \cV_N \psi_2 \rangle | &\leq \langle \psi_1, \cV_N \psi_1 \rangle^{1/2} \langle \psi_2, \cV_N \psi_2 \rangle^{1/2} \\
&\leq \frac{1}{\sqrt{N}} \langle \psi_1, \cV_N \psi_1 \rangle + \frac{1}{\sqrt{N}} \int dx dy N^{3\beta} V(N^\beta (x-y))  \langle \psi_2 , a_x^* a_y^* a_y a_x  \psi_2 \rangle \\
&\leq \frac{1}{\sqrt{N}} \left[ \langle \psi_1, \cV_N \psi_1 \rangle +  \langle \psi_2, (\cK + \cN+1)^2 \psi_2 \rangle \right] 
\end{split} \]  
Analogously, with Lemma \ref{lm:LN4-cc}, we find 
 \[ \begin{split} \Big| \frac{1}{N} &\int dx dy N^{3\beta} V(N^\beta (x-y)) \langle \psi_1, a^* (c_x^{N}) a^* (c_y^N) a (c_y^N) a (c_x^{N}) \psi_2 \rangle \Big| 
\\ &\leq \left[ \langle \psi_1, \cV_N \psi_1 \rangle + \frac{Ce^{K|t|}}{N} \| (\cN+1) \psi_1 \|^2 \right]^{\frac{1}{2}}  \left[ \langle \psi_2, \cV_N \psi_2 \rangle + \frac{Ce^{K|t|}}{N} \| (\cN+1) \psi_2 \|^2 \right]^{\frac{1}{2}} \\
&\leq  \frac{Ce^{K|t|}}{\sqrt{N}} \Big[ \langle \psi_1, \cV_N \psi_1 \rangle+ \frac{1}{N} \| (\cN+1) \psi_1 \|^2 + \langle \psi_2, (\cK+\cN+1)^2 \psi_2 \rangle \Big]
\end{split} \]
The last two equations show that
\[ | \langle \psi_1, \cE_{4,N}^{(4,1)} (t) \psi_2 \rangle| \leq C N^{-1/2} e^{K|t|} \Big[ \langle \psi_1, \cV_N \psi_1 \rangle+ \frac{1}{N} \| (\cN+1) \psi_1 \|^2 + \langle \psi_2, (\cK+\cN+1)^2 \psi_2 \rangle \Big]  \]

We switch now to $\cE_{4,N}^{(4,2)}$, defined in (\ref{eq:E442}), which can be further decomposed as 
\[
\cE_{4,N}^{(4,2)} (t) = A_1 (t) + A_2 (t) + A_3 (t)  \]
with
\[ \begin{split}  A_1 (t) = \; & \frac{1}{N} \int dx dy \, N^{3\beta} V (N^\beta (x-y)) a^* (s_y^N) a^* (s_x^N) a (s_x^N) a (s_y^N) \\ \; &+ \frac{1}{N} \int dx dy \, N^{3\beta} V (N^\beta (x-y)) \\ &\times \Big[ a^* (c_x^{N})  a^* (s_y^N) a^* (s_x^N) a (s_y^N) + a^* (c_y^N) a^* (s_y^N) a^* (s_x^N)  a (s_x^N) \\ &\hspace{1cm} + a^* (s_y^N)  a (s_x^N) a (s_y^N) a (c_x^{N})  + a^* (s_x^N) a (s_x^N) a (s_y^N) a(c_y^N) \Big] \\ &+ \frac{1}{N} \int dx dy \, N^{3\beta} V (N^\beta (x-y)) \\ &\times \Big[ a^* (c_x^{N}) a^* (s_y^N) a (s_y^N) a (c_x^{N}) + a^* (c_x^{N})  a^* (s_x^N)  a (s_y^N) a(c_y^N)  \\ &\hspace{1cm} + a^* (c_y^N) a^* (s_y^N)  a (s_x^N) a (c_x^{N}) +  a^* (c_y^N)  a^* (s_x^N) a (s_x^N) a(c_y^N) \Big] , 
\end{split} \] 
and 
\[ \begin{split} A_2 (t) = \; & \frac{1}{N} \int dx dy \, N^{3\beta} V (N^\beta (x-y)) \\ & \hspace{1cm} \times \Big[ a^* (c_x^{N}) a^* (c_y^N) a^* (s_y^N) a^* (s_x^N)  + a (s_x^N) a (s_y^N) a (c_y^N) a (c_x^{N}) \Big] \\
 A_3 (t) = \; & \frac{1}{N} \int dx dy \, N^{3\beta} V (N^\beta (x-y)) \\ &\times \Big[ a^* (c_x^{N}) a^* (c_y^N) a^* (s_y^N) a (c_x^{N})  + a^* (c_x^{N}) a^* (c_y^N) a^* (s_x^N) a(c_y^N) \\ &\hspace{1cm} + a^* (c_x^{N}) a (s_y^N) a (c_y^N) a (c_x^{N})  + a^* (c_y^N)  a (s_x^N)  a (c_y^N) a (c_x^{N}) \Big] \end{split} 
\]
In $A_1 (t)$, we collected terms with $0$, $1$, and some of those with $2$ kernels $\cosh_{k_{N,t}}$ (namely the terms with $2$ kernels that can be separated by Cauchy-Schwarz). In $A_2 (t)$ we collected all other terms with $2$ factors $\cosh_{k_{N,t}}$. In $A_3 (t)$, we included all terms with $3$ factors $\cosh_{k_{N,t}}$. Terms in $A_1 (t)$ can be controlled with Lemma \ref{lm:LN4}. We find 
\[ \begin{split} |\langle \psi_1, A_1 (t) \psi_2 \rangle | &\leq C N^{-1} e^{K|t|} \| (\cN+1) \psi_1 \|  \| (\cN+1) \psi_2 \| \\ &\leq C N^{-1/2} e^{K|t|} \left[ \frac{1}{N} \| (\cN+1) \psi_1 \|^2 + \| (\cN+1) \psi_2 \|^2 \right] \end{split} \]
Taking $\psi_1 = \psi_2$, we get 
\[ \begin{split}  \pm A_1 (t) &\leq C e^{K|t|} (\cN+1)^2 / N 
\\ \pm \left[ \cN, A_1 (t) \right] &\leq Ce^{K|t|} (\cN+1)^2 / N \\
\pm \dot{A}_1 (t) &\leq C e^{K|t|} (\cN+1)^2 / N 
\end{split}
\]

Next, we consider terms in $A_2 (t)$. By Cauchy-Schwarz and Lemma \ref{lm:LN4-cc}, we have
\begin{equation}\label{eq:E42b} \begin{split} 
\Big| \frac{1}{N} \int dx &dy N^{3\beta} V(N^\beta (x-y)) \langle \psi_1, a^* (c_x^{N}) a^* (c_y^N) a^* (s_y^N) a^* (s_x^N) \psi_2 \rangle \Big| \\ \leq\; &\frac{\delta}{N} \int dx dy N^{3\beta} V(N^\beta (x-y)) \| a (c_x^{N}) a(c_y^N) \psi_1 \|^2 \\ &+ \frac{1}{\delta N} \int dx dy N^{3\beta} V(N^\beta (x-y)) \| a^* (s_x^N) a^* (s_y^N) \psi_2 \|^2 \\
\leq \;& 2\delta \langle\psi_1, \cV_N \psi_1 \rangle + \frac{C \delta e^{K|t|}}{N} \| (\cN+1) \psi_1 \|^2 + 
\frac{C e^{K|t|}}{\delta N} \| (\cN+1) \psi_2 \|^2 \end{split} \end{equation}
Hence, choosing $\delta = N^{-1/2}$, we find 
\[ \begin{split} \Big| \frac{1}{N} \int dx &dy N^{3\beta} V(N^\beta (x-y)) \langle \psi_1, a^* (c_x^{N}) a^* (c_y^N) a^* (s_y^N) a^* (s_x^N) \psi_2 \rangle \Big| \\ \leq \; &\frac{Ce^{K|t|}}{\sqrt{N}} \left[ \langle \psi_1, \cV_N \psi_1 \rangle + \frac{1}{N} \| (\cN+1) \psi_1 \|^2 + \langle \psi_2 , ( \cN + 1)^2 \psi_2 \rangle \right] \end{split} \]
The hermitian conjugated term can be bounded using Lemma \ref{lm:VNK2}; we get 
\[ \begin{split} \Big| \frac{1}{N} \int dx &dy N^{3\beta} V(N^\beta (x-y)) \langle \psi_1,  a (s_x^N) a (s_y^N) a (c_y^N) a (c_x^{N}) \psi_2 \rangle \Big| \\ \leq \; &\frac{Ce^{K|t|}}{\sqrt{N}} \left[ \frac{1}{N} \| (\cN+1) \psi_1 \|^2 + \langle \psi_2 , (\cK + \cN + 1)^2 \psi_2 \rangle \right] \end{split} \]
Hence, we find
\[\begin{split} |\langle \psi_1, A_2 (t) \psi_2 \rangle | &\leq \frac{Ce^{K|t|}}{\sqrt{N}} \Big[ \langle \psi_1, \cV_N \psi_1 \rangle + \frac{1}{N} \| (\cN+1) \psi_1 \|^2   + \langle \psi_2 , (\cK + \cN + 1)^2 \psi_2 \rangle \Big] \end{split} \]
Going back to (\ref{eq:E42b}) and choosing $\psi_1 = \psi_2$, we obtain that for every $\delta > 0$, there exists $C_\delta > 0$ with 
\[ \begin{split} 
 \pm A_2 (t) \leq & \delta \cV_N + C_\delta e^{K|t|} (\cN+1)^2 / N \\
\pm \left[ \cN , A_2 (t) \right] \leq &  \delta 
\cV_N + C_\delta e^{K|t|} (\cN+1)^2 / N  \\
\pm \dot{A}_2 (t) \leq & \delta \cV_N + C_\delta e^{K|t|} (\cN+1)^2 / N  \end{split} \]

Finally, we bound the term $A_3$. To this end, we observe that 
\[ \begin{split} 
\Big| \frac{1}{N} \int dx dy &N^{3\beta} V(N^\beta (x-y)) \langle \psi_1, a^* (c_x^{N}) a^* (c_y^N) a^* (s_y^N)  a (c_x^{N}) \psi_2 \rangle \Big| 
\\ \leq \; &\frac{\delta}{N} \int dx dy N^{3\beta} V(N^\beta (x-y)) \| a (c_x^{N}) a (c_y^N) \psi_1 \|^2 \\ &+ \frac{1}{\delta N} \int dx dy N^{3\beta} V(N^\beta (x-y)) \| a^* (s_y^N) a (c_x^{N}) \psi_2 \|^2 
\\ \leq \; &\frac{2\delta}{N} \int dx dy N^{3\beta} V(N^\beta (x-y)) \| a_x a_y \psi_1 \|^2 \\ &+ \frac{C\delta e^{K|t|}}{N} \| (\cN+1) \psi_1 \|^2 + \frac{Ce^{K|t|}}{\delta N} \| (\cN+1) \psi_2 \|^2 
\end{split} \]
Similarly, the complex conjugated term can be controlled by 
\[ \begin{split} 
\Big| \frac{1}{N} \int dx dy &N^{3\beta} V(N^\beta (x-y)) \langle \psi_1, a^* (c_x^{N}) a (s_y^N) a (c_x^{N})  a (c_y^N) \psi_2 \rangle \Big| \\  \leq \; &\delta \langle \psi_2, \cV_N \psi_2 \rangle + \frac{C \delta e^{K|t|}}{N} \| (\cN+1) \psi_2 \|^2 + \frac{C e^{K|t|}}{\delta N} \| (\cN+1) \psi_1 \|^2  \end{split} \]
Hence, we obtain that, for every $\delta_1, \delta_2 > 0$, 
\[ \begin{split} 
\Big| \langle \psi_1, A_3 (t) \psi_2 \rangle \Big| \leq \; &\frac{\delta_1}{N} \int dx dy N^{3\beta} V(N^\beta (x-y)) \| a_x a_y \psi_1 \|^2 \\ &+ \frac{\delta_2}{N} \int dx dy N^{3\beta} V(N^\beta (x-y)) \| a_x a_y \psi_2 \|^2 \\ &+ \frac{Ce^{K|t|}}{N} \left( \delta_2 + \frac{1}{\delta_1} \right) \| (\cN+1) \psi_2 \|^2 + \frac{Ce^{K|t|}}{N} \left( \delta_1 + \frac{1}{\delta_2} \right) \| (\cN+1) \psi_1 \|^2 \end{split}  \]
Choosing $\delta_1 = N^{-1/2}$ and $\delta_2 = N^{1/2}$, we find
\[ \begin{split} 
\Big| \langle \psi_1, A_3 (t) \psi_2 \rangle \Big| \leq \frac{Ce^{K|t|}}{\sqrt{N}} \left[ \langle \psi_1, \cV_N \psi_1 \rangle + \frac{1}{N} \langle \psi_1, (\cN+1)^2 \psi_1 \rangle + \langle \psi_2, (\cK + \cN+1)^2 \psi_2 \rangle \right] \end{split}  \]
If instead we set $\delta_1 = \delta_2 = \delta \leq 1$, we obtain, with $\psi_1 = \psi_2$, that for every $\delta > 0$ there exists $C_\delta > 0$ with 
\[ \begin{split} 
\pm A_3 (t) \leq \; &\delta\cV_N + C_\delta e^{K|t|} (\cN+1)^2 / N  \\
\pm \left[ \cN , A_3 (t) \right] \leq & \delta \cV_N + C_\delta e^{K|t|} (\cN+1)^2 / N  \\
\pm \dot{A}_3 (t) \leq &\delta  \cV_N + C_\delta e^{K|t|} (\cN+1)^2 / N \end{split} \]
\end{proof}

\subsection{Analysis of $(i \dpr_t T_{N,t}^*) T_{N,t}$}

We set
\[
B (t) = \frac 12 \int \di x \di y \Big(k_{N,t}(x,y) a^*_x a^*_y - \bar{k}_{N,t}(x,y) a_x a_y \Big)
\]
and
\[
\dot B(t) = \frac 12 \int \di x \di y \Big( \dot k_{N,t}(x,y) a^*_x a^*_y - \bar{ \dot k}_{N,t}(x,y) a_x a_y \Big)
\]
Then $T_{N,t} = \exp(B)$. Defining $\text{ad}_B^0(C)=C$ and $\text{ad}_B^{n+1}(C)=[B, \text{ad}_B^n(C)]$ we have  
\be \label{TTstar}
(i \dpr_t T^*_{N,t}) T_{N,t} = -\int_0^1 d\l e^{-\l B(t)} \dot{B}(t)e^{\l B(t)} = \sum_{n \geq 0} \frac{(-1)^{n+1}}{(n+1)!} \text{ad}_B^n(\dot B)
\ee
on the domain $\DD(\NN)$ of the number of particle operator (this can be shown as in \cite[Lemma 2.3]{BdOS12}).  
In the next lemma, we collect some bounds for the operators $\text{ad}_B^n (\dot{B})$, whose proof 
can be found in \cite[Lemma 6.9]{BdOS12}.
\begin{lemma} \label{lemma:fni}
For each $n \in \NNN$ there exist $f_{n,1}, f_{n,2} \in L^2(\RRR^3 \times \RRR^3)$ such that
\begin{align} \label{TTstar2} 
\text{ad}_B^n(\dot B) &= \frac 12 \int \di x \di y \Big(f_{n,1}(x,y)a^*_y a^*_x + f_{n,2}(x,y)a_x a_y \Big)\,, \qquad  \text{ for all $n$ even and} \non \\
\text{ad}_B^n(\dot B) &= \frac 12 \int \di x \di y \Big(f_{n,1}(x,y)a^*_x a_y + f_{n,2}(x,y)a_x a^*_y \Big)\,, \qquad \text{for all odd $n$}\,,
\end{align}
where
\be \label{fni}
\|f_{n,i}\|_{L^2 \times L^2} \leq 2^n \|k_{N,t}\|_2^n\, \|\dot k_{N,t}\|_2
\ee
for all $n \geq 0$ and $i=1,2$,
\be  \label{fnidot}
 \|\dot f_{n,i}\|_2 \leq \begin{cases}
\| \ddot k_{N,t} \|_2 & \text{if $n=0$} \\
4^n \|k_{N,t} \|_2^{n-1} \Big(\|\ddot k_{N,t}\|_2 \|k_{N,t}\|_2 + \|\dot k_{N,t}\|_2^2 \Big) & \text{if $n \geq 1$}\,,
\end{cases}
\ee
and
\be \label{fnidot2}
\begin{split} 
\int \di x |f_{n,i}(x,x)| &\leq 2^n \|k_{N,t} \|_2^n \|\dot k_{N,t}\|_2\,, \\  \int \di x |\dot f_{n,i}(x,x)| &\leq 4^n \|k_{N,t} \|_2^{n-1} \Big( \|\dot k_{N,t}\|_2^2 + \|\ddot k_{N,t} \|_2 \|k_{N,t}\| \Big) \end{split} 
\ee
for all $n \geq 1$.
\end{lemma}

Using the bounds from the last lemma, we can now control the operator $(i\partial_t T_{N,t}^*) T_{N,t}$, appearing in the generator $\cL_N (t)$.
\begin{prop}\label{prop:TstarT}
Under the assumptions of Theorem \ref{thm:LN2},
we have 
\[ \begin{split} 
\pm ( i\dpr_t T^*_{N,t}) T_{N,t} &\leq C e^{K|t|} (\cN + 1), \qquad \hspace{.7cm} (( i\dpr_t T^*_{N,t}) T_{N,t} )^2 \leq C e^{K|t|} (\cN +1)^2 \\
\pm \left[ \cN ,( i\dpr_t T^*_{N,t}) T_{N,t}   \right] &\leq  C e^{K|t|} (\cN +1), \qquad \pm [\cN^2 , ( i\dpr_t T^*_{N,t}) T_{N,t}  ] \leq C e^{K|t|} (\cN+1)^2
 \\
\pm \dpr_t [( i\dpr_t T^*_{N,t}) T_{N,t}]  &\leq C e^{K|t|} (\cN+1), \qquad \hspace{.2cm} \left[ \dpr_t[( i\dpr_t T^*_{N,t}) T_{N,t}]  \right]^2 \leq C e^{K|t|} (\cN+1)^2 \end{split} \]
\end{prop}

\begin{proof}
For any $f_1, f_2 \in L^2(\RRR^3 \times \RRR^3)$ with $\int \di x |f_2(x,x)|<\io$ we have 
\[
\Big| \big \langle \ps, \int \di x \di y \big( f_1(x,y) a^*_x a^*_y + f_2(x,y) a_x a_y \big) \ps  \big \rangle \Big| \leq \big(\|f_1\|_2 + \|f_2\|_2 \big) \bmedia{\ps, (\NN+1) \ps}\,,
\]
and
\begin{align*}
\Big| \big \langle \ps, \int \di x \di y \big( f_1(x,y) a^*_x a_y &+ f_2(x,y) a_x a^*_y \big) \ps  \big \rangle \Big| \non \\ & \leq \big(\|f_1\|_2 + \|f_2\|_2 \big) \bmedia{\ps, \NN \ps} + \int \di x |f_2(x,x)| \|\ps\|^2\,.
\end{align*}
{F}rom Lemma \ref{lemma:fni} and \eqref{TTstar}, we find therefore 
\begin{align}
| \langle \ps, &(i\dpr_t T^*_{N,t})T_{N,t} \ps \rangle | \\ &\leq \sum_{n\geq 0} \frac{1}{(n +1)!} | \bmedia{\ps, \text{ad}_B^{n}(\dot B) \ps}| \non \\
& \leq  \sum_{n \geq 0} \frac{(2 \|k_{N,t}\|_2)^{n}}{(n +1)!} \|\dot k_{N,t}\|_2  \|(\NN+1)^{1/2}\ps\|^2 + \sum_{n \geq 1} \frac{(2 \|k_{N,t}\|_2)^{n}}{(2n)!} \|\dot k_{N,t}\|_2  \|\ps\|^2  \non \\
& \leq e^{2 \|k_{N,t}\|_{2}} \|\dot k_{N,t}\|_2 \|(\NN+1)^{1/2}\ps\|^2 \non \\& 
\leq C e^{K|t|} \|(\NN+1)^{1/2}\ps\|^2
\end{align}
The bound
\[ \left| \langle \psi, [ \cN , (i\partial_t T_{N,t}^*) T_{N,t} ] \psi \rangle \right| \leq C e^{K|t|} \| (\cN+1)^{1/2} \psi \|^2 \] 
follows similarly, since the commutator with $\NN$ is just eliminating the contributions $\text{ad}^n_B(\dot B)$ for all odd $n$. 

To prove the bound for $[(i\partial_t T_{N,t}^*) T_{N,t}]^2$, we use again \eqref{TTstar}: 
\begin{equation}
\label{eq:cl0} 
\begin{split}
\bmedia{\ps, [(i \dpr_t T_{N,t}^*) T_{N,t}]^2 \ps} & = \sum_{n \geq 0}\,\sum_{m \geq 0}  
\frac{1}{(n+1)!}\, \frac{1}{(m+1)!}\,
\bmedia{\ps, \text{ad}_B^n(\dot B)  \text{ad}_B^m(\dot B) \ps} \\ &\leq \sum_{n,m} \frac{\| \text{ad}_B^n (\dot{B}) \psi \|}{(n+1)!} \frac{ \| \text{ad}_B^m (\dot{B}) \psi \|}{(m+1)!}
\end{split} \end{equation}
We claim that 
\begin{equation}\label{eq:cl1} \| \text{ad}_B^n (\dot{B}) \psi \|^2 = \left\langle \psi, \left[ \text{ad}_B^n (\dot{B}) \right]^2 \psi \right\rangle \leq C \left( \| f_{n,1} \|_2 + \| f_{n,2} \|_2 \right)^2 \langle \psi, (\cN+1)^2 \psi \rangle \end{equation}
for $n$ even and 
\begin{equation}\label{eq:cl2} \begin{split} \| \text{ad}_B^n (\dot{B}) \psi \|^2 =\; & \left\langle \psi, \left[ \text{ad}_B^n (\dot{B}) \right]^2 \psi \right\rangle \\ \leq \; &C \left( \| f_{n,1} \|_2 + \| f_{n,2} \|_2 \right)^2 \langle \psi, (\cN+1)^2 \psi \rangle + C \left[ \int dx |f_{n,2} (x,x) | \right]^2 \| \psi \|^2 \end{split} \end{equation}
for $n$ odd. The bound for $n$ even follows because
\begin{equation}\label{eq:ad2} \begin{split} \langle \psi, \left[ \text{ad}_B^n (\dot{B}) \right]^2 \psi \rangle = \; &\int dx dy dx' dy' \, f_{n,1} (x,y) f_{n,1} (x' , y') \langle \psi, a^*_x a^*_y a_{x'}^* a_{y'}^* \psi \rangle  \\  &+\int dx dy dx' dy' \, f_{n,2} (x,y) f_{n,2} (x', y') \langle \psi, a_x a_y a_{x'} a_{y'} \psi \rangle \\ &+ \int dx dy dx' dy' \, f_{n,1} (x,y) f_{n,2} (x' , y') a^*_x a^*_y a_{x'} a_{y'} \\ &+ \int dx dy dx' dy' \, f_{n,2} (x,y) f_{n,1} (x' , y') \langle \psi , a_x a_y a_{x'}^* a_{y'}^* \psi \rangle 
\end{split} 
\end{equation}
The absolute value of the first term is bounded by 
\[ \begin{split} \int dx dx' \, \| a_x a_{x'} \psi \| &\| a^* (f_{n,1,x}) a^* (f_{n,1,x'}) \psi \| \\ \leq \; &C \| (\cN+1) \psi \| \int dx dx' \| f_{n,1,x} \|_2 \| f_{n,1,x'} \|_2 \| a_x a_{x'} \psi \| \\ \leq \; &C \| f_{n,1} \|_2^2 \| (\cN+1) \psi \|^2 \end{split} \]
The contribution of the second term on the r.h.s. of (\ref{eq:ad2}) can be estimated similarly. The third term, on the other hand, is bounded by 
\[ \begin{split} 
\int dx dy dx' dy' |f_{n,1} (x,y) | &|f_{n,2} (x', y')| \| a_x a_y \psi \| \| a_{x'} a_{y'} \psi \| \\ \leq
\; &\int dx dy dx' dy' \left[ |f_{n,1} (x,y) |^2  \| a_{x'} a_{y'} \psi \|^2 +  |f_{n,2} (x', y')|^2 \| a_x a_y \psi \|^2 \right] \\ \leq \; &(\| f_{n,1} \|_2^2 + \| f_{n,2} \|^2) \| (\cN+1) \psi \|^2 
\end{split} \]
To control the last term on the r.h.s. of (\ref{eq:ad2}), we observe that
\[ \begin{split} \int dx dy dx' &dy' f_{n,2} (x,y) f_{n,1} (x' , y') \langle \psi , a_x a_y a_{x'}^* a_{y'}^* \psi \rangle 
\\ = \;& \int dx dy dx' dy' f_{n,2} (x,y) f_{n,1} (x' , y') \langle a_{x'} \psi , a_x a_y a_{y'}^* \psi \rangle  \\ &+ \int dx dy dy' f_{n,2} (x,y) f_{n,1} (x , y') \langle \psi, a_y  a_{y'}^* \psi \rangle \\ &+ \int dx dy dy' f_{n,2} (x,y) f_{n,1} (y , y') \langle \psi, a_x a^*_{y'} \psi \rangle \\ = \;& \int dx dy dx' f_{n,2} (x,y)  \langle a_{x'} \psi , a_x a_y a^* (f_{n,1,x'}) \psi \rangle  \\ &+ \int dx dy dy' f_{n,2} (x,y) f_{n,1} (x , y') \langle \psi, a_{y'}^* a_y \psi \rangle  + \| \psi \|^2 \int dx dy f_{n,1} (x,y) f_{n,2} (x,y)  
\\ &+ \int dx dy dy' f_{n,2} (x,y) f_{n,1} (y , y') \langle \psi, a^*_{y'} a_x \psi \rangle + \| \psi \|^2 \int dx dy f_{n,2} (x,y) f_{n,1} (y,x) \end{split} \]
This implies that 
\[ \begin{split} \Big| \int dx dy dx' &dy' f_{n,2} (x,y) f_{n,1} (x' , y') \langle \psi , a_x a_y a_{x'}^* a_{y'}^* \psi \rangle  \Big| 
\\ = \;& \int dx dy dx' |f_{n,2} (x,y)| \| a_{x'} (\cN+1)^{1/2} \psi \| \| a_x a_y a^* (f_{n,1,x'}) (\cN+1)^{-1/2} \psi \|   \\ &+ \int dx dy dy' |f_{n,2} (x,y)| |f_{n,1} (x , y')| \| a_y \psi \| \| a_{y'} \psi \| + \| \psi \|^2 \| f_{n,1} \|_2 \| f_{n,2} \|_2   
\\ &+ \int dx dy dy' |f_{n,2} (x,y)| |f_{n,1} (y , y')| \| a_{y'} \psi \| \| a_x \psi \|  + \| \psi \|^2 \| f_{n,2} \|_2 \| f_{n,1} \|_2  \\
\leq \; & C (\|f_{n,1} \|_2^2 + \| f_{n,2} \|_2^2) \| (\cN+1) \psi \|^2 \end{split} \]
and concludes the proof of (\ref{eq:cl1}). Eq. (\ref{eq:cl2}) can be shown analogously. Combined, they imply that 
\[ \| \text{ad}_B^n (\dot{B}) \psi \| \leq C e^{K|t|} \| (\cN+1) \psi \| \]
and thus, inserting in (\ref{eq:cl0}), that 
\[ \left[ (i\partial_t T_{N,t}^*) T_{N,t} \right]^2 \leq C e^{K|t|} (\cN+1)^2 \]

Let us now consider the bound for the commutator $[\cN^2, (i\partial_t T_{N,t}^*)T_{N,t}]$ in Proposition \ref{prop:TstarT}. We observe that
\[ \begin{split} 
\Big| \langle  \psi, \big[ \cN^2 , (i\partial_t T_{N,t}^*) T_{N,t} \big] \psi \rangle \Big| = \; & \left| \langle \psi , \cN [ \cN , (i\partial_t T_{N,t}^*) T_{N,t} ] \psi \rangle + \langle \psi, [ \cN , (i\partial_t T_{N,t}^*) T_{N,t} ] \cN \psi \rangle \right| 
\\ 
\leq \; &2 \sum_{n \text{ even}} \frac{1}{(n+1)!} \left| \langle \psi, \cN \text{ad}_B^n (\dot{B}) \psi \rangle \right| \\
\leq \; &C \sum_{n \text{ even}} \frac{1}{(n+1)!} (\| f_{n,1} \|_2 + \| f_{n,2} \|_2 ) \| (\cN+1) \psi \|^2 
\\
\leq \; &C e^{K|t|} \| (\cN+1) \psi \|^2 
\end{split} \]

Finally, we observe that the bounds involving the time-derivative of $(i\partial_t T_{N,t}^*)T_{N,t}$ in Proposition \ref{prop:TstarT} can be proven similarly, taking the time derivative of the expressions for $\text{ad}_B^n(\dot B)$, and using the bounds for $\|\dot f_{n,i}\|_2$ and $\int \di x |\dot f_{n,i}(x,x)|$ in  \eqref{fnidot} and \eqref{fnidot2} and, finally, using the estimate for $\|\ddot k_{N,t}\|_2$ proven in Appendix \ref{s:kernel}.
\end{proof}

\subsection{Proof of Theorem \ref{thm:LN2}}
\label{s:pf-thm3}

Finally, we combine all results of the previous subsections to prove Theorem \ref{thm:LN2}. 

{F}rom (\ref{eq:G0N}), (\ref{eq:G1N}), Prop. \ref{prop:quadr}, (\ref{eq:G3N}), (\ref{eq:G4N}) and (\ref{TTstar}), we conclude that 
\begin{equation}\label{eq:LNtfin} \begin{split} \cL_N (t) = &\; \eta_N (t) + (i\partial_t T_{N,t}^*)T_{N,t} +  \cL^{(2)}_{2,N} (t) + \cV_N + \cE_{3,N} (t) + \cE_{4,N} (t) \\ & + \text{A}_1 + \text{A}_2  \end{split} 
\end{equation}
with
\begin{equation} \label{eq:LNtfin2} 
\begin{split} \text{A}_1 =  \; &\sqrt{N} \, T_{N,t}^* \left[  a^* ((N^{3\beta} V(N^\beta .) \omega_{N,\ell} * |\phn_t|^2) \phn_t) + \text{h.c.} \right] T_{N,t} \\ 
&+ \frac{1}{\sqrt{N}} \int dx dy N^{3\b} V(N^\b(x-y)) \left[ \bar{\ph}^N_t (y) k_{N,t} (x,y) T_{N,t}^* a^*_x T_{N,t}  + \text{h.c.} \right]  \\ 
\text{A}_2 = \; & \int dx dy \, (-\Delta_x k_{N,t} (x,y)) a_x^* a_y^* + \int dx dy \, (-\Delta_x k_{N,t} (x,y)) a_x a_y \\  &+  \frac{1}{2} \int dx dy N^{3\beta} V(N^\beta (x-y)) \left[ \phn_t (x) \phn_t (y) a^*_x a^*_y + \bar{\ph}^N_t (x) \bar{\ph}^N_t (y) a_x a_y \right]  \\  &+ \frac{1}{2N} \int dx dy \, N^{3\beta} V(N^\beta (x-y)) \\ &\hspace{3.5cm} \times \left[ \langle c_y^N, s_x^N \rangle a^* (c_x^{N}) a^* (c_y^N) + \langle s_x^N,c_y^N \rangle a(c_y^N) a(c_x^{N}) \right]  
\end{split} 
\end{equation}
We observe that, by the definition of $k_{N,t}$, 
\[ \begin{split} 
\text{A}_1 = \; &\sqrt{N} \int dx dy N^{3\beta} V(N^\beta (x-y)) \phn_t (y) \o_{N,\ell}(x-y)\\
&\hspace{4cm}\times\left[ \overline{\ph}^N_t (y) \bar{\ph}^N_t (x) - \bar{\ph}^N_t ((x+y)/2)^2 \right] T_t^* a_x T_t  + \text{h.c.}  \end{split} \]
Using the bounds in Lemma \ref{lm:propomega}, we find 
\[ \begin{split} \Big\| \int dy N^{3\beta} V(N^\beta (.-y)) \phn_t (y) \o_{N,\ell}(.-y)\ph^N(y) &\left[ \overline{\ph}^N_t (y) \overline{\ph}_t^N (.) - \overline{\ph}^N_t ((.+y)/2)^2 \right] \Big\|_2 \\ &\hspace{3cm} \leq C N^{-1} e^{K|t|} \end{split}  \]
and therefore  
\begin{equation}\label{eq:b1}
  | \langle \psi_1, \text{A}_1  \psi_2 \rangle | \leq  C N^{-1/2} e^{K|t|} \| (\cN+1)^{1/2} \psi_1 \| \| (\cN+1)^{1/2} \psi_2 \| \end{equation}

Let us now consider the second term in (\ref{eq:LNtfin2}). Using the permutation symmetry of $\cF$, we find 
\begin{equation}\label{eq:canc1} \begin{split} 
\int dx dy &(\Delta_x k_{N,t}) (x,y) a_x^* a_y^* \\ = \; &\int dx dy (\Delta_y k_{N,t}) (x,y) a_x^* a_y^* = \frac{1}{2} \int dx dy \left[ (\Delta_x + \Delta_y) k_{N,t} (x,y) \right] a_x^* a_y^* \\ = \; &\frac{1}{2} \int dx dy \left[ \left( 2\Delta_{x-y} + \frac{1}{2} \Delta_{(x+y)/2} \right) k_{N,t} (x,y) \right] a_x^* a_y^* \\ = \; &-N \int dx dy (\Delta \omega_{N,\ell}) (x-y) (\phn_t ((x+y)/2))^2 a_x^* a_y^* \\ &- \frac{N}{2} \int dx dy \, \omega_{N,\ell} (x-y) \\ &\hspace{1.5cm} \times \left[ \phn_t ((x+y)/2) \Delta \phn_t ((x+y)/2) + |\nabla \phn_t ((x+y)/2)|^2 \right] a_x^* a_y^* \end{split} \end{equation}
Moreover, we have 
\begin{equation}\label{eq:canc2} 
\begin{split} \frac{1}{2} \int dx dy &N^{3\beta} V(N^\beta (x-y))  \phn_t (x) \phn_t (y) a_x^* a_y^* \\ &= \frac{1}{2} \int dx dy N^{3\beta} V(N^\beta (x-y)) \phn_t ((x+y)/2)^2 a_x^* a_y^* + B_1 \end{split} \end{equation}
where 
\[ B_1 = - \frac{1}{2} \int dx dy N^{3\beta} V(N^\beta (x-y)) \left[ \phn_t ((x+y)/2)^2 - \phn_t (x) \phn_t (y) \right] a_x^* a_y^* \]
and 
\[ \begin{split} 
\frac{1}{2N} \int dx &dy N^{3\beta} V(N^\beta (x-y)) \langle c_y^N, s_x^N \rangle a^* (c_x^N) a^* (c_y^N) \\ = \; &-\frac{1}{2} \int dx dy N^{3\beta} V(N^\beta (x-y)) \omega_{N,\ell} (x-y) (\phn_t ((x+y)/2))^2 a_x^* a_y^* +  B_2 \end{split} \]
with
\[ \begin{split} B_2 = \; &\frac{1}{2N} \int dx dy N^{3\beta} V(N^\beta (x-y)) [ \langle p_y^N , s_x^N \rangle + r_{N,t} (y,x)] a^* (c_x^N) a^* (c_y^N) \\ & - \frac{1}{2} \int dx dy N^{3\beta} V(N^\beta (x-y)) \omega_{N,\ell} (x-y) (\phn_t ((x+y)/2))^2 a^* (p_x^N) a^*_y \\ &- \frac{1}{2} \int dx dy N^{3\beta} V(N^\beta (x-y)) \omega_{N,\ell} (x-y) (\phn_t ((x+y)/2))^2 a^* (c_x^{N}) a^* (p_y^N) \end{split} \]
By the definition of $f_{N,\ell} = 1 - \omega_{N,t}$, we conclude that 
\begin{equation}\label{eq:A2} \begin{split} 
A_2 = &\; \frac{N}{2} \int dx dy \, \omega_{N,\ell} (x-y) \\ &\hspace{1.5cm} \times \left[ \phn_t ((x+y)/2) \Delta \phn_t ((x+y)/2) + |\nabla \phn_t ((x+y)/2)|^2 \right] a_x^* a_y^* \\
&+ N \lambda_{N,\ell} \int dx dy {\bf 1} (|x-y| \leq \ell)  (\phn_t ((x+y)/2))^2 a_x^* a_y^* \\ 
&+ B_1 + B_2  +  B_3  + \text{h.c.}  \end{split} \end{equation}
with 
\[ B_3 = -N \lambda_{N,\ell} \int dx dy \, \omega_{N,\ell} (x-y) (\phn_t ((x+y)/2))^2 a_x^* a_y^* \]

Next, we show that $B_1, B_2, B_3$ are small. To control $B_1$, we define $\ff_{N,t}$ through 
\[  |x-y|\,\ff_{N,t}(x,y) :=\phn_t (x) \phn_t (y)-\ph^2_t ((x+y)/2)\]
We set $W (x) = V(x) |x|$, we define 
\[ g (x) = -\frac{1}{4\pi} \int dy \frac{1}{|x-y|} W  (y) \]
and $v = \nabla g$; then we have $\nabla \cdot v = W$. 
Integrating by parts, we find 
\begin{equation}\label{eq:B1-bd} \begin{split} 
& |\langle \psi_1, B_1 \psi_2 \rangle|= \\
& \qquad =  \frac 1 {2N^\beta}\,\Big|\int \di x\,\di y\,N^{3\b} V(N^\b(x-y)) N^\b |x-y| \, \ff_{t,N}(x,y) \bmedia{\ps_1,a^*_x a^*_y\, \ps_2} \Big| \\
& \qquad \leq  \frac 1 {2N^\b} \int \di x\,\|a_x  \ps_1 \| \,\| a^*(N^{2\beta} v(N^\beta(x- \cdot)) \nabla_x \ff_{t,N}(x,\cdot) )\,  \ps_2 \|  \\
& \qquad \qquad +\frac 1 {2N^\b} \int \di x\,\| \nabla_x  a_x \ps_1\|\, \|a^*(N^{2\beta} v(N^\beta(x-\cdot)) \ff_{t,N}(x,\cdot))\,  \ps_2\|\end{split}
\end{equation}
Since $\| v \|_2 \leq C \| |.| V \|_{6/5}$ and $\| \phn_t \|_{H^2} \leq C e^{K|t|}$, we obtain
\[\begin{split}
\| N^{2\b} v(N^\b(x-y)) \ff_{t,N}(x,y)\|_2 & \leq C N^{\b/2} e^{K|t|}  \\
\|N^{2\b} v(N^\b(x-y)) \nabla_x \ff_{t,N}(x,y) \|_2 & \leq C N^{\b /2} e^{K|t|} \end{split} \]
{F}rom (\ref{eq:B1-bd}),  we conclude that
\[ |\langle \psi_1, (B_1 + B_1^*) \psi_2 \rangle| \leq  C N^{-\beta /2} e^{K|t|} \left[ \langle \psi_1, (\cN + \cK + 1) \psi_1 \rangle + \langle \psi_2 , (\cN + \cK + 1) \psi_2 \rangle \right] \]

To bound $B_2$, we notice first that
\[ \begin{split} \Big| \frac{1}{2N} \int dx dy N^{3\beta} &V(N^\beta (x-y)) \left[ \langle p_y^N, s_x^N \rangle + r_{N,t} (y;x) \right] \langle \psi_1, a^* (c_x^{N}) a^* (c_y^N) \psi_2 \rangle\Big| \\ &
 \leq  \frac{C e^{K|t|}}{ N^{1/2}} \left[ \bmedia{ \ps_1, \left[ \VV_N + (\cN+1)^2/N \right]\ps_1} + \langle \psi_2 , \psi_2 \rangle \right]   \end{split}  \]
as it easily follows from Lemma \ref{lm:LN4-cc}. Moreover, we find 
\[ \begin{split} 
\Big| \frac{1}{2} \int dx dy &N^{3\beta} V(N^\beta (x-y)) \omega_{N,\ell} (x-y) (\phn_t ((x+y)/2))^2 \langle \psi_1, a^* (p_x^N) a_y^* \psi_2 \rangle \Big| \\ \leq \; &\frac{e^{K|t|}}{2N} \int dx dy \, N^{3\beta} V(N^\beta (x-y))  \frac{1}{|x-y|} \| p_x^N \|_2 \, \| a_y \psi_1 \| \| (\cN+1)^{1/2} \psi_2 \| 
\\ \leq \; &C N^{\beta-1} e^{K|t|} \| (\cN+1)^{1/2} \psi_1 \| \| (\cN+1)^{1/2} \psi_2 \| \end{split} \]
and an analogous estimate for the term proportional to $a^* (c_x^N) a^* (p_y^N)$. Hence, we conclude that 
\[ \begin{split} 
|\langle \psi_1, &(B_2 + B_2^*) \psi_2 \rangle| \\  \leq \; &C \max \left[ N^{\beta-1}, N^{-1/2} \right] \, e^{K|t|}  \\ &\times \left[ \langle \psi_1, (\cV_N + (\cN+1)^2/N + 1) \psi_1 \rangle + \langle \psi_2 , (\cV_N + (\cN+1)^2/N + 1)\psi_2 \rangle \right] \end{split} \]

Finally, we control $B_3$. {F}rom the bounds in Lemma \ref{lm:propomega}, we immediately find that
\[ | \langle \psi_1, (B_3+B_3^*) \psi_2 \rangle | \leq C N^{-1}  \left[ \langle \psi_1, (\cN+1) \psi_1 \rangle + \langle \psi_2, (\cN+1) \psi_2 \rangle \right] \]

Combining (\ref{eq:LNtfin}) with (\ref{eq:b1}), (\ref{eq:A2}) and with the bounds for $B_1, B_2, B_3$, we find
\[ \cL_N (t) = \eta_N (t) + \cL_{2,N} (t) + \cV_N  + \cE_N (t) \]
with 
\begin{equation}\label{eq:new}
\begin{split}
\cL_{2,N} (t) = &\; (i\partial_tT_{N,t}^*)T_{N,t} + \cL^{(2)}_{2,N} (t) \\
&+\frac{N}{2} \int dx dy \, \omega_{N,\ell} (x-y) \\ &\hspace{.2cm} \times \left[( \phn_t ((x+y)/2) \Delta \phn_t ((x+y)/2) + |\nabla \phn_t ((x+y)/2)|^2)  a_x^* a_y^*+\text{h.c.}\right]\\
&+ N \lambda_{\ell, N} \int dx dy \, {\bf 1} (|x-y| \leq \ell) \left[ (\phn_t ((x+y)/2))^2 a_x^* a_y^* + \text{h.c.} \right] 
\end{split} \end{equation}
and where $\cE_N (t) = \cE_{N,3} (t) + \cE_{N,4} (t) + A_1 + B_1 + B_2 +B_3$ satisfies the estimates (\ref{eq:est-cE1}), (\ref{eq:est-cE2}) in Theorem \ref{thm:LN2}. 

To complete the proof of Theorem \ref{thm:LN2}, we still have to show that the last two terms on the r.h.s. of (\ref{eq:new}) satisfy the bounds (\ref{eq:est-cL2N}). But this fact follows easily, since, from Lemma \ref{lm:propomega}, $N\omega_{N,\ell} (x-y) \leq C {\bf 1} (|x-y| \leq \ell)/ |x-y|$ is square-integrable and $N\lambda_{N,\ell} \leq C$ is of order one.

\section{Growth of Number of Particles and Energy}

We apply the estimates of the Section \ref{s:est} to show a bound for the growth of the expectation of the number of particles with respect to the fluctuation dynamics generated by $\cL_N (t)$. Moreover, we prove bounds for the growth of the expectation of the number of particles, of the kinetic energy and of their square with respect to the dynamics generated by the quadratic part of the generator $\cL_{2,N} (t)$.


We need, first of all, to compare kinetic and potential energy operators with the generators $\cL_N (t)$ of the full fluctuation dynamics and with its quadratic part $\cL_{2,N}$. The following proposition follows directly from Theorem \ref{thm:LN2}. 
\begin{prop}\label{prop:compLKV}
Under the same assumptions as in Proposition \ref{prop:UN-UN2}, there exists a constant 
$C > 0$ such that 
\begin{equation}\label{eq:bd-LKVN} 
\begin{split}
\frac{1}{2} (\cK + \cV_N )- C e^{K|t|} \left[ \cN+1 + \cN^2/N \right] &\leq \cL_N (t) - \eta_N (t) \\ &\leq 2 (\cK +  \cV_N) + C e^{K|t|} \left[ \cN + 1 + \cN^2 /N \right]
\end{split}
\end{equation}
Moreover, we have 
\begin{equation}\label{eq:L2K}  \cK  - C e^{K|t|} (\cN+1) \leq \cL_{2,N} (t) \leq \cK + C e^{K|t|} (\cN + 1) \end{equation}
and 
\begin{equation}\label{eq:L2K2} \frac{1}{2} \cK^2 - C e^{K|t|} (\cN+1)^2 \leq \cL^2_{2,N} (t) \leq 2\cK^2 + C e^{K|t|} (\cN+1)^2 \end{equation}
\end{prop}

Next, we control the growth of the expectation of the number of particles and of the energy with respect to the fluctuation dynamics $\cU_N (t;s)$. 
\begin{theorem}\label{tm:gn}
Under the same assumptions as in Proposition \ref{prop:UN-UN2}, there exist constants $C,c_1, c_2 > 0$ such that 
\[ \begin{split} \langle \cU_N (t;0) \psi, \cN \cU_N (t;0) \psi \rangle &\leq C \exp (c_1 \exp (c_2 |t|)) \langle \psi , \left( \cN +  \cN^2/N + \cH_N \right) \psi \rangle \\
\left| \langle \cU_N (t;0) \psi, (\cL_N (t) - \eta_N (t)) \cU_N (t;0) \psi \rangle \right| &\leq C \exp (c_1 \exp (c_2 |t|)) \langle \psi , \left( \cN +  \cN^2/N + \cH_N \right) \psi \rangle 
\end{split}
\]
\end{theorem}
To prove Theorem \ref{tm:gn}, we use the following lemma, whose proof can be found in~$\cite{BdOS12}$.
\begin{lemma}\label{lm:unu}
 Let the fluctuation dynamics $\cU_N (t;s)$ be defined as in (\ref{eq:flucd}). Then there exists a constant $C > 0$ such that
 \[ \cU_N^* (t;0) \cN^2 \, \cU_N (t;0) \leq C \big( N  \, \cU_N^* (t;0) \cN \, \cU_N (t;0) + N (\cN+1) + (\cN+1)^2 \big).
 \]
\end{lemma}

\begin{proof}[Proof of Theorem \ref{tm:gn}]
From Theorem \ref{thm:LN2}, \eqref{eq:bd-LKVN} and Lemma \ref{lm:unu}, we obtain
\[ \begin{split} \frac{d}{dt} \langle \cU_{N} &(t;0) \psi, \cN \cU_{N} (t;0) \psi \rangle  \\ = \; & \langle \cU_{N} (t;0) \psi, [i\cN, \cL_{N} (t) ] \cU_{N} (t;0) \psi \rangle  \\ 
\leq \; &C e^{K|t|} \expv{\UU_N(t;0)\psi}{\NN\UU_N(t;0)\psi} + \langle \cU_N (t;0) \psi, (\cL_N (t) - \eta_N (t)) \, \cU_N (t;0) \psi \rangle \\ &+ C e^{K|t|} \left\langle \psi, \left[ (\cN+1) + \cN^2/N \right] \psi \right\rangle \end{split} \]
and
\[ \begin{split}  \frac{d}{dt} \langle \cU_{N} &(t;0) \psi, (\cL_N(t) - \eta_N (t)) \, \cU_{N} (t;0) \psi \rangle \\ = \; & \langle \cU_{N} (t;0) \psi, (\dot \cL_N(t) - \dot{\eta}_N (t))\, \cU_{N} (t;0) \psi \rangle \\
\leq \; &C e^{K|t|} \expv{\UU_N(t;0)\psi}{\NN\UU_N(t;0)\psi} + \langle \cU_N (t;0) \psi, (\cL_N (t) - \eta_N (t)) \, \cU_N (t;0) \psi \rangle \\ &+ C e^{K|t|} \left\langle \psi, \left[ (\cN+1) + \cN^2/N \right] \psi \right\rangle \end{split} \]
We conclude that
\[ \begin{split} 
\frac{d}{dt} & \left\langle \cU_N (t;0) \psi, \left[ C e^{K |t|} (\cN + 1) + (\cL_N (t) - \eta_N (t)) \right] \cU_N (t;0) \psi \right\rangle \\ &\leq D e^{K |t|} \left\langle \cU_N (t;0) \psi, \left[ C e^{K|t|} (\cN + 1) + (\cL_N (t) - \eta_N (t)) \right] \cU_N (t;0) \psi \right\rangle \\ &\quad + D e^{2K|t|} \left\langle \psi, \left[ (\cN+1) + \cN^2/N \right] \psi \right\rangle \end{split} \]
for appropriate constants $C,K$ (chosen so, that the operator $C\exp (K|t|) (\cN+1) + (\cL_N (t) - \eta_N (t))$ is non-negative) and $D > 0$. Gronwall's Lemma implies that 
\[ \begin{split}  &\left\langle \cU_N (t;0) \psi, \left[ C e^{K |t|} (\cN + 1) + (\cL_N (t) - \eta_N (t)) \right] \cU_N (t;0) \psi \right\rangle \\ &\hspace{3cm} \leq C \exp (c_1 \exp (c_2 |t|))  \left\langle \psi, \left[ (\cN+1) + \cN^2/N + \cH_N \right] \psi \right\rangle \end{split} \]
which completes the proof of the theorem.
\end{proof}

Finally, we control growth of number of particles, energy and their squares w.r.t. quadratic fluctuation dynamics. 
\begin{theorem}\label{thm:5}
Under the same assumptions as in Proposition \ref{prop:UN-UN2}, there exist constants $C,c_1,c_2 > 0$ such that 
\begin{equation}\label{eq:bds-U2} \begin{split} \langle \cU_{2,N} (t;0) \psi, \cN \cU_{2,N} (t;0) \psi \rangle &\leq C \exp (c_1 \exp (c_2 |t|)) \langle \psi, (\cN+1) \psi \rangle \\
\langle \cU_{2,N} (t;0) \psi, \cN^2 \cU_{2,N} (t;0) \psi \rangle &\leq C \exp (c_1 \exp (c_2 |t|)) \langle \psi, (\cN+1)^2 \psi \rangle  \\
\langle \cU_{2,N} (t;0) \psi, \cL_{2,N}^2 (t) \cU_{2,N} (t;0) \psi \rangle &\leq C \exp (c_1 \exp (c_2 |t|)) \langle \psi, (\cK + \cN +1)^2 \psi \rangle \end{split} \end{equation}
\end{theorem}

{\it Remark.} The last bound in (\ref{eq:bds-U2}) , combined with (\ref{eq:L2K2}), also implies that 
\[ \langle \cU_{2,N} (t;0) \psi, \cK^2 \, \cU_{2,N} (t;0) \psi \rangle \leq C \exp (c_1 \exp (c_2 |t|)) \langle \psi, (\cK + \cN +1)^2 \psi \rangle \]

\begin{proof}
{F}rom Theorem \ref{thm:LN2}, we obtain 
\[ \begin{split} \frac{d}{dt} \langle \cU_{2,N} (t;0) \psi, (\cN+1) \cU_{2,N} (t;0) \psi \rangle &= \langle \cU_{2,N} (t;0) \psi, [i\cN, \cL_{2,N} (t) ] \cU_{2,N} (t;0) \psi \rangle  \\ &\leq C e^{K|t|} \langle \cU_{2,N} (t;0) \psi, (\cN+1) \cU_{2,N} (t;0) \psi \rangle \end{split} \]
Gronwall's lemma implies that
\[ \langle \cU_{2,N} (t;0) \psi, \cN \cU_{2,N} (t;0) \psi \rangle \leq C \exp (c_1 \exp (c_2 |t|)) \langle \psi, (\cN+1) \psi \rangle \] 

The estimate for $\cN^2$ follows analogously. Finally, we compute, using again the bounds in Theorem \ref{thm:LN2},
\[ \begin{split} 
\frac{d}{dt} \,  \langle \cU_{2,N} &(t;0) \psi, \cL_{2,N}^2 (t) \cU_{2,N} (t;0) \psi \rangle \\
&= \langle \cU_{2,N} (t;0) \psi, \left[ \dot{\cL}_{2,N} (t) \cL_{2,N} (t) + \text{h.c.} \right] \cU_{2,N} (t;0) \psi \rangle \\ &\leq \langle \cU_{2,N} (t;0) \psi, \cL_{2,N}^2 (t) \cU_{2,N} (t;0) \psi \rangle + \langle \cU_{2,N} (t;0) \psi, \dot{\cL}_{2,N}^2 (t) \cU_{2,N} (t;0) \psi \rangle \\ &\leq \langle \cU_{2,N} (t;0) \psi, \cL_{2,N}^2 (t) \cU_{2,N} (t;0) \psi \rangle + C e^{K|t|} \langle \cU_{2,N} (t;0) \psi, (\cN+1)^2 \cU_{2,N} (t;0) \psi \rangle \\ &\leq \langle \cU_{2,N} (t;0) \psi, \cL_{2,N}^2 (t) \cU_{2,N} (t;0) \psi \rangle + C \exp (c \exp (c|t|)) \langle \psi, (\cN+1)^2 \psi \rangle \end{split} \]
Gronwall's lemma and equation \eqref{eq:L2K2} imply the last bound in (\ref{eq:bds-U2}). 
\end{proof}

With the help of Theorem \ref{tm:gn} and Theorem \ref{thm:5}, we can now conclude the proof of Proposition \ref{prop:UN-UN2}. 

\begin{proof}[Proof of Proposition \ref{prop:UN-UN2}] 
We have 
\[ \begin{split} \frac{d}{dt}&\left\| \cU_N (t;0) \xi_N - e^{-i \int_0^t \eta_N (s) ds} \, \cU_{2,N} (t;0) \xi_N \right\| ^2 \\
&=2\Im\scal{\UU_N(t,0)\xi_N}{ (\LL_N(t) -\LL_{2,N}(t)-\eta_N (t))\,e^{-i \int_0^t \eta_N (s) ds} \, \cU_{2,N} (t;0) \xi_N}  \\
&\leq C|\scal{\UU_N(t,0)\xi_N}{\mathcal{E}(t) \,\cU_{2,N} (t;0) \xi_N} |+C | \scal{\UU_N(t,0)\xi_N}{\cV_N \,\cU_{2,N} (t;0) \xi_N}| \\
&\leq C N^{-\alpha} e^{K|t|} \Big[ \langle \cU_N (t;0) \xi_N, (\cK + \cV_N + \cN + 1) \cU_N (t;0) \xi_N \rangle \\ &\hspace{4cm} + \langle \cU_{2,N} (t;0) \xi_N, (\cK^2 + \cN^2 + 1) \cU_{2,N} (t;0) \xi_N \rangle \Big] \end{split} \]
where $\alpha = \min (\beta/2, (1-\beta)/2)$ and where we used the bounds in Theorem \ref{thm:LN2} and Lemma \ref{lm:VNK2}. Applying Proposition \ref{prop:compLKV}, and then Theorems \ref{tm:gn} and \ref{thm:5}, we obtain 
\[ \begin{split} \frac{d}{dt}&\left\| \cU_N (t;0) \xi_N - e^{-i \int_0^t \eta_N (s) ds} \, \cU_{2,N} (t;0) \xi_N \right\| ^2 \\ & \leq C \exp (c_1 \exp (c_2 |t|)) N^{-\alpha} \langle \xi_N, (\cN^2 + \cK^2 + \cV_N) \xi_N \rangle \end{split} \]
Integrating over time, we conclude the proof of Proposition \ref{prop:UN-UN2}.
\end{proof}

\section{Comparison with the limiting fluctuation dynamics}    
\label{Uio}

Next, we prove Theorem \ref{thm:main2}. Proceeding similarly as in Proposition \ref{prop:UN-UN2}, we observe that 
\[ \begin{split} 
e^{-i\cH_N t} W(\sqrt{N} \ph) T_{N,0} \xi_N &- e^{-i \int_0^t \eta_N (s) ds} W(\sqrt{N} \phn_t) T_{N,t} \cU_{2,\infty} \xi_N \\ &= W(\sqrt{N} \phn_t) T_{N,t} \left[ \cU_N (t;0) \xi_N - e^{-i \int_0^t \eta_N (s) ds} \cU_{2,\infty} (t;0) \xi_N \right] \end{split} \]
{F}rom the result of Theorem \ref{thm:main1}, it is enough to compare $\cU_{2,N} (t;0) \xi_N$ with $\cU_{2,\infty} (t;0) \xi_N$. To this end, we need to compare the two generators $\cL_{2,N} (t)$ and $\cL_{2,\infty} (t)$ defined in (\ref{eq:L2N0}) and (\ref{eq:L2inf0}); we do so in the next four lemmas. 

\begin{lemma}\label{lm:k}
Recall the definition of $\cL_{2,N}^{(K)}$ and $\cL_2^{(K)} (t)$ given in (\ref{eq:LK}) and after (\ref{eq:L2inf0}). Under the same assumptions as in Theorem \ref{thm:main2}, with $\alpha = \min (\beta/2, (1-\beta)/2)$, we have 
\begin{equation}\label{eq:l1-K} 
\begin{split} &\left| \langle \psi_1, \left( \cL_{2,N}^{(K)} (t) - \cL_2^{(K)} (t) \right) \psi_2 \rangle \right| \\ &\hspace{2cm} \leq CN^{-\alpha} \exp (c_1 \exp (c_2 |t|)) \, \| (\cN+1)^{1/2} \psi_1 \| \, \| (\cN+1)^{1/2} \psi_2 \| \end{split} 
\end{equation}
\end{lemma}

\begin{proof}
{F}rom (\ref{eq:LK}), all contributions to the difference $\cL^{(K)}_{2,N} - \cL^{(K)}_2$ have the form
\begin{equation}\label{eq:A1A2} \begin{split} 
A_1 &= \int dx a^\sharp (m_{x}) a_x \qquad \text{or } \qquad A_1^* = \int dx a_x^* a^\sharp (m_x)
\\
A_2 &= \int dx a^\sharp (m_{x}) a^\sharp (j_{x}) \qquad \text{or } \qquad A_2^* = \int a^\sharp (j_x) a^\sharp (m_x)  
\end{split} 
\end{equation}
Here $m,j \in L^2 (\bR^3 \times \bR^3)$ with $\| j \|_2 \leq C e^{K|t|}$ and
\begin{equation}\label{eq:diff-bd} \| m \|_2 \leq C N^{-\alpha}  \exp (c_1 \exp (c_2 |t|)) 
\end{equation}
Eq. (\ref{eq:l1-K}) follows therefore from the bounds 
\[ \left| \langle \psi_1, A_1 \psi_2 \rangle \right| \leq C \| m \|_2 \| (\cN+1)^{1/2} \psi_1 \| \, \| (\cN+1)^{1/2} \psi_2 \| \]
and 
\[ \left| \langle \psi_1, A_2 \psi_2 \rangle \right| \leq C \| m \|_2 \| j \|_2 \| (\cN+1)^{1/2} \psi_1 \| \, \| (\cN+1)^{1/2} \psi_2 \| \]
which can be proven as in Lemma \ref{lm:L2K}.

When we consider the different terms in (\ref{eq:LK}), the kernel $m$ is a difference like $k_{N,t} - k_t$, $p_{N,t} - p_t$, $r_{N,t}- r_t$, $\nabla p_{N,t} - \nabla p_t$, $\Delta p_{N,t} - \Delta p_t$ or $\Delta r_{N,t} - \Delta r_{t}$. In these cases, (\ref{eq:diff-bd}) follows from the results of Appendix \ref{s:kernel}. 

Notice that also the contribution
\begin{equation}\label{eq:termB} B = \int dx \nabla_x a^* (k^{N}_x) \nabla_x a(k^{N}_x) - \int dx \nabla_x a^* (k_x) \nabla_x a(k_x) \end{equation}
can be written in the form $A_1$ in (\ref{eq:A1A2}), with $m= u_{N,t} - u$ and 
\[ \begin{split} u_{N,t} (x,y) &= \int dz \nabla_x k_{N,t} (x,z) \nabla_x \overline{k}_{N,t} (x,y) , \qquad u_t (x,y) = \int dz \nabla_x k_t (x,z) \nabla_x \overline{k}_t (x,y) \end{split} \]
To prove that (\ref{eq:diff-bd}) holds in this case, we 
estimate
\begin{align}
 \norml{u_{N,t}-u_t}{2}^2=&\int dx dy |u_{N,t}(x,y)-u_t(x,y)|^2 \non\\
 \leq \; &C \int dx dy \left|\int dz \nabla_z k_{N,t}(z,y)( \nabla_z\bar k_{N,t}(x,z)-\nabla_z\bar k_t(x,z))\right|^2 \label{eq:diffU}\\
  &+C\int dx dy \left|\int dz ( \nabla_z k_{N,t}(z,y)-\nabla_zk_t(z,y))\nabla_z \bar k_{t}(x,z)\right|^2 \non
 \end{align}
Integrating by parts in the first term, we obtain 
\[\begin{split}
  \int dx dy& \left|\int dz \nabla_z k_{N,t}(z,y)( \nabla_z\bar k_{N,t}(x,z)-\nabla_z\bar k_t(x,z))\right|^2 \\
\leq&\int dx dy \left|\int dz \,\Delta_z k_{N,t}(z,y)( \bar k_{N,t}(x,z)-\bar k_t(x,z))\right|^2 \\
 \leq&\int dx dydz_1dz_2 \,\left|\Delta_{z_1} k_{N,t}(z_1,y)\right|\left|\Delta_{z_2} k_{N,t}(z_2,y)\right|\\
 &\hspace{2cm}\times\left|  k_{N,t}(x,z_1)- k_t(x,z_1)\right|\left|  k_{N,t}(x,z_2)- k_t(x,z_2)\right| 
 \end{split}
\]
By Cauchy-Schwarz inequality, we find  
\begin{equation}\label{eq:CSuu}\begin{split}
 \int &dx dy \left|\int dz \nabla_z k_{N,t}(z,y)( \nabla_z\bar k_{N,t}(x,z)-\nabla_z\bar k_t(x,z))\right|^2 \\
  \leq & \; \int dxdz_1\left|  k_{N,t}(x,z_1)- k_t(x,z_1)\right|^2\sup_{z_1}\int dy \left|\Delta_{z_1} k_{N,t}(z_1,y)\right|\sup_{y}\int dz_2 \left|\Delta_{z_2} k_{N,t}(z_2,y)\right|\\
  \leq & \; Ce^{K|t|} \norml{k_{N,t}-k_t}{2}^2
 \end{split}
\end{equation}
In the last inequality, we observed that 
\[\begin{split}
 \Delta_{z_2} & k_{N,t}(z_2,y) \\ =&-N \Delta \omega_{N,\ell}(z_2-y) (\phn_t ((z_2+y)/2))^2 \\
 &-\frac{N}{2} \o_{N,\ell}(z_2-y) \left[ \phn_t ((z_2+y)/2) \Delta \phn_t ((z_2+y)/2)) + (\nabla \phn_t )^2 ((z_2+y)/2) \right] 
\\ &-2N\nabla\o_{N,\ell} (z_2 -y) \phn_t ((z_2+y)/2) \nabla \phn_t ((z_2+y)/2)
 \end{split}
\]
Using the scattering equation for $f_{N,\ell} = 1-\omega_{N,\ell}$, we get 
\[\begin{split}
 \Delta_{z_2} &k_{N,t}(z_2,y) \\= & \; \frac
{N^{3\beta}}{2} V(N^\beta(z_2-y))f_{N,\ell}(z_2-y)(\phn_t((z_2+y)/2))^2\\
&-N\lambda_{N,\ell}f_{N,\ell}(z_2-y)\mathbf{1}(|z_2-y|\leq\ell) (\phn_t((z_2+y)/2))^2 \\
 &-\frac{N}{2} \o_{N,\ell}(z_2-y)\left[ \phn_t ((z_2+y)/2) \Delta\phn_t((z_2+y)/2)) + (\nabla \phn_t)^2 ((z_2+y)/2) \right] \\ &-2N \nabla \o_{N,\ell} (z_2 - y) \phn_t ((z_2+y)/2) \nabla\phn_t ((z_2+y)/2) \end{split}
\]
With the bounds from Lemma \ref{lm:propomega} and Proposition \ref{prop:ph}, we conclude that
\[ \sup_{z_2} \int dy \, |\Delta k_{N,t} (z_2 , y)| \leq C \| \phn_t \|_{H^2}^2 \leq C e^{K|t|}  \]

The results of Appendix \ref{s:kernel}, combined with (\ref{eq:CSuu}), imply that
\[ \begin{split} \int dx dy &\left|\int dz \nabla_z k_{N,t}(z,y)( \nabla_z\bar k_{N,t}(x,z)-\nabla_z\bar k_t(x,z))\right|^2 \leq C N^{-2\alpha} \exp (c_1 \exp (c_2 |t|)) \end{split} \]
Proceeding analogously to bound the second term on the r.h.s. of (\ref{eq:diffU}), we obtain 
\[ \| u_{N,t} - u_t \|_2 \leq C \exp (c_1 \exp (c_2 |t|)) \left[ N^{-1+\beta} + N^{-\beta/2} \right] \]
Hence, as claimed, also the term (\ref{eq:termB}) can be written as the term $A_1$ in (\ref{eq:A1A2}), with the kernel $m=u_N - u$ satisfying (\ref{eq:diff-bd}).
\end{proof}

\begin{lemma}\label{lm:v}
Recall the definition of $\cL_{2,N}^{(V)}$ and $\cL_2^{(V)} (t)$ given in (\ref{eq:LV}) and after (\ref{eq:L2inf0}). Under the same assumptions as in Theorem \ref{thm:main2}, with $\alpha = \min (\beta/2, (1-\beta)/2)$, we have  
\begin{equation*}
\begin{split} 
&\left| \langle \psi_1, \left( \cL_{2,N}^{(V)} (t) - \cL_2^{(V)} (t) \right) \psi_2 \rangle \right| \\ &\hspace{2.5cm}  \leq C N^{-\alpha} \exp (c_1 \exp (c_2 |t|)) \| (\cN+1)^{1/2} \psi_1 \| \| (\cN+\cK+ 1)^{1/2} \psi_2 \| 
\end{split}
\end{equation*}
(In fact, the kinetic energy operator $\cK$ could also be placed on $\psi_1$.)
\end{lemma}

\begin{proof}
{F}rom (\ref{eq:LV}), we observe that all terms in $\cL_{2,N}^{(V)} (t)$ have one of the following forms:
\begin{equation}\label{eq:A11-A23} 
\begin{split} 
A_{1,1}^N &= \int dx dy N^{3\beta} V(N^\beta (x-y)) \phn_t (x) \phn_t (y)  a^\sharp (j^N_{1,x}) a^\sharp (j^N_{2,y}) \\
A_{1,2}^N &= \int dx dy N^{3\beta} V(N^\beta (x-y)) \phn_t (x) \phn_t (y)  a^\sharp (j^N_{1,x}) a_y  \\
A_{1,3}^N &= \int dx dy N^{3\beta} V(N^\beta (x-y)) \phn_t (x) \phn_t (y)  a^*_x a_y \\ 
A_{2,1}^N &= \int dx dy N^{3\beta} V(N^\beta (x-y)) \phn_t (x) \phn_t (x) a^\sharp (j^N_{1,y}) a^\sharp (j^N_{2,y}) \\
A_{2,2}^N &= \int dx dy N^{3\beta} V(N^\beta (x-y)) \phn_t (x) \phn_t (x)   a^\sharp (j^N_{1,y}) a_y  \\
A_{2,3}^N &= \int dx dy N^{3\beta} V(N^\beta (x-y)) \phn_t (x) \phn_t (x) a^*_y a_y \end{split} \end{equation}
(or possibly, the form of the adjoint of $A_{1,2}^N$ or $A_{2,2}^N$). Here $\phn_t$ denotes the solution of the $N$-dependent nonlinear Schr\"odinger equation (\ref{eq:NLSN1}), and $j^N_1, j^N_2 \in L^2 (\bR^3 \times \bR^3)$ being either the operator $\sinh_{k_{N,t}}$ or $p_{N,t} = \cosh_{k_{N,t}} - 1$. (In fact, some of the $\ph_t^N$ factors should be replaced by $\bar \ph^N_t$, but this does not affect our analysis).

To estimate the difference $\cL_{2,N}^{(V)} (t) - \cL_2^{(V)} (t)$, we need to compare the terms in (\ref{eq:A11-A23}) with their formal limits:
\[ \begin{split} 
A_{1,1} &= b_0 \int dx dy \, \delta (x-y) \, \ph_t (x) \ph_t (y)  a^\sharp (j_{1,x}) a^\sharp (j_{2,y}) \\
A_{1,2} &= b_0 \int dx dy \, \delta (x-y) \ph_t (x) \ph_t (y)  a^\sharp (j_{1,x}) a_y  \\
A_{1,3} &= b_0 \int dx dy \,\delta (x-y) \ph_t (x) \ph_t (y)  a^*_x a_y \\ 
A_{2,1} &= b_0\int dx dy  \d (x-y) \ph_t (x) \ph_t (x) a^\sharp (j_{1,y}) a^\sharp (j_{2,y}) \\
A_{2,2} &= b_0\int dx dy \d(x-y) \ph_t (x) \ph_t (x)   a^\sharp (j_{1,y}) a_y  \\
A_{2,3} &= b_0 \int dx dy \d(x-y)\ph_t (x) \ph_t (x) a^*_y a_y \end{split} \]
where $\ph_t$ is the solution of the limiting nonlinear Schr\"odinger equation (\ref{eq:NLS0}) and $j_1, j_2 \in L^2 (\bR^3 \times \bR^3)$ are either $\sinh (k_t)$ or $p_t = \cosh_{k_t} - 1$, with $k_t$ given by (\ref{eq:kell}). Note that, from the results of Appendix \ref{s:kernel}, we always have 
\begin{equation}\label{eq:jj}
\| j_i - j_i^N \|_2 \leq C N^{-\alpha} \exp (c_1 \exp (c_2 |t|)) 
\end{equation}
for $i=1,2$. The lemma will follow if we can prove that 
\[\begin{split} &\left| \langle \psi_1,  (A_{i,j}^N - A_{i,j}) \psi_2 \rangle \right| \\ &\hspace{2.5cm} \leq C N^{-\alpha} \exp (c_1 \exp (c_2 |t|)) \| (\cN + 1)^{1/2} \psi_1 \| \, \| (\cN+\cK+1)^{1/2} \psi_2 \|    \end{split} \]
for all $i=1,2$ and $j=1,2,3$. 

We start comparing $A_{1,1}^N$ with $A_{1,1}$. To this end, we observe that 
\begin{equation}
\label{6.70}
\begin{split}
A_{1,1}^N &- A_{1,1} \\ &=  \int \di x \di y N^{3\b} V(N^\b(x-y))\phn_t(x) \phn_t(y) a^\sharp(j_{1,x}^N) a^\sharp(j_{2,y}^N -j_{2,y}) \\
& +\int \di x \di y N^{3\b} V(N^\b(x-y)) \phn_t(x) \phn_t(y) a^\sharp(j_{1,x}^N - j_{1,x}) a^\sharp(j_{2,y})    \\
&+ \int \di x \di y \Big[N^{3\b} V(N^\b(x-y)) - b_0 \delta(x-y)\Big]\phn_t(x) \phn_t(y) a^\sharp(j_{1,x}) a^\sharp(j_{2,y})  \\
& +  \int \di x \di y\,  b_0 \d(x-y) \Big(\phn_t(x) \phn_t(y) -\ph_t(x) \ph_t(y)\Big) a^\sharp(j_{1,x}) a^\sharp(j_{2,y})  
\end{split}
\end{equation}
The contribution from the first term can be bounded by 
\begin{equation}\label{eq:term1} \begin{split} 
\Big| \int \di x &\di y \, N^{3\b} V(N^\b(x-y))\phn_t(x) \phn_t(y) \langle \psi_1, a^\sharp (j_{1,x}^N) a^\sharp(j^N_{2,y} -j_{2,y}) \ps_2 \rangle \Big| \\
&\leq C \| \phn_t\|_{H^2}^2 \, \|j_1^N\|_2 \, \|j_2^N -j_2\|_2\, \|(\NN+1)^{1/2} \ps_1\|^2  \|(\NN+1)^{1/2} \ps_2\| \\
& \leq C N^{-\alpha} \exp(c_1 \exp(c_2|t|)) \|(\NN+1)^{1/2}\ps_1\| \|(\NN+1)^{1/2}\ps_2\| 
\end{split}\end{equation}
Also the contribution from the second term can be bounded similarly. The fourth term on the r.h.s. of (\ref{6.70}), on the other hand, is bounded by 
\begin{equation}\label{eq:term4} \begin{split} 
\int \di x  \Big| ( &\phn_t(x))^2 - \ph_t^2(x) \Big| \,  \left| \langle \psi_1,  a^\sharp(j_{1,x}) a^\sharp(j_{2,x})  \ps_2 \rangle \right|  \\
& \leq C \big(\|\phn_t\|_{\io} + \|\ph_t\|_{\io} \big) \|\phn_t - \ph_t\|_\io \int \di x \,\|a^\sharp(j_{1,x})\ps_1\|\,\|a^\sharp(j_{2,x})  \ps_2\| \\
& \leq  C N^{-\alpha} \exp(c_1 \exp(c_2|t|))\, \|(\NN+1)^{1/2}  \ps_1\|  \|(\NN+1)^{1/2} \ps_2\| 
\end{split} \end{equation}
Finally, we have to bound the contribution arising from the third term on the r.h.s. of (\ref{6.70}). 
We have
\[ \begin{split} 
\int dx dy \, &(N^{3\beta} V(N^\beta (x-y)) - b_0 \delta (x-y)) \phn_t (x) \phn_t (y) \langle \psi_1, a^\sharp (j_{1,x}) a^\sharp (j_{2,y}) \psi_2 \rangle \\ = \; &\int dx dy V(y) \Big[ \phn_t (x) \phn_t (x+y/N^\beta)  \langle \psi_1, a^\sharp (j_{1,x}) a^\sharp (j_{2,x+y/N^\beta}) \psi_2 \rangle \\ &\hspace{6cm} - 
\phn_t (x)^2  \langle \psi_1, a^\sharp (j_{1,x}) a^\sharp (j_{2,x}) \psi_2 \rangle \Big] 
\end{split} \]
and thus
\begin{equation}\label{eq:term3} \begin{split} 
\Big| \int &dx dy \, (N^{3\beta} V(N^\beta (x-y)) - b_0 \delta (x-y)) \phn_t (x) \phn_t (y) \langle \psi_1, a^\sharp (j_{1,x}) a^\sharp (j_{2,y}) \psi_2 \rangle \Big| \\ \leq \; &\int dx dy V(y) |\phn_t (x)| \, |\phn_t (x+y/N^\beta)  - \phn_t (x)| \|  a^\sharp (j_{1,x}) \psi_1 \| \| a^\sharp (j_{2,x+y/N^\beta}) \psi_2 \|  \\ &+ \int dx dy V(y) |\phn_t (x)|^2 
\| a^\sharp (j_{1,x}) \psi_1 \| \| a^\sharp (j_{2,x+y/N^\beta} - j_{2,x}) \psi_2 \| 
\end{split} \end{equation}
The first term can be easily controlled with Proposition \ref{prop:ph}. We find  
\[  \begin{split} \int dx dy V(y) |\phn_t (x)| & \, |\phn_t (x+ y/N^\beta)  - \phn_t (x)| \|  a^\sharp (j_{1,x}) \psi_1 \| \| a^\sharp (j_{2,x+y/N^\beta}) \psi_2 \| \\ &\leq C N^{-\beta} \exp (c_1 \exp (c_2 |t|)) \| (\cN+1)^{1/2} \psi_1 \| \| (\cN+1)^{1/2} \psi_2 \| \end{split} \]
As for the second term, we estimate it by 
\[ \begin{split} 
\int dx dy &V(y) |\phn_t (x)|^2 
\| a^\sharp (j_{1,x}) \psi_1 \| \| a^\sharp (j_{2,x+y/N^\beta} - j_{2,x}) \psi_2 \| \\  &\leq C N^{-\beta} e^{K|t|} \| (\cN+1)^{1/2} \psi_1 \| \, \| (\cN+1)^{1/2} \psi_2 \|  \int dx dy \\ &\hspace{3cm} \times \int_0^1 d\lambda V(y) |y| |\phn_t (x)|^2 \|j_{1,x} \|_2 \, \| \nabla_x j_{2,x+\lambda y N^{-\beta}} \|_2 \\
&\leq C N^{-\beta} e^{K|t|} \| j_1 \|_2 \| \nabla_2 j_2 \|_2 \| (\cN+1)^{1/2} \psi_1 \| \, \| (\cN+1)^{1/2} \psi_2 \|\end{split} \]
We are interested in $j_2 = \sinh_{k_t}$ and $j_2 = p_t = \cosh_{k_t} - 1$. In both cases we have $\| \nabla_2 j_2 \|_2 \leq C N^{\beta/2}$ (in the case $j_2 = p_t$, we actually have the better bound $\| \nabla_2 j_2 \|_2 \leq C$; see Appendix \ref{s:kernel}). Hence, we conclude that 
\[ \begin{split} 
\int dx dy V(y) |\phn_t (x)|^2 
\| a^\sharp (j_{1,x}) \psi_1 \| &\| a^\sharp (j_{2,x+y/N^\beta} - j_{2,x}) \psi_2 \| \\  &\leq C N^{-\beta/2} e^{K|t|}  \| (\cN+1)^{1/2} \psi_1 \| \, \| (\cN+1)^{1/2} \psi_2 \| \end{split} \]
and therefore, inserting in (\ref{eq:term3}), that 
\[\begin{split}  \Big| \int dx dy \, (N^{3\beta} &V(N^\beta (x-y)) - b_0 \delta (x-y)) \phn_t (x) \phn_t (y) \langle \psi_1, a^\sharp (j_{1,x}) a^\sharp (j_{2,y}) \psi_2 \rangle \Big| \\ \leq \; &C N^{-\beta/2} \exp (c_1 \exp (c_2 |t|)) \| (\cN+1)^{1/2} \psi_1 \| \, \| (\cN+1)^{1/2} \psi_2 \| \, . 
\end{split} \]
Together with (\ref{eq:term1}) and (\ref{eq:term4}), we obtain from (\ref{6.70}) that
\[ \begin{split} |\langle \psi_1, &(A_{1,1}^N - A_{1,1}) \psi_2 \rangle| \leq C N^{-\alpha} \exp (c_1 \exp (c_2 |t|)) \| (\cN + 1)^{1/2} \psi_1 \| \, \| (\cN + 1)^{1/2} \psi_2 \| \end{split} \]
Similarly, we can show that 
\[ \begin{split} |\langle \psi_1, &(A_{1,2}^N - A_{1,2}) \psi_2 \rangle| \leq C N^{-\alpha} 
\exp (c_1 \exp (c_2 |t|))  \| (\cN + 1)^{1/2} \psi_1 \| \, \| (\cN + 1)^{1/2} \psi_2 \| \end{split} \]
In fact, the contribution 
\[ \begin{split} 
\Big| \int dx dy &\left( N^{3\beta} V(N^\beta (x-y)) - b_0 \delta (x-y) \right) \phn_t (x) \phn_t (y) \langle \psi_1, a^\sharp (j_{1,x}) a_y \psi_2 \rangle  \Big| 
\\ \leq \; &\int dx dy V(x) \Big|  \phn_t (y) \ph_t (y+x N^{-\beta}) \langle \psi_1, a^\sharp (j_{1,y+xN^{-\beta}})  a_y \psi_2 \rangle \\ &\hspace{5cm} - \phn_t (y) \ph_t (y) \langle \psi_1, a^\sharp (j_{1,y})  a_y \psi_2 \rangle \Big| \\ 
\leq\; & \int dx dy V(x) |\phn_t (y)| \, \left| \ph_t (y+x N^{-\beta}) - \phn (y)\right| \, \| a^\sharp (j_{1,y+x N^{-\beta}}) \psi_1 \| \| a_y \psi_2 \| 
\\ &+ \int dx dy V(x) |\phn_t (y)|^2 \| a^\sharp (j_{1,y+x N^{-\beta}} - j_{1,y}) \psi_1 \| \| a_y \psi_2 \|
\end{split} \]
can be dealt with as the third term on the r.h.s. of (\ref{6.70}). 

Next, we consider the difference $A_{1,3}^N - A_{1,3}$. We write 
\begin{equation}\label{eq:A13} \begin{split}
\Big| \langle \psi_1, &(A_{1,3}^N - A_{1,3}) \psi_2 \rangle \Big| \\ \leq \; &\int dxdy N^{3\beta} V(N^\beta (x-y)) |\phn_t (x) - \ph_t (x)| \, |\phn_t (y)| \,  \| a_x \psi_1 \| \| a_y \psi_2 \| \\ &+   \int dxdy N^{3\beta} V(N^\beta (x-y)) |\ph_t (x)| |\phn_t (y) - \ph_t (y)| \,  \| a_x \psi_1 \| \| a_y \psi_2 \|
\\ &+ \int dx dy N^{3\beta} V(N^\beta (x-y)) \left[ \ph_t (x) \ph_t (y) \langle a_x \psi_1 , a_y \psi_2 \rangle - \ph_t (x)^2 \langle a_x \psi_1, a_x \psi_2 \rangle \right] \\ = \; &\text{I} + \text{II} + \text{III} 
\end{split} \end{equation}
We can control the first term on the r.h.s. by
\[ \begin{split} 
\text{I} &\leq \| \phn_t  - \ph_t \|_\infty \| \ph_t \|_\infty \| (\cN+1)^{1/2} \psi_1 \| \| (\cN+1)^{1/2} \psi_2 \| \\ &\leq C N^{-\alpha} \exp (c_1 \exp (c_2 |t|)) \| (\cN+1)^{1/2} \psi_1 \| \| (\cN+1)^{1/2} \psi_2 \|    \end{split} \]
The second term on the r.h.s. of (\ref{eq:A13}) can be controlled analogously. As for the third term, with a change of variables we can write 
\begin{equation}\label{eq:III-A13} \begin{split} 
\text{III} \leq \; &\int dx dy V(y) |\ph_t (x)| |\ph_t (x+yN^{-\beta}) - \ph_t (x)| \| a_x \psi_1 \| \| a_{x+yN^{-\beta}} \psi_2 \| \\ &+ \int dx dy V(y) |\ph_t (x)|^2 |\langle a_x \psi_1, (a_{x+yN^{-\beta}} - a_x) \psi_2 \rangle |  \\
\leq \; &C e^{K|t|} N^{-\beta} \| (\cN+1)^{1/2} \psi_1 \| \| (\cN+1)^{1/2} \psi_2 \| \\ &+ \int dx dy \int_0^1 d\lambda V(y) |y| N^{-\beta} |\ph_t (x)|^2 \| a_x \psi_1 \| \| \nabla_x a_{x+\lambda y N^{-\beta}} \psi_2 \| \\
\leq \; &C e^{K|t|} N^{-\beta} \| (\cN+1)^{1/2} \psi_1 \| \| (\cK + \cN +1)^{1/2} \psi_2 \|
\end{split} 
\end{equation}

The bounds for the terms 
\begin{equation}\label{eq:A2i} \left| \langle \psi_1, (A_{2,i}^N - A_{2,i}) \psi_2 \rangle \right| \end{equation}
for $i=1,2,3$ can be obtained similarly. In this case, since both field operators depend on the same integration variable $y$, the difference between $N^{3\beta} V(N^\beta (x-y))$ and $b_0 \delta (x-y)$ can be controlled through the difference between $\phn_t (x)$ and $\phn_t (x+y N^{-\beta})$. Smallness here follows from the regularity of $\phn_t$, there is no need to use the kinetic energy of $\psi_1$ and $\psi_2$ (in contrast with the term $A_{1,3}^N - A_{1,3}$). To illustrate this point, let us briefly consider the term
\begin{equation}\label{eq:A23} \begin{split} \Big| \langle \psi_1, (A_{2,3}^N -& A_{2,3}) \psi_2 \rangle \Big| \\ \leq \; &\left| \int dx dy (N^{3\beta} V(N^\beta (x-y)) - b_0 \delta (x-y)) \, \phn_t (x)^2 \langle a_y \psi_1 , a_y \psi_2 \rangle \right| \\
&+ b_0 \int dx \, \left| |\phn_t (x)|^2 - |\ph_t (x)|^2 \right| \| a_x \psi_1 \| \| a_x \psi_2 \| \end{split} \end{equation}
The second term can be bounded by 
\[ \begin{split} b_0 \left[ \| \phn_t \|_\infty + \| \ph_t \|_\infty \right] \, &\| \phn_t - \ph_t \|_\infty \left( \int dx \,  \| a_x \psi_1 \|^2 \right)^{1/2} \, \left( \int dx \,  \| a_x \psi_2 \|^2 \right)^{1/2} \\ &\leq C N^{-\alpha} \exp (c_1 \exp (c_2 |t|)) \| \cN^{1/2}  \psi_1 \| \|\cN^{1/2} \psi_2 \| 
\end{split} \]
As for the first term on the r.h.s. of (\ref{eq:A23}), it can be estimated by
\[ \begin{split} 
\int dx dy \, V(x) &\left| \phn_t (y+x N^{-\beta})^2 - \phn_t (y)^2 \right| \, \| a_y \psi_1 \| \| a_y \psi_2 \| \\ \leq \; & C N^{-\beta} e^{K|t|} \int dx dy \int_0^1 d\lambda \, V(x) |x|  |\nabla \phn_t (y + \lambda x N^{-\beta})| \| a_y \psi_1 \| \| a_y \psi_2 \| \\ \leq \; &C N^{-\beta} e^{K|t|} \| \cN^{1/2} \psi_1 \| \| \cN^{1/2} \psi_2 \| \end{split} \]
From the last two estimates, we conclude that
\[ |\langle \psi_1, (A_{2,3}^N - A_{2,3}) \psi_2 \rangle| \leq C N^{-\alpha}  \exp (c_1 \exp (c_2 |t|)) \| \cN^{1/2} \psi_1 \| \| \cN^{1/2} \psi_2 \|  \]
Analogously, we can also control (\ref{eq:A2i}), with $i=1,2$.
\end{proof}

\begin{lemma}\label{lm:t}
Recall the definitions (\ref{eq:TN}) and (\ref{eq:Tt}) of $T_{N,t}$ and $T_t$. 
Under the same assumptions as in Theorem \ref{thm:main2}, with $\alpha = \min (\beta/2, (1-\beta)/2)$, we have 
\[ \begin{split} &\left| \langle \psi_1, \left( (i\partial_t T_{N,t}^*) T_{N,t} - (i\partial_t T_t^*) T_t \right) \psi_2 \rangle \right| \\ &\hspace{2cm} \leq C \exp (c_1 \exp (c_2 |t|)) N^{-\alpha} \| (\cN+1)^{1/2} \psi_1 \| \| (\cN+1)^{1/2} \psi_2 \| \end{split} \]
\end{lemma}

\begin{proof}
We use the notation
\[ B_N (t) = \frac 12 \int \di x \di y \Big( k_{N,t}(x,y) a^*_x a^*_y - \lis{k}_{N,t}(x,y) a_x a_y \Big)
\]
and 
\[ B (t) = \frac 12 \int \di x \di y \Big( k_{t}(x,y) a^*_x a^*_y - \lis{k}_{t}(x,y) a_x a_y \Big)
\]
Then we have $T_{N,t} = \exp (B_N (t))$ and $T_t = \exp (B(t))$. We will make use of the representations
\[ \begin{split}  (i\partial_t T_{N,t}^*)T_{N,t} &= \sum_{n \geq 0} \frac{(-1)^{n+1}}{(n+1)!} \text{ad}_{B_N (t)}^n (\dot{B}_N (t)), \qquad
 (i\partial_t T_{t}^*)T_{t} = \sum_{n \geq 0} \frac{(-1)^{n+1}}{(n+1)!} \text{ad}_{B (t)}^n (\dot{B} (t)) \end{split} \]
with
\[ \begin{split} \text{ad}_{B_N (t)} (\dot{B}_N (t)) = \frac{1}{2} \int dx dy \left( f^N_{n,1} (x,y) a_x^* a_y^* + f^N_{n,2} (x,y) a_x a_y \right) \\
\text{ad}_{B (t)} (\dot{B} (t)) = \frac{1}{2} \int dx dy \left( f_{n,1} (x,y) a_x^* a_y^* + f_{n,2} (x,y) a_x a_y \right) 
\end{split} \]
where $f_{n,i}^N$ and $f_{n,i}$, for $i=1,2$ and $n \in \bN$ satisfy the bounds from Lemma \ref{lemma:fni} (once with $k_{N,t}$ and once with the limiting kernel $k_t$). We find
\begin{equation}\label{eq:TN-T} \begin{split}
\big| \langle \ps_1,  \big( [i &\dpr_t T^*_{N,t}] T_{N,t} - [i \dpr_t T_{t}^*] T_{t} \big) \ps_2 \rangle  \big| \\
\leq \; &\sum_{n \geq 0} \frac{1}{(n+1)!} \big| \bmedia{\ps_1, (\text{ad}_{B_N (t)}^n(\dot B_N (t)) - \text{ad}_{B (t)}^n(\dot B (t)))\ps_2} \big| \\
\leq \; &\sum_{n \geq 0} \frac{1}{(n+1)!}\big( \| f_{n,1}^N - f_{n,1}\|_2 + \| f_{n,2}^N - f_{n,2}\|_2 \big) \|(\NN+1)^{1/2} \ps_1\| \|(\NN+1)^{1/2} \ps_2\| \\
&+ \| \psi_1 \| \| \psi_2 \|  \sum_{n \geq 1} \frac{1}{(2n)!} \int \di x  \, \left| f^N_{2n-1,2}(x,x)- f_{2n-1,2}(x,x) \right| \end{split} \end{equation}

Next, we claim that 
\begin{equation}
\begin{split}
\|f_{n,i}^N - f_{n,i}\|_2  \leq \; &\big(2 \|k_t\|_2\big)^n\|\dot k_{N,t} - \dot k_t\|_2 \\
& + 2 \|k_{N,t} -  k_t\|_2 \|\dot k_{N,t}\| \sum_{j =0}^{n-1} \big(2 \|k_{N,t}\|\big)^j \big(2 \|k_t\|\big)^{n-1-j} \label{T.3} 
\end{split} \end{equation}
and that 
\begin{equation}
\begin{split}
\int \di x |f_{n,i}^N(x,x) - f_{n,i}(x,x)|  \leq \; & \big(2 \|k_t\|_2\big)^n\|\dot k_{N,t} - \dot k_t\|_2 \\ & + 2 \|k_{N,t} -  k_t\|_2 \|\dot k_{N,t}\| \sum_{j =0}^{n-1} \big(2 \|k_{N,t}\|\big)^j \big(2 \|k_t\|\big)^{n-1-j} \label{T.4}
\end{split} 
\end{equation}
for all $n \in \mathbb{N}$ (with the convention that the sum disappears if $n=0$). These bounds, together with the estimates in the appendix \ref{s:kernel}, complete the proof of the lemma. 

It remains to show (\ref{T.3}) and (\ref{T.4}). We proceed by induction over $n \in \bN$. 
For $n=0$, we have \[ f_{0,1}^N (x,y) = \dot{k}_{N,t} (x,y), \quad f_{0,2}^N (x,y) = \lis{\dot{ k}}_{N,t} (x,y) \] and analogously for $f_{0,1}$ and $f_{0,2}$; the bounds (\ref{T.3}) and (\ref{T.4}) are clearly satisfied. Let us now assume that (\ref{T.3}) and (\ref{T.4}) hold for some $n \in \bN$. We prove they also hold for $(n+1)$. We consider first the case $n$ even. By Lemma \ref{lemma:fni} we have
\begin{align*}
f_{n+1,1}(x,z) &= -\frac 12 \int \di y \Big( k_{t}(x,y) \big(f_{n,2}(z,y) + f_{n,2}(y,z) \big) \\ &\hspace{5cm} +  \lis{k_{t}(y,z)} \big(f_{n,2}(x,y) + f_{n,2}(y,x) \big) \Big) \non \\
f_{n+1,2}(x,z) &= -\frac 12 \int \di y \Big( k_{t}(y,z) \big(f_{n,1}(x,y) + f_{n,1}(y,x) \big) \\ &\hspace{5cm} +  \lis{k_{t}(x,y)} \big(f_{n,1}(z,y) + f_{n,1}(y,z) \big) \Big)\,,
\end{align*}
and similarly for $f_{n+1,1}^N(x,z) $ and $f_{n+1,2}^N(x,z)$. Therefore 
\begin{align*}
\|f_{n+1,1}^N - f_{n+1,1}\|_2 \leq \; &\frac{1}{2} \left[ \int dx dz \left| \int dy \big(k_{N,t}(x,y) f^N_{n,2}(z,y) - k_{t}(x,y) f_{n,2}(z,y) \big) \right|^2 \right]^{1/2} \\ &+\frac{1}{2} \left[ \int dx dz \left|\int \di y \big( k_{N,t}(x,y) f^N_{n,2}(y,z) - k_{t}(x,y) f_{n,2}(y,z) \big) \right|^2 \right]^{1/2} 
\\ &+ \frac{1}{2} \left[ \int dx dz  \left| \int \di y \big(\lis{k_{N,t}(y,z)} f^N_{n,2}(x,y) - \lis{k_{t}(y,z)} f_{n,2}(x,y) \big) \right|^2 \right]^{1/2} \\
&+\frac{1}{2} \left[ \int dx dz \left| \int \di y \big(\lis{k_{N,t}(y,z)} f^N_{n,2}(y,x) - \lis{k_{t}(y,z)} f_{n,2}(y,x)  \big) \right|^2 \right]^{1/2} 
\end{align*}
By the Cauchy-Schwarz inequality, we find 
\[
\|f_{n+1,1}^N - f_{n+1,1}\|_2 \leq   2\Big(\|k_{N,t}  - k_{t}\|_2\| f^N_{n,2}\|_2 + \|k_t\|_2 \| f^N_{n,2}- f_{n,2}\|_2  \Big)
\]
Using the induction assumption \eqref{T.3} and the bound $\| f^N_{n,i}\|_2 \leq (2 \|k_{N,t}\|_2)^n \|\dot k_{N,t}\|_2 $ from Lemma \ref{lemma:fni}, we obtain \eqref{T.3} with $n$ replaced by $(n+1)$. Moreover, we notice that 
\[ \begin{split}
\int dx \, &\left| f^N_{n+1,1} (x,x) - f_{n+1,1} (x,x) \right| \\ = \; &\frac{1}{2} \int dx \, \Big| \int dy (k_{N,t} (x,y) + \lis{k}_{N,t} (y,x)) (f^N_{n,2} (x,y) + f^N_{n,2} (y,x)) \\ &\hspace{4cm} - \int dy (k_{t} (x,y) + \lis{k}_{t} (y,x)) (f_{n,2} (x,y) + f_{n,2} (y,x)) \Big| \\
\leq \; & \int dx dy \, |k_{N,t} (x,y) - k_t (x,y)| \left( |f_{n,2}^N (x,y)| + |f_{n,2}^N (y,x)| \right) \\ &+ 2\int dx dy \, |k_{t} (x,y)| \, |f_{n,2}^N (x,y) - f_{n,2} (x,y)| \\
\leq \; & 2 \| k_{N,t} - k_t \|_2 \| f_{n,2}^N \|_2 +  2\| f_{n,2}^N - f_{n,2} \|_2 \| k_t \|_2 
\end{split} \]
Again, the induction assumption, combined with the bounds of Lemma \ref{lemma:fni} imply (\ref{T.4}), with $n$ replaced by $(n+1)$. The bounds for $i=2$ and for $n$ odd can be proven similarly. 
\end{proof}

Finally, we compare the last two terms on the r.h.s. of (\ref{eq:L2N0}) with the corresponding terms on the r.h.s. of (\ref{eq:L2inf0}). 
\begin{lemma}\label{lm:d}
Under the same assumptions as in Theorem \ref{thm:main2}, with $\alpha = \min (\beta/2, (1-\beta)/2)$, we let  
\[ \begin{split} D^N_{1} &= \int dx dy \, N \omega_{N,\ell} (x-y) \phi_N ((x+y)/2) a_x^* a_y^* \\ 
D^N_{2} &= N \lambda_{N,\ell} \int dx dy {\bf 1} (|x-y| \leq \ell) \phn_t ((x+y)/2)^2 a_x^* a_y^* \end{split} 
\]
and 
\[ \begin{split} 
D_1 &= \int dx dy \, \omega_{\ell}^{\text{asymp}} (x-y) \phi ((x+y)/2) a_x^* a_y^* \\ 
D_{2} &= \frac{3b_0}{8\pi \ell^3} \int dx dy {\bf 1} (|x-y| \leq \ell) \ph_t ((x+y)/2)^2 a_x^* a_y^* \end{split} 
\]
with 
\[ \phi_N (x) = \phn_t (x) \Delta \phn_t (x) + |\nabla \phn_t (x)|^2 \quad \text{ and } \quad \phi (x) = \ph_t (x) \Delta \ph_t (x) + |\nabla \ph_t (x)|^2 \]
Then we have
\[ \left| \langle \psi_1, (D_{i,N} - D_i) \psi_2 \rangle \right| \leq C N^{-\a} \exp (c_1 \exp (c_2 |t|)) \| (\cN+1)^{1/2} \psi_1 \| \| (\cN+1)^{1/2} \psi_2 \| \] 
for $i=1,2$.
\end{lemma}

\begin{proof}
We have 
\begin{equation}
\begin{split}
D^N_1 - D_1 = \; &\int \di x\,\di y\,\big(N \o_{\ell,N}(x-y) - \o^\text{asymp}_\ell(x-y) \big) \phi_N(x+y) a^*_xa^*_y \\
& + \int \di x\,\di y\,\o^\text{asymp}_\ell(x-y) \left(  \phi_N(x+y) -  \phi(x+y) \right)a^*_xa^*_y \label{6.I7}
\end{split}\end{equation}
To estimate $|\langle \psi_1, (D_1^N - D_1) \psi_2 \rangle|$, we observe that, for any Hilbert-Schmidt operator $A$ on $L^2 (\bR^3)$, with the integral kernel $A(x,y)$, we have (with the usual notation $A_x (y) = A(x,y)$)
\begin{equation}\label{eq:HS-bd} 
\begin{split} \left| \int dx dy A(x,y) \langle \psi_1 , a_x^* a_y^* \psi_2 \rangle \right| = \; &\left| \int dx \langle a_x \psi_1 , a^* (A_x) \psi_2 \rangle \right|
\\ \leq \; &\int dx \| a_x \psi_1 \| \| A_x \|_2 \| (\cN+1)^{1/2} \psi_2 \| \\ \leq \; &\| A \|_\text{HS} \| (\cN+1)^{1/2} \psi_1 \| \| (\cN+1)^{1/2} \psi_2 \| \end{split} \end{equation}
Hence, we find 
\[ \begin{split} 
\Big| \langle \psi_1, (D_1^N - &D_1) \psi_2 \rangle \Big| \\ \leq \; &\| (\cN+1)^{1/2} \psi_1 \| \| (\cN+1)^{1/2} \psi_2 \| \\ &\times \Big[ \int dx dy |N \o_{N,\ell} (x-y) - \o_\ell^\text{asymp} (x-y)|^2 |\phi_N ((x+y)/2)|^2 \\ &\hspace{2cm} + \int dx dy \o^\text{asymp}_\ell (x-y)^2 |\phi_N ((x+y)/2) - \phi ((x+y)/2)|^2 \Big]^{1/2}  \\
\leq \; &C N^{-\alpha} \exp (c_1 \exp (c_2 |t|))  \| (\cN+1)^{1/2} \psi_1 \| \| (\cN+1)^{1/2} \psi_2 \| \end{split} \]

To control the difference $D_2^N -D_2$, we write 
\begin{equation*}
\begin{split}
(D^N_2 - D_2) = \;& \big(N\l_{\ell,N} - \l_\ell \big)\,\int \di x\,\di y\,{\bf 1} (|x-y| \leq \ell)\,\phiq{x+y}^2\,a^*_xa^*_y  \\
& +  \l_\ell \int \di x\,\di y\,{\bf 1} (|x-y| \leq \ell)\,\Big( \phiq{x+y}^2- \ph_t^2((x+y)/2) \Big)\,a^*_xa^*_y 
\end{split} \end{equation*}
where we defined $\lambda_\ell = 3b_0/8\pi\ell^3$.
Hence, with (\ref{eq:HS-bd}), we find 
\[ \begin{split}  \Big| \langle \psi_1, &(D_2^N - D_2) \psi_2 \rangle \Big| \\\leq \; &C \| (\cN+1)^{1/2} \psi_1 \| \| (\cN+1)^{1/2} \psi_2 \|\\ &\hspace{1cm} \times  \left[ \left| N \lambda_{N,\ell} - \frac{3b_0}{8\pi \ell^3} \right|+ C e^{K|t|} \| \phn_t - \ph_t \|_2 \right]\\  \leq \; &C N^{-\alpha} \exp (c_1 \exp (c_2 |t|))  \,  \| (\cN+1)^{1/2} \psi_1 \| \| (\cN+1)^{1/2} \psi_2 \|
\end{split} \]
\end{proof}
We are now ready to show Theorem \ref{thm:main2}.
\begin{proof}[Proof of Theorem \ref{thm:main2}]
We write
\[ \begin{split} 
&\left\| e^{-i \cH_N t} W(\sqrt{N} \ph) T_{N,0} \xi_N - e^{-i \int_0^t \eta_N (s) ds} \, W (\sqrt{N} \phn_t) T_{N,t} \cU_{2,\infty} (t) \xi_N \right\| \\
& \hspace{1cm}\leq\|\cU_N (t;0) \xi_N - e^{-i \int_0^t \eta_N (s) ds} \cU_{2,\infty} (t;0) \xi_N\| \\
& \hspace{1cm}\leq\|\cU_N (t;0) \xi_N - e^{-i \int_0^t \eta_N (s) ds} \cU_{2,N} (t;0) \xi_N\|+\|\cU_{2,N} (t;0) \xi_N - \cU_{2,\infty} (t;0) \xi_N\| \end{split} \]
The first term can be estimated with Proposition \ref{prop:UN-UN2}. To bound the second term, we apply Lemmas \ref{lm:k}, \ref{lm:v}, \ref{lm:t}, \ref{lm:d}.  We find 
\[ \begin{split} \frac{d}{dt}&\|\cU_{2,N} (t;0) \xi_N - \cU_{2,\infty} (t;0) \xi_N\|^2 \\
&\leq C|\scal{\cU_{2,N} (t;0) \xi_N}{(\LL_{2,N}(t) - \LL_{2,\infty}(t) )\,\cU_{2,N} (t;0) \xi_N} | \\
&\leq C N^{-\alpha} \exp (c_1 \exp (c_2 |t|)) \\ &\hspace{.5cm} \times  \left[ 
\langle \cU_{2,N} (t;0) \xi_N, (\cN + \cK + 1) \cU_{2,N} (t;0) \xi_N \rangle + \langle \cU_{2,\infty} (t;0) \xi_N, (\cN  + 1) \cU_{2,\infty} (t;0) \xi_N \rangle \right] \\
&\leq C N^{-\alpha} \exp (c_1 \exp (c_2 |t|)) 
\langle \xi_N, (\cN + \cK + 1)  \xi_N \rangle 
\end{split} \]
Theorem \ref{thm:main2} follows integrating this bound over time. In the last step, we used Theorem \ref{thm:5} and the fact that, exactly as in Theorem \ref{thm:5}, we can also control the growth of the expectation of $\cN$  with respect to the limiting fluctuation dynamics $\cU_{2,\infty}$, i.e. there exist constants $C,c_1, c_2 > 0$ such that 
\[  \langle \cU_{2,\infty} (t;0) \xi_N, \cN \cU_{2,\infty} (t;0) \xi_N \rangle \leq C \exp (c_1 \exp (c_2 |t|)) \langle \xi_N, (\cN+1) \xi_N \rangle \]
\end{proof}

\appendix

\section{Properties of the Scattering Function}\label{s:scattering}

In this appendix, we give a proof of Lemma \ref{lm:propomega}, containing basic properties of the ground state $f_{N,\ell}$ of the Neumann problem 
\[ \left[ -\Delta + \frac{1}{2N} N^{3\beta} V(N^\beta .) \right] f_{N,\ell} = \l_{N,\ell} f_{N,\ell} \]
on the sphere $|x| \leq \ell$. 

\begin{proof}[Proof of Lemma \ref{lm:propomega}] 
To prove an upper bound on $\lambda_{N,\ell}$, we use a constant trial function. With the notation \[ \frak{h} = -\Delta + \frac{1}{2N} N^{3\beta} V(N^\beta x) \, , \] we find
\[ \l_{N,\ell} \leq \frac{\langle 1, \frak{h} 1 \rangle}{\langle 1, 1 \rangle} = \frac{1}{2N} \frac{\int_{|x| \leq \ell} N^{3\beta} V (N^\beta x) dx}{\int_{|x| \leq \ell} dx} = \frac{3b_0}{8\pi N \ell^3} \]
for $\ell \gg N^{-\beta}$. Before proving the lower bound in \eqref{energy}, we prove parts ii) and iii) of the Lemma. 

Part ii) can be proven as in \cite[Lemma A.1]{ESY0}; we skip the details. As for part iii), we observe that, from (\ref{rescaled_scattChi}), we can write 
\begin{equation}\label{eq:ome-exp} \begin{split} \o_{N, \ell}(x) &= C \int \di y \frac{1}{|x-y|} \big(\,N^{3\b-1} V(N^{\b}y) f_{N, \ell}(y) - \l_{N, \ell} f_{N,\ell} \big) \\ &\leq  \frac{C}{N} \int \di y \frac{1}{|x-y|}  N^{3\b} V(N^{\b}y) f_{N, \ell}(y) \end{split} \end{equation}
The Hardy-Littlewood-Sobolev inequality implies that 
\[ \begin{split} 
(|x| +N^{-\beta}) \, \omega_{N,\ell} (x) \leq \; &\frac{C}{N} \int dy \, \frac{|x|+N^{-\beta}}{|x-y|} N^{3\beta} V(N^{\beta} y) \\ \leq \; &\frac{C}{N} \| V \|_1 + \frac{C}{N^{1+\beta}} \int dy \frac{1}{|x-y|} N^{3\beta} V (N^\beta y) (N^\beta |y| + 1) \\
\leq &\frac{C}{N} \left[ \| V \|_1 + \|(|.|+1) V(.)\|_{3/2} \right] \end{split} \]
Furthermore, taking the gradient of (\ref{eq:ome-exp}), we can also estimate
\[\begin{split}
 \big|(N^{-2\b}& + |x|^2)  \nabla \o_{N, \ell}(x) \big|  \\
 &\leq C \int \di y \Big(\frac{N^{-2\b}}{|x-y|^2} + 1 + \frac{|y|^2}{|x-y|^2} \Big) \big(N^{3\b-1} V(N^{\b}y) + \l_{N, \ell} {\bf 1}(|y| \leq \ell) \big) f_{N, \ell}(y)  \non \\[3pt]
& \leq C \big(\|N^{3\b-1} V(N^{\b}\cdot)\|_{1} + N^{-2\beta} \|N^{3\b-1} V(N^{\b} . ) (N^{2\b} |.|^2 + 1) \|_{3} \big)  \non \\
& \qquad + C \l_{N, \ell} \big( N^{-2\b}\|{\bf 1}(|y| \leq \ell)\|_{3} + \|{\bf 1}(|y| \leq \ell)\|_{1} + \||y|^2 {\bf 1}(|y| \leq \ell)\|_{3} \big)  \non \\[3pt]
&   \leq \frac{C}{N} \big(\|V\|_{3} + \|V\|_{1} + \||.|^2 V(.)\|_{3} +1 \big)
 \leq \frac {C'} N\end{split}
\]
for any $\ell \gg N^{-\b}$.

Finally, we have to prove a lower bound for the eigenvalue $\lambda_{N,\ell}$. To this end, we use the estimates $0 \leq \omega_{N,\ell} (x) \leq C/(N|x|)$, established in part ii) and iii). We have
\[ \begin{split}  
\lambda_{N,\ell} = \;&\frac{\langle f_{N,\ell}, \frak{h} f_{N,\ell} \rangle}{\langle f_{N,\ell}, f_{N,\ell} \rangle} \\ \geq \; & \frac{1}{\int_{|x| \leq \ell} \left( 1 + \omega_{N,\ell}^2 (x) \right) dx} \frac{1}{2N} \int_{|x| \leq \ell} N^{3\beta} V (N^\beta x)  \left( 1- 2\omega_{N,\ell} (x) \right)  \\
\geq \; & \frac{1}{\frac{4}{3} \pi \ell^3 + \frac{C\ell}{N^2}} \frac{1}{2N} \left[ b_0 - \frac{C}{N^{1-\beta}} \| V (.) / |.| \|_1 \right] \\
\geq \; & \frac{3b_0}{8\pi N \ell^3} \left(1 - \frac{C}{N^{1-\beta}} \right) \end{split} \]
for all $\ell \gg N^{-\beta}$ and $0 < \beta < 1$. 
\end{proof}

\section{Properties of Nonlinear Schr\"odinger Equations} \label{phi}

For a given initial data $\varphi\in H^1(\mathbb{R}^3)$, we define $\phn_t$ as the solution of the modified time dependent nonlinear Schr\"odinger equation
\be
 i\partial_t \phn_t(x)=-\Delta \phn_t(x)+(N^{3\beta}V(N^\beta\cdot)f_{N,\ell}(\cdot)*|\phn_t|^2)(x)\phn_t(x)\label{MGP}
\ee
where $f_{N,\ell}$ is the solution of the scattering equation \eqref{rescaled_scattChi}. Moreover, we denote by $\ph_t$ the solution of the limiting equation 
\be
 i\partial_t \ph_t=-\Delta \ph_t+b_0|\ph_t|^2\ph_t\label{GP}
\ee
where $b_0=\int dx \,V(x)$. 

We need some standard facts concerning the well-posedness and the propagation of higher Sobolev regularity for these equations. Moreover, we need to estimate their difference in the limit of large $N$. This is the content of the next proposition. 
\begin{prop}\label{prop:ph}
Let $V \in L^1 \cap L^3 (\bR^3, (1+|x|^6) dx)$ be non-negative and spherically symetric. Let $\ph \in H^1 (\bR^3)$ with $\| \ph \|_2 = 1$.
\begin{itemize}
\item[i)] There exist unique global solutions $\phn_.$ and $\ph_.$ in $C(\bR; H^1 (\bR^3))$ of (\ref{MGP}) and, respectively, of (\ref{GP}) with initial data $\ph$. The solutions are such that $\| \ph_t \|_2 = \| \phn_t \|_2 = \| \ph \|_2 =1$ and 
\[ \| \ph_t \|_{H^1} , \| \phn_t \|_{H^1} \leq C \]
for a constant $C > 0$ and all $t \in \bR$. 
\item[ii)] Under the additional assumption that $\ph \in H^n (\bR^3)$, for an integer $n \in \bN$, then $\ph_t, \phn_t \in H^n (\bR^3)$ for all $t \in \bR$ and there exist constants $C > 0$ (depending on $\| \ph \|_{H^n}$ and on $n$) and $K > 0$ (depending only on $\| \ph \|_{H^1}$ and on $n$) such that 
\[ \| \ph_t \|_{H^n} , \| \phn_t \|_{H^n} \leq C e^{K|t|} \] for all $t \in \bR$.
\item[iii)] Let $\ph \in H^4 (\bR^3)$. Then there exists $C > 0$ (depending on $\| \ph \|_{H^4}$) and $K > 0$ (depending only on $\| \ph \|_{H^1}$) such that 
\[ \| \dot{\ph}_t \|_{H^2}, \| \ddot{\ph}_t \|_2 \leq C e^{K|t|} \]
\item[iv)] Let $\ph \in H^2 (\bR^3)$. Then there exist constants $C,c_1 > 0$ (depending on $\| \ph \|_{H^2}$) and $c_2 > 0$ (depending only on $\| \ph \|_{H^1}$) such that 
\[ \| \ph_t - \phn_t \|_2 \leq C N^{-\gamma} \exp (c_1 \exp (c_2 |t|)) \qquad \text{with } \gamma = \min (\beta , 1-\beta). \]
\item[v)] Let $\ph \in H^4 (\bR^3)$. Then there exist constants $C,c_1 > 0$ (depending on $\| \ph \|_{H^4}$)
and $c_2 > 0$ (depending only on $\| \ph \|_{H^1}$) such that 
\[ \| \ph_t - \phn_t \|_{H^2}, \| \dot\ph_t - \dot\phn_t \|_2 \leq C N^{-\gamma} \exp (c_1 \exp (c_2 |t|)) \qquad \text{with } \gamma = \min (\beta, 1-\beta)\]
\end{itemize}
\end{prop}

\begin{proof}
The proof of i)-iii) is quite standard and can be found, for example, in \cite[Prop. 3.1]{BdOS12}. Also the proof of the bound 
\[ \| \ph_t - \phn_t \|_2 \leq C N^{-\gamma} \exp (c_1 \exp (c_2 |t|)) , \qquad \text{with } \gamma = \min (\beta, 1-\beta)  \]
can be found in \cite[Prop. 3.1]{BdOS12}, up to the observation that 
\[ \left| \int dx N^{3\beta} V(N^\beta x) f_{N,\ell} (x) - b_0 \right| \leq C N^{-1+\beta} \]
as follows from Lemma \ref{lm:propomega}, writing $f_{N,\ell} = 1- \omega_{N,\ell}$. To prove v), we observe that 
\[ \ph_t - \phn_t = i \int_0^t ds \, e^{-i\Delta (t-s)} \left[ (U*|\phn_s|^2) \phn_s - b_0 |\ph_s|^2 \ph_s \right] \]
with $U_N (x) = N^{3\beta} V(N^\beta x) f_{N,\ell} (x)$. We estimate
\[ \begin{split} 
\| \ph_t &- \phn_t \|_{H^2} \\ \leq \; &\sum_{|\alpha_1| + |\alpha_2| + |\alpha_3| \leq 2} \int_0^t ds \, \left\| (U_N * \nabla^{\alpha_1} \phn_s \, \nabla^{\alpha_2} \overline{\ph}_s^{N}) \nabla^{\alpha_3} \ph_s^{N} - b_0 \nabla^{\alpha_1} \ph_s \nabla^{\alpha_2} \overline{\ph}_s \nabla^{\alpha_3} \ph_s \right\| \\
\leq \; &\sum_{|\alpha_1| + |\alpha_2| + |\alpha_3| \leq 2} \int_0^t ds \, \left\| (U_N * (\nabla^{\alpha_1} \phn_s - \nabla^{\alpha_1} \ph_s) \, \nabla^{\alpha_2} \overline{\ph}_s^{N}) \nabla^{\alpha_3} \phn_s \right\| \\&+ 
\sum_{|\alpha_1| + |\alpha_2| + |\alpha_3| \leq 2} \int_0^t ds \, \left\| (U_N * \nabla^{\alpha_1} \ph_s  \, (\nabla^{\alpha_2} \overline{\ph}_s^{N} - \nabla^{\alpha_2} \overline{\ph}_t)) \nabla^{\alpha_3} \ph_s^{N} \right\| \\
&+ \sum_{|\alpha_1| + |\alpha_2| + |\alpha_3| \leq 2} \int_0^t ds \, \left\| (U_N * \nabla^{\alpha_1} \ph_s  \, \nabla^{\alpha_2} \overline{\ph}_s) (\nabla^{\alpha_3} \ph_s^{N} - \nabla^{\alpha_3} \ph_s) \right\| \\
&+  \sum_{|\alpha_1| + |\alpha_2| + |\alpha_3| \leq 2} \int_0^t ds \, \left\| (U_N * \nabla^{\alpha_1} \ph_s  \, \nabla^{\alpha_2} \overline{\ph}_s) \nabla^{\alpha_3} \ph_s - b_0 \nabla^{\alpha_1} \ph_s  \, \nabla^{\alpha_2} \overline{\ph}_s \nabla^{\alpha_3} \ph_s \right\| 
\end{split}
\]
Hence, we obtain 
\[ \begin{split} \| \ph_t &- \phn_t \|_{H^2} \\ \leq \;& \sum_{|\alpha_1| + |\alpha_2| + |\alpha_3| \leq 2} \| U_N \|_1  \int_0^t ds \, \Big[ \| \nabla^{\alpha_3} \phn_s \|_\infty \| \nabla^{\alpha_2} \phn_s \|_\infty  \| \nabla^{\alpha_1} \phn_s - \nabla^{\alpha_1} \ph_s \|_2 \\ &\hspace{3cm} + \| \nabla^{\alpha_1} \ph_s \|_\infty \| \nabla^{\alpha_3} \phn_s \|_\infty  \| \nabla^{\alpha_2} \phn_s - \nabla^{\alpha_2} \ph_s \|_2 \\ &\hspace{3cm} + 
\| \nabla^{\alpha_1} \ph_s \|_\infty \| \nabla^{\alpha_2} \ph_s \|_\infty  \| \nabla^{\alpha_3} \phn_s - \nabla^{\alpha_3} \ph_s \|_2 \Big] \\ &+
\sum_{|\alpha_1| + |\alpha_2| + |\alpha_3| \leq 2} \int_0^t ds \, \| (U* \nabla^{\alpha_1} \ph_s \nabla^{\alpha_2} \overline{\ph}_s) \nabla^{\alpha_3} \ph_s - b_0 \nabla^{\alpha_1} \ph_s \nabla^{\alpha_2} \overline{\ph}_s \nabla^{\alpha_3} \ph_s \| \end{split} \]
and using the propagation of regularity from part ii) of the proposition, we conclude that 
\[ \| \ph_t - \phn_t \|_{H^2}  \leq C \int_0^t e^{K|s|} \| \ph_s - \phn_s \|_{H^2} + C N^{-\gamma} e^{K|t|}  \]
By Gronwall, we find 
\[ \|  \ph_t - \phn_t \|_{H^2} \leq C N^{-\gamma} \exp (c_1 \exp (c_2 |t|)) \]
The estimate for $\| \dot{\ph}_t - \dot{\ph}_t^{N} \|_2$ follows with the help of the equations for $\dot\ph_t$ and $\dot{\ph}^{N}_t$.
\end{proof}

\section{Properties of the operator $k_{N,t}$}\label{s:kernel}

We define the integral kernel 
\begin{equation}\label{eq:kNtA} k_{N,t} (x,y) = -N \omega_{N,\ell} (x-y) (\phn_t ((x+y)/2))^2 
\end{equation}
where $f_{N,\ell} = 1- \omega_{N,\ell}$ is the solution of the scattering equation (\ref{rescaled_scattChi}). 
Note that $k_{N,t}$ is the integral kernel of a Hilbert-Schmidt operator on $L^2 (\bR^3)$. In fact, using part iii) of Lemma \ref{lm:propomega}, we find 
\begin{equation}\label{eq:kN-HS} \begin{split}  \| k_{N,t} \|_\text{HS}^2 &= \int dx dy \, |k_{N,t} (x,y)|^2 = \int dxdy N^2 \omega_{N,\ell}^2 (x-y) |\phn_t ((x+y)/2)|^4 \\ &\leq C \int dx dy \, \frac{{\bf 1} (|x-y|\leq \ell)}{|x-y|^2} |\phn_t ((x+y)/2)|^4 \leq C \ell \end{split} \end{equation}

In the next five lemmas, we collect some properties of the kernel $k_{N,t}$ and of the operators $s_{N,t} = \sinh_{k_{N,t}}$, $p_{N,t} = \cosh_{k_{N,t}}-1$, $r_{N,t} = s_{N,t} - k_{N,t}$. 

\begin{lemma} \label{l:BoundsK}
Let $V$ be as in Lemma \ref{lm:propomega} and $k_{N,t}$ be defined as in (\ref{eq:kNtA}), with $\ph = \ph_{t=0}^N \in H^2 (\bR^3)$ (so that $\| \ph^N_t \|_{H^2} \leq C e^{K|t|}$ for all $t \in \bR$, by Proposition \ref{prop:ph}). Then 
\[ \begin{split}
(i) \qquad &\|k_{N,t}\|_2\,, \|p_{N,t}\|_2\,,\|s_{N,t}\|_2 \,,\|r_{N,t}\|_2   \leq C \non \\[3pt]
(ii) \qquad  &\sup_{x \in \bR^3} \|k_{N,t}(x, \cdot)\|_2\,,\sup_{x \in \bR^3} \|p_{N,t}(x, \cdot)\|_2 \leq C e^{K|t|} \,, \\ &\sup_{x \in \bR^3}\|s_{N,t}(x, \cdot)\|_2,\,\sup_{x \in \bR^3} \|r_{N,t}(x, \cdot)\|_2   \leq C e^{K|t|}  \non  \\[3pt]
(iii) \qquad &\sup_{x,y} |r_{N,t} (x;y)|, \, \sup_{x,y} |p_{N,t} (x,y)| \leq C e^{K|t|} 
 \non\end{split}
\]
\end{lemma}

\vskip 0.5cm

\begin{proof} 
We start with the property (i). The bound for $k_{N,t}$ is proven in (\ref{eq:kN-HS}) above. The bounds for $\|p_{N,t}\|_2$, $\|r_{N,t}\|_2$ and $\|s_{N,t}\|_2$ follow from 
\[ 
\|K_1 K_2\|_2 \leq \|K_1 \|_2 \,\|K_2\|_2 
\]
for any two kernels $K_1,\,K_2 \in L^2(\RRR^3 \times \RRR^3)$ 
(equivalently, $\| K_1 K_2 \|_\text{HS} \leq \| K_1 \|_\text{HS} \| K_2 \|_\text{HS}$ for any Hilbert-Schmidt operators $K_1, K_2$). 

Next, we show (ii). We observe that
\[ \begin{split} \| k_{N,t} (x,.) \|_2^2 &= \int dy \, |k_{N,t} (x,y)|^2 \\ &\leq C \int dy \frac{1}{|x-y|^2} |\phn_t ((x+y)/2)|^4 \\ &\leq \int dy \, \left| \nabla |\phn_t (y)|^2 \right|^2 \leq C \| \phn_t \|_{H^1}^2 \| \phn_t \|_\infty^2 \leq C e^{K|t|} \end{split} \]
using Hardy's inequality. To prove the bound for $\sup_x \| p_{N,t} (x,.) \|_2$, we use 
\[ \begin{split} 
\| (k_{N,t} \overline{k}_{N,t})^n (x,.) \|_2^2  =\; &\int dy \left| \int dz \, k_{N,t} (x,z) \left[ (\overline{k}_{N,t} k_{N,t})^{n-1} \overline{k}_{N,t} \right] (z,y) \right|^2  \\ \leq \; &\int dy dz dw \, 
|k_{N,t} (x,z)|^2 \, \left| \left[ (\overline{k}_{N,t} k_{N,t})^{n-1} \overline{k}_{N,t} \right] (w,y) \right|^2 \\ \leq \; & \| k_{N,t} (x,.) \|_2^2 \, \| (\overline{k}_{N,t} k_{N,t})^{n-1} \overline{k}_{N,t}  \|_2^2 
\leq \; \| k_{N,t} (x,.) \|_2^2 \, \| k_{N,t} \|_2^{4n-2}
\end{split} \]
This gives
\[ \sup_x \| p_{N,t} (x,) \|_2 \leq C e^{\| k_{N,t} \|} \sup_x \| k_{N,t} (x,.) \|_2  \leq C  e^{K|t|} \]
as claimed. The bounds for $\sup_x \| r_{N,t} (x,.) \|_2$ and $\sup_x \| s_{N,t} (x,.) \|_2$ can be proven similarly. 

Finally, we show iii). From the definition of $r_{N,t}$, we find
\[ \begin{split}
r_{N,t} (x,y) = \; & \sum_{n=1}^\infty \frac{1}{(2n+1)!} (k_{N,t} \overline{k}_{N,t})^n k_{N,t} (x,y) \\ = \; &\sum_{n=1}^\infty \frac{1}{(2n+1)!} \int dz dw \, k_{N,t} (x;z) \left[ (\overline{k}_{N,t} k_{N,t})^{n-1} \overline{k}_{N,t}\right](z,w)\, k_{N,t} (w,y)
\end{split} \]
and therefore
\[\begin{split} 
|r_{N,t} (x,y)| \leq \; &\sum_{n=1}^\infty \frac{1}{(2n+1)!} \left[ \int dw dz |k_{N,t} (x,z)|^2 |k_{N,t} (w,y)|^2 \right]^{1/2} \\ &\hspace{4cm} \times \left[ \int dw dz \, |(\overline{k}_{N,t} k_{N,t})^{n-1} \overline{k}_{N,t} (z,w)|^2 \right]^{1/2} \\ \leq \; &\sum_{n=1}^\infty \frac{1}{(2n+1)!} \| k_{N,t} \|_2^{2n-1} \, \| k_{N,t} (x,.) \|_2 \| k_{N,t} (.,y) \|_2 \leq C e^{K|t|} \end{split} \]
for every $x,y \in \bR^3$. The bound for $p_{N,t}$ can be proven analogously. 
\end{proof}

\begin{lemma}\label{lm:deri}
Let $V$ be as in Lemma \ref{lm:propomega} and $k_{N,t}$ be defined as in (\ref{eq:kNtA}), with $\ph = \ph_{t=0}^N \in H^2 (\bR^3)$ (so that $\| \ph^N_t \|_{H^2} \leq C e^{K|t|}$ for all $t \in \bR$, by Proposition \ref{prop:ph}). Then 
\[ \begin{split} 
(i) \qquad & \| \nabla_1 k_{N,t} \|_2, \| \nabla_2 k_{N,t} \|_2 \leq C N^{\beta/2} + C e^{K|t|} \\ 
(ii) \qquad & \| \nabla_1 p_{N,t} \|_2, \| \nabla_2 p_{N,t} \|_2, \| \nabla_1 r_{N,t} \|_2, \| \nabla_2 r_{N,t} \|_2 \leq C \\
(iii) \qquad & \| \Delta_1 r_{N,t} \|_2, \| \Delta_2 r_{N,t} \|_2, \| \Delta_1 p_{N,t} \|_2, \| \Delta_2 p_{N,t} \|_2 \leq C e^{K|t|} \end{split} \]
\end{lemma}

\begin{proof}
From the definition of $k_{N,t}$, we find
\[ \begin{split}\nabla_1 k_{N,t} (x,y) = \; &-N \nabla \omega_{N,\ell} (x-y) (\phn_t ((x+y)/2))^2 \\ &- N \omega_{N,\ell} (x-y) \phn_t ((x+y)/2) \nabla \phn_t ((x+y)/2) \end{split} \]
Hence
\[ \begin{split} 
\int dx dy |\nabla_1 &k_{N,t} (x,y)|^2  \\ \leq \; &C\int dx dy \frac{1}{(|x-y|^2 + N^{-2\beta})^2} |\phn_t ((x+y)/2)|^4 \\ &+ C\int dx dy \frac{1}{(|x-y|+N^{-\beta})^2} |\nabla \phn_t ((x+y)/2)|^2 |\phn_t ((x+y_/2)|^2 \\ \leq \; &C N^\beta + C e^{K|t|}  \end{split} \]
where we used the bounds from Lemma \ref{l:BoundsK}. 

To bound the $\| \nabla_1 p_{N,t} \|_2, \|\nabla_1 r_{N,t}\|_2$, we need an estimate for $\| \nabla_1 (k_{N,t} \overline{k}_{N,t}) \|_2$. To this end, we notice that 
\[ \begin{split} 
\sup_x \int dy \, &|\nabla_x k_{N,t} (x,y)|\\ &\leq C \sup_x \int dy \left[ \frac{|\phn_t ((x+y)/2)|^2}{|x-y|^2} + \frac{|\phn_t ((x+y)/2)| |\nabla \phn_t ((x+y)/2)|}{|x-y|} \right] \\
&\leq C \sup_x \int dy \, \left[ \frac{|\phn_t ((x+y)/2)|^2}{|x-y|^2} + |\nabla \phn_t ((x+y)/2)|^2 \right] \leq C \end{split} 
\]
by Hardy's inequality. Analogously, we find 
\[ \sup_y \int dx |\nabla_x k_{N,t} (x,y)| \leq C \]
Hence, we obtain 
\begin{equation}\label{eq:nabkk} 
\begin{split} 
\int dx dy &\big|\nabla_1 k_{N,t} \overline{k}_{N,t} (x,y) \big|^2 \\  \leq \; & \int dx dy dz_1 dz_2 |\nabla_x k_{N,t} (x,z_1)| |k_{N,t} (z_1 , y)| |\nabla_x k_{N,t} (x,z_2)| 
|k_{N,t} (z_2, y)| \\
\leq \; & \int dx dy dz_1 d z_2 |\nabla_x k_{N,t} (x,z_1)| |\nabla_x k_{N,t} (x,z_2)| |k_{N,t} (z_2, y)|^2 \\
\leq \;& \| k_{N,t} \|_2^2 \left[ \sup_x \int dy |\nabla_x k_{N,t} (x,y)| \right] \left[ \sup_{y} \int dx |\nabla_x k_{N,t} (x,y)| \right] \leq C 
\end{split} 
\end{equation}
Now, we are ready to bound $\| \nabla_1 p_{N,t} \|_2$. {F}rom  the definition, we have 
\[ \nabla_1 p_{N,t} = \nabla_1 (k_{N,t} \bar k_{N,t}) \Big[ \sum_{n=1}^\io \frac{1}{(2n)!} (k_{N,t} \bar k_{N,t})^{n-1}\Big]  \]
and therefore
\[ 
\|\nabla_1 p_{N,t}\|_2 \leq e^{\|k_{N,t}\|_2}\,\|\nabla (k_{N,t} \bar k_{N,t})\|_2 \leq C  
\]
Similarly, one can show the bounds for $\nabla_2 p_{N,t}$, $\nabla_1 r_{N,t}$ and $\nabla_2 r_{N,t}$. 

As for the estimates involving second derivatives, we claim that 
\begin{equation}\label{eq:D1kk-1} \|\D_1 (k_{N,t} \bar k_{N,t})\|_{2} \leq C e^{K|t|}  \end{equation}
In fact,
\begin{equation} \begin{split} \Delta_x k_{N,t} \overline{k}_{N,t} (x,y) = \; & \int dz \Delta_x k_{N,t} (x,z) k_{N,t} (z,y) \\ = \; &- N 
\int dz (\Delta \omega_{N,\ell}) (x-z) (\phn_t ((x+z)/2))^2 k_{N,t} (z,y) \\ &-2N \int dz (\nabla \omega_{N,\ell}) (x-z) \left[ \nabla_x (\phn_t ((x+z)/2))^2 \right] k_{N,t} (z,y) \\ &-N \int dz \omega_{N,\ell} (x-z) \left[ \Delta_x (\phn_t ((x+z)/2))^2 \right] k_{N,t} (z,y) \end{split} \end{equation}
and therefore
\begin{equation}\label{eq:D1kk-2} \begin{split}
\Delta_x k_{N,t}\overline{k}_{N,t} (x,y) = \; &- N 
\int dz (\nabla \omega_{N,\ell}) (x-z) (\phn_t ((x+z)/2))^2 \nabla_z k_{N,t} (z,y) \\ &-3N \int dz (\nabla \omega_{N,\ell}) (x-z) \left[ \nabla_x (\phn_t ((x+z)/2))^2 \right] k_{N,t} (z,y) \\ &-N 
\int dz \omega_{N,\ell} (x-z) \left[ \Delta_x (\phn_t ((x+z)/2))^2 \right] k_{N,t} (z,y) \\
=\; & \text{I} + \text{II} + \text{III} 
\end{split} \end{equation}
where we used integration by parts. The $L^2$-norm of first term  can be bounded by 
\begin{equation}\label{eq:termI} \begin{split} 
\| \text{I} \|^2_2 \leq \; & \int dx dz_1 dz_2 \, \frac{{\bf 1}(|x-z_1| \leq \ell) {\bf 1} (|x-z_2| \leq \ell)}{|x-z_1|^2 |x-z_2|^2} \, \\ &\hspace{.3cm}  \times  |\phn_t ((x+z_1)/2)|^2 |\phn_t ((x+z_2)/2)|^2 \int dy \, |\nabla_{z_1} k_{N,t} (z_1, y)| |\nabla_{z_2} k_{N,t} (z_2,y)| \end{split} \end{equation}
Notice that
\[ \begin{split}  \int dy \, &|\nabla_{z_1} k_{N,t} (z_1, y)| |\nabla_{z_2} k_{N,t} (z_2,y)| \\ \leq \; &\int dy \left[ \frac{|\phn_t ((z_1+y)/2)|^2}{|z_1 -y|^2} + |\nabla \phn_t ((z_1 + y)/2)|^2 \right]  \\ & \hspace{3cm} \times \left[ \frac{|\phn_t ((z_2+y)/2)|^2}{|z_2 -y|^2} + |\nabla \phn_t ((z_2 + y)/2)|^2 \right]  \\ \leq \; &C e^{K|t|}  \left[ \frac{1}{|z_1 - z_2|} + 1 \right] \end{split} \]
Inserting this bound on the r.h.s. of (\ref{eq:termI}) and shifting the integration variables, we conclude that 
\[ \begin{split} 
\| \text{I} \|^2_2 \leq \; &C e^{K|t|}  \int dz_1 dz_2 \, \frac{{\bf 1}(|z_1| \leq \ell) {\bf 1}(|z_2| \leq \ell)}{|z_1|^2 |z_2|^2} \left[\frac{1}{|z_1 - z_2|} + 1 \right] \\ &\hspace{4cm} \times  \sup_{z_1, z_2} \int dx |\phn_t (x+z_1/2)|^2 |\phn_t (x+z_2/2)|^2 \\
\leq \; &C e^{K|t|}  
 \end{split} \]
with the Hardy-Littelwood-Sobolev inequality. The second term on the r.h.s. of (\ref{eq:D1kk-2}) can be controlled by 
\[ \begin{split} \| \text{II} \|_2^2 \leq \; & C \| k_{N,t} \|_2^2 \, \sup_{z_1} \int dx \, \frac{1}{|x-z_1|^2} \left[ |\nabla_x \phn_t ((x+z_1)/2)|^2 + |\phn_t ((x+z_1)/2)|^2 \right] \\ &\hspace{.5cm} \times \sup_{x} \int dz_2 \frac{1}{|x-z_2|^2} \left[ |\nabla_x \phn_t ((x+z_2)/2)|^2 + |\phn_t ((x+z_2)/2)|^2 \right] \\ \leq \;&C e^{K|t|}  \end{split} \]
Finally, the third term on the r.h.s. of (\ref{eq:D1kk-2}) is bounded by 
\[ \begin{split} 
\| \text{III} \|_2^2 \leq \; &C \int dx dy dz_1 dz_2 \frac{1}{|x-z_1| |x-z_2|} |k_{N,t} (z_1 , y)| |k_{N,t} (z_2 ,y)| \\ &\hspace{.5cm} \times \left[ |\Delta \phn_t ((x+z_1/2))| |\phn_t ((x+z_1)/2)| + |\nabla \phn_t ((x+z_1)/2)|^2 \right]\\ &\hspace{.5cm} \times \left[ |\Delta \phn_t ((x+z_2/2))| |\phn_t ((x+z_2)/2)| + |\nabla \phn_t ((x+z_2)/2)|^2 \right]  \\
\leq \; &C \| k_{N,t} \|_2^2 \, \sup_{z_1} \int dx \Big[ |\Delta \phn_t ((x+z_1/2))|^2 +\frac{|\phn_t ((x+z_1)/2)|^2}{|x-z_1|^2} \\ &\hspace{8cm} + \frac{|\nabla \phn_t ((x+z_1)/2)|^2}{|x-z_1|} \Big] 
\\ & \times 
\sup_x \int dz_2 \Big[ |\Delta \phn_t ((x+z_2/2))|^2 + \frac{|\phn_t ((x+z_2)/2)|^2}{|x-z_2|^2} \\ &\hspace{8cm} + \frac{|\nabla \phn_t ((x+z_2)/2)|^2}{|x-z_2|} \Big] 
\\ \leq \; &C e^{K|t|}  \end{split} \]
This concludes the proof of (\ref{eq:D1kk-1}). {F}rom (\ref{eq:D1kk-1}), we immediately obtain
\[  \| \Delta_1 p_{N,t} \|_2 \leq C e^{\| k_{N,t} \|_2} \| \Delta (k_{N,t} \overline{k}_{N,t}) \|_2 \leq C e^{K|t|}  \] 
and 
\[  \| \Delta_1 r_{N,t} \|_2 \leq C \| k_{N,t} \|_2 \, e^{\| k_{N,t} \|_2} \| \Delta (k_{N,t} \overline{k}_{N,t}) \|_2 \leq C e^{K|t|}  \] 
as claimed. 
\end{proof}

\begin{lemma} \label{BoundsKdot}
Let $V$ be as in Lemma \ref{lm:propomega} and $k_{N,t}$ be defined as in (\ref{eq:kNtA}), with $\ph = \ph_{t=0}^N \in H^4 (\bR^3)$ (so that $\| \ph^N_t \|_{H^4} \leq C e^{K|t|}$ for all $t \in \bR$, by Proposition \ref{prop:ph}). Then 
\[\begin{split}
(i) \qquad &\|\dot k_{N,t}\|_2 \leq C e^{K|t|}
\\
(ii) \qquad & \|\ddot k_{N,t}\|_2 \leq C e^{K|t|} \\
(iii) \qquad & \|\dot p_{N,t} \|_2,\,\|\dot s_{N,t} \|_2 \leq C e^{K|t|}  \end{split}
\]
\end{lemma}
\begin{proof}
The bounds (i) and (ii) follow from 
\[\begin{split}
\|\dot k_{N,t}\|_2^2 & = 4 \int \di x \di y |N \omega_{\ell,N}(x-y)|^2 |\phis{x+y}|^2 | \dotphis{x+y}|^2  \\
& \leq C \| \phn_t \|^2_\infty  \| \dot{\ph}^N_t \|_2^2 \leq C e^{K|t|} \end{split} \]
and 
\[ \begin{split} 
\|\ddot k_{N,t}\|_2^2 & \leq C \int \di x \di y |N \omega_{\ell,N}(x-y)|^2 \\ &\hspace{1cm} \times \big[|\phis{x+y}|^2 | \ddotphis{x+y}|^2 + | \dotphis{x+y}|^4 \Big]  \\
& \leq C \left[ \| \ph_t \|_{H^2}^2 \| \ddot \ph^N_t \|^2_2 + \| \dot\ph^N_t \|^4_{H^1} \right] \leq C e^{K|t|} \end{split} \]

The bounds (iii) follow from (see \cite{BdOS12}) 
 \[
  \| \dot p_{N,t} \|_2 \leq \| \dot k_{N,t} \|_2 \, \cosh (\| k_{N,t} \|_2), \qquad \text{and } \qquad  \| \dot s_{N,t} \|_2 \leq \| \dot k_{N,t} \|_2 \, \sinh (\| k_{N,t} \|_2) \]
\end{proof}

\begin{lemma} \label{BoundsKxdot}
Let $V$ be as in Lemma \ref{lm:propomega} and $k_{N,t}$ be defined as in (\ref{eq:kNtA}), with $\ph = \ph_{t=0}^N \in H^4 (\bR^3)$ (so that $\| \ph^N_t \|_{H^4} \leq C e^{K|t|}$ for all $t \in \bR$, by Proposition \ref{prop:ph}). Then 
\[\begin{split}
(i) \quad & \sup_{x \in \bR^3} \|\dot k_{N,t} \, (\cdot, x)\|_2 \leq C e^{K|t|} \\
(ii) \quad & \sup_{x \in \bR^3} \|\dot p_{N,t} (\cdot,x) \|_2,\, \sup_{x \in \bR^3} \|\dot r_{N,t} (\cdot, x) \|_2 , \, \sup_{x \in \bR^3} \| \dot s_{N,t} (\cdot, x) \|_2 \leq C e^{K|t|} \\
(iii) \quad & \sup_{x,y \in \bR^3} |\dot r_{N,t} (x,y)| \leq C e^{K|t|} 
 \label{BoundsKdot2}
\end{split}
\]
\end{lemma}

\begin{proof}
The bounds can be proven as in Lemma \ref{l:BoundsK}, taking into account the fact that the kernel of $\dot{k}_{N,t}$ is similar to the one of $k_{N,t}$, just with $\ph^N_t ((x+y)/2)^2$ replaced by $2 \phn_t ((x+y)/2) \dot{\ph}^N_t ((x+y)/2)$. 
\end{proof}

\begin{lemma}   
Let $V$ be as in Lemma \ref{lm:propomega} and $k_{N,t}$ be defined as in (\ref{eq:kNtA}), with $\ph = \ph_{t=0}^N \in H^4 (\bR^3)$ (so that $\| \ph^N_t \|_{H^4} \leq C e^{K|t|}$ for all $t \in \bR$, by Proposition \ref{prop:ph}). Then 
\[\begin{split}
(i) \qquad & \|\nabla \dot p_{N,t}\|_2, \| \nabla_1 \dot r_{N,t} \|_2 \leq  C e^{K|t|} ,   \\
(ii) \qquad &\|\D_1 \dot r_{N,t} \|_{2}\,, \|\D_1 \dot p_{N,t} \|_{2} \leq C e^{K|t|} \end{split}
\]
\end{lemma}

\begin{proof}
Also in this case, the statements can be proven similarly as in Lemma \ref{lm:deri}, keeping in mind that the kernels of  $\dot p_{N,t}$ and $\dot r_{N,t}$ look exactly like the kernels of $p_{N,t}$ and $r_{N,t}$ with the only difference that in the series expansion defining we have to replace one factor of $k_{N,t}$ or $\overline{k}_{N,t}$ by its time-derivative (which simply means that in this one factor, we have to replace $\phn_t ((x+y)/2)^2$ by $2 \phn_t ((x+y)/2) \dot{\ph}^N_t ((x+y)/2)$).
\end{proof} 

Next, we establish the convergence of the kernel $k_{N,t}$ towards its limit, as $N \to \infty$.

We denote
\[\begin{split}
k_{t}(x,y) & = -\o_{\ell}^{\text{asymp}}(x-y)\,\ph_t^2({(x+y)}/2)\end{split}
\]
with $\o_{\ell}^{\text{asymp}}$ given by
\[ \o_{\ell}^{\text{asymp}}(x)= \frac{b_0}{8\pi} \left[  \frac{1}{|x|} - \frac{3}{2\ell} + \frac{x^2}{2\ell^3}\right] \]
for all $x \in \bR^3$ with $|x| \leq \ell$ and by $\o_\ell^\text{asymp} (x) = 0$ for $|x| > \ell$. Moreover, we define $s_t = \sinh (k_t)$, $p_t = \cosh (k_t) - 1$, $r_t = \sinh (k_t) - k_t$.

{F}rom (\ref{rescaled_scattChi}), we find that 
\[ \o_{N,\ell} (x) =  1- \frac{\sin (\lambda^{1/2}_{N,\ell} (|x|-\ell))}{\lambda^{1/2}_{N,\ell}|x|} - \frac{\ell \cos (\lambda_{N,\ell}^{1/2} (|x|-\ell))}{|x|} \]
for all $x \in \bR^3$ with $R N^{-\beta} \leq |x| \leq \ell$, where $R > 0$ is the radius of the support of $V$. This gives the bound
\begin{equation}\label{eq:o-oasy} \left|N \o_{N,\ell} (x) - \o_{\ell}^\text{asymp} (x) \right| \leq \frac{C}{N^{1-\beta} |x|} \end{equation}
for a constant $C$ depending on $\ell$, and for all $x \in \bR^3$ with $|x| > RN^{-\beta}$. 

We use (\ref{eq:o-oasy}) to bound the difference $k_{N,t} - k_t$.
\begin{lemma} 
Let $V$ be as in Lemma \ref{lm:propomega} and $k_{N,t}$ be defined as in (\ref{eq:kNtA}), with $\ph = \ph_{t=0}^N \in H^4 (\bR^3)$ (so that $\| \ph^N_t \|_{H^4} \leq C e^{K|t|}$ for all $t \in \bR$, by Proposition \ref{prop:ph}). Let $\delta = \min (\beta/2, 1-\beta)$. There exist constants $C,c_1, c_2 > 0$ such that 
\[ \begin{split} 
\| k_{N,t} - k_{t} \|_2, \| p_{N,t} - p_{t} \|_2, 
\| r_{N,t} - r_t \|_2  &\leq C N^{-\delta} \exp (c_1 \exp (c_2 |t|)) \\
\| \nabla \left( p_{N,t} - p_{t} \right) \|_2 , 
\| \nabla (r_{N,t} - r_{t}) \|_2  &\leq C N^{-\delta} \exp (c_1 \exp (c_2 |t|)) \\
\| \Delta \left( p_{N,t} - p_{t} \right) \|_2 ,
\| \Delta (r_{N,t} - r_{t}) \|_2 &\leq C N^{-\delta} \exp (c_1 \exp (c_2 |t|)) 
\end{split} \]
\end{lemma}

\begin{proof}
We estimate
\begin{equation}\label{eq:k-k0} \begin{split}
 \| k_{N,t} -k_t\|_2^2 \leq \; &2\int  dx dy \, \left| N \o_{N,\ell}(x-y)  - \o_{\ell}^{\text{asymp}}(x-y) \right|^2 \, |\phn_t ((x+y)/2)|^2   \\  
 \; &+2 \int dx dy \, | \o_{\ell}^{\text{asymp}}(x-y)|^2 \, \Big| \phn_t ((x+y)/2)^2- \ph_t({(x+y)}/2)^2  \Big|^2 \\
= \; &\text{I} + \text{II} \end{split}
\end{equation}

Using (\ref{eq:o-oasy}), we find 
\[\begin{split}
|\text{I}| \leq \; &C N^{-2(1-\beta)} \int dx dy \frac{{\bf 1} (|x-y| \leq \ell)}{|x-y|^2} \, |\phn_t ((x+y)/2)|^2 \\ &+ C \int_{|x-y| \leq R N^{-\beta}} dx dy \frac{1}{|x-y|^2} |\phn_t ((x+y)/2)|^2 \\ \leq \; &C N^{-2+2\beta} +C  N^{-\beta}  \end{split} \]

On the other hand, using Proposition \ref{prop:ph}, we find 
\[ \begin{split} |\text{II}| &\leq C \int dx dy \frac{{\bf 1} (|x-y| \leq \ell)}{|x-y|^2} \Big| \phn_t ((x+y)/2)^2 - \ph_t ((x+y)/2)^2  \Big|^2 \\ &\leq C \int dR \, \Big| \phn_t (R)^2- \ph_t (R)^2   \Big|^2 \\ 
&\leq C \int dR \, |\phn_t (R) - \ph_t (R) |^2 \left( | \phn_t (R) |^2 + | \ph_t (R) |^2 \right)
\\ &\leq C e^{K|t|} \, \| \phn_t - \ph_t \|^2_2 \\ &\leq C N^{-2\gamma} \exp (c_1 \exp (c_2 |t|)) \qquad \text{with } \gamma = \min (\beta, 1-\beta) \end{split} \]
Hence, we conclude that
\[ \| k_{N,t} - k_t \|_2 \leq C N^{-\delta} \exp (c_1 \exp (c_2 |t|))  \]
as claimed. Furthermore, we have 
\begin{equation} \label{eq:p-pN} \begin{split}
&p_{N,t} - p_t = 
\sum_{n \geq 1} \frac{1}{(2n)!} \sum_{j=0}^{n-1} (k_t \bar k_t)^j (k_{N,t} \bar k_{N,t} - k_t \bar k_t)(k_{N,t} \bar k_{N,t} )^{n-1-j} \end{split}
\end{equation}
This implies that
\[ \| p_{N,t} - p_t \|_2 \leq C  e^{\| k_t \|_2 + \| k_{N,t} \|_2} \| k_{N,t} - k_t \|_2 \leq C N^{-\delta} \exp (c_1 \exp (c_2 |t|)) \]
Similarly, we find
\[ \|  r_{N,t} - r_t \|_2 \leq C N^{-\delta} \exp (c_1 \exp (c_2 |t|))  \]

{F}rom (\ref{eq:p-pN}), we find
\[\begin{split}
\nabla \big( p_{N,t} - p_t \big)  = \; &\nabla   (k_{N,t} \bar k_{N,t} - k_t \bar k_t) \Big[\sum_{n \geq 1} \frac{1}{(2n)!}(k_{N,t} \bar k_{N,t} )^{n-1} \Big] \\
&+ 
\nabla (k_t \bar k_t) \Big[ \sum_{n \geq 2} \frac{1}{(2n)!} \sum_{j=1}^{n-1} (k_t \bar k_t)^{j-1} (k_{N,t} \bar k_{N,t} - k_t \bar k_t)(k_{N,t} \bar k_{N,t} )^{n-1-j} \Big]  \end{split}
\]
Therefore 
\begin{equation}\label{eq:nab-pp}
\begin{split}
\|\nabla \big( p_{N,t} &- p_t \big)\|_2  \\ \leq \; & \|\nabla (k_{N,t} \bar k_{N,t} - k_t \bar k_t)\|_2 \;\sum_{n \geq 1} \frac{1}{(2n)!}  \|k_{N,t}\|^{2(n-1)}_2  \\
&+\|\nabla  (k_t \bar k_t)\|_2   \|k_{N,t} \bar k_{N,t} - k \bar k\|_2  \sum_{n \geq 1} \frac{1}{(2n)!}   \sum_{j=1}^{n-1} \|k_t\|_2^{2j}\,\|k_{N,t}\|^{2(n-1-j)}_2 \\
\leq \; &  \|\nabla (k_{N,t} \bar k_{N,t} - k_t \bar k_t )\|_2\,e^{\|k_{N,t}\|_2} + 
 \|k_{N,t} \bar k_{N,t} - k_t \bar k_t\|_2\, \|\nabla  (k_t \bar k_t)\|_2   \,e^{\|k_{\ell}\|_2} \\ \leq \; &C \| \nabla (k_{N,t} \overline{k}_{N,t} - k_t \overline{k}_\ell ) \|_2 + C N^{-\delta} \exp (c_1 \exp (c_2 |t|)) 
\end{split}
\end{equation}

We have
\begin{equation}\label{eq:nab-kk}
\|\nabla (k_{N,t} \bar k_{N,t} - k_t \bar k_t)\|_2 \leq \|\nabla k_{N,t} (\bar k_{N,t} - \bar k_t)\big\|_2 + \|\nabla ( k_{N,t} -  k_t)\bar k_t\|_2
\end{equation}
On the one hand, we find 
\begin{equation}\label{eq:nabkk-k}
\begin{split}
& \|\nabla k_{N,t} (\bar k_{N,t} - \bar k_t)\big\|_2^2  \\
&=    \int \di x \di y \di z_1 \di z_2 |\nabla_x k_{N,t}(x,z_1)|\,| (k_{N,t} -  k_t)(z_1,y)|\,|\nabla_x k_{N,t}(x,z_2)|\,| (k_{N,t} -  k_t)(z_2,y)| \\
& \leq   \int dx dy dz_1 dz_2 \, |(k_{N,t} - k_t) (z_1, y)|^2 |\nabla_x k_{N,t} (x,z_1)| |\nabla_x k_{N,t} (x,z_2)| \\  &\leq C  \| k_{N,t} - k_t \|^2_2 \end{split} \end{equation}
On the other hand, we compute
\begin{equation}\label{eq:nab-kk2} \begin{split} \| \nabla (k_{N,t} - k_t) \overline{k}_t \|_2^2 =\; & \int dx dy  \, \left| \int dz \, \nabla_x (k_{N,t} (x,z) - k_t (x,z) ) \overline{k}_t (z,y) \right|^2 \\
\leq \; &\int dx dy  \, \left| \int dz \, (\nabla_x + \nabla_z) (k_{N,t} (x,z) - k_t (x,z) ) \overline{k}_t(z,y) \right|^2 \\ &+ \int dx dy  \, \left| \int dz \, (k_{N,t} (x,z) - k_t (x,z) ) \nabla_z \overline{k}_t (z,y) \right|^2
\end{split} \end{equation}
integrating by parts. As in (\ref{eq:nabkk-k}), we can estimate
\[ 
\int dx dy  \, \left| \int dz \, (k_{N,t} (x,z) - k_t (x,z) ) \nabla_z \overline{k}_t (z,y) \right|^2 \leq C \| k_{N,t}  - k_t \|_2^2 \]
As for the first term on the r.h.s. of (\ref{eq:nab-kk2}), it can be estimated by 
\[ \begin{split} \int &dx dy  \, \Big| \int dz \, \Big[ N \omega_{N,\ell} (x-z) \phn_t ((x+z)/2) \nabla \phn_t ((x+z)/2) \\ & \hspace{4cm} - \omega_\ell^\text{asymp} (x-z) \ph_t ((x+z)/2) \nabla \ph_t ((x+z)/2) \Big] k_t (z,y) \Big|^2 \\ \leq \; &\sup_z \int dy |k_t (z,y)|^2 \\ &\times \left[ \int dx dz \, \left| N \omega_{N,\ell} (x-z) - \omega_\ell^\text{asymp} (x-z) \right|^2 | \phn_t ((x+z)/2)|^2 |\nabla \phn_t ((x+z)/2)|^2  \right. \\ &\hspace{.4cm} + \int dx dz \, \left| \omega_\ell^\text{asymp} (x-z) \right|^2 \left| \phn_t ((x+z)/2) - \ph_t ((x+z)/2)\right|^2 |\nabla \phn_t ((x+z)/2)|^2  \\ &\hspace{.4cm} \left. + 
\int dx dz \, \left| \omega_\ell^\text{asymp} (x-z) \right|^2  | \ph_t ((x+z)/2)|^2 |\nabla \phn_t ((x+z)/2) - \nabla \ph_t ((x+z)/2) |^2 \right] \end{split} \]
Proceeding as in the analysis of (\ref{eq:k-k0}) and using Proposition \ref{prop:ph}, we conclude that
\[ \begin{split} \int dx dy  \, \Big| \int dz \, (\nabla_x + \nabla_z) (k_{N,t} (x,z) - &k_t (x,z) ) \overline{k}_t (z,y) \Big|^2 \leq C N^{-\delta} \exp (c_1 \exp (c_2 |t|)) \end{split} \]
Inserting the last bound and (\ref{eq:nabkk-k}) in (\ref{eq:nab-kk}), we find 
\[ \|\nabla (k_{N,t} \bar k_{N,t} - k_t \bar k_t)\|_2 \leq C N^{-\delta} \exp (c_1 \exp (c_2 |t|)) \]
Hence, from (\ref{eq:nab-pp}), we conclude that 
\[ \| \nabla (p_{N,t} - p_t) \|_2 \leq C N^{-\delta} \exp (c_1 \exp (c_2 |t|)) \]
Analogously, we can show the bound for $\| \nabla (r_{N,t} - r_t) \|$. 

Also the bounds for $\| \Delta (p_{N,t} - p_t) \|$ and $\| \Delta (r_{N,t} - r_t) \|$ can be proven in a similar way. In fact, observing that
\[ \begin{split} \Delta_x k_{N,t} (x,y) = &-N \Delta \omega_{N,\ell} (x-y) (\phn_t ((x+y)/2))^2 \\ &- 2N \nabla \omega_{N,\ell} (x-y) \phn_t ((x+y)/2)) \nabla \phn_t ((x+y)/2) \\ &- N \omega_{N,\ell} (x-y) (\nabla \ph_t ((x+y)/2))^2 \\ &- N \omega_{N,\ell} (x-y) \phn_t ((x+y)/2) \Delta \phn_t ((x+y)/2) \end{split} \]
and that
\[ \begin{split} -N \Delta \omega_{N,\ell} (x-y) (\phn_t ((x+y)/2))^2 = \; & \frac{1}{2} N^{3\beta} V (N^\beta (x-y)) f_{N,\ell} (x-y) \\ &- N\lambda_{N,\ell} f_{N,\ell} (x-y) {\bf 1} (|x-y| \leq \ell) \end{split} \]
we find 
\[ \sup_x \int dy \, |\Delta_x k_{N,t} (x,y)| \leq C e^{K|t|} \]
and therefore, similarly to (\ref{eq:nabkk-k}), 
\begin{equation}\label{eq:Delta1} \| \Delta k_{N,t} (\overline{k}_{N,t} - \overline{k}_t ) \|_2 \leq C e^{K|t|} \| k_{N,t} - k_t \|_2 \leq C N^{-\delta} \exp (c_1 \exp (c_2 |t|))  \end{equation}
Moreover, integrating by parts twice, we obtain
\[ \begin{split} \| \Delta (k_{N,t} - k_t) \overline{k}_t \|_2^2 \leq &\int dx dy \left| \int dz (\nabla_x + \nabla_z)^2 (k_{N,t} (x,z) - k_t (x,z)) \overline{k}_t (z,y) \right|^2 \\ &+ \int dx dy \left| \int dz (\nabla_x + \nabla_z) (k_{N,t} (x,z) - k_t (x,z)) \nabla_z \overline{k}_t (z,y) \right|^2\\ &+
\int dx dy \left| \int dz (k_{N,t} (x,z) - k_t (x,z)) \Delta_z \overline{k}_t (z,y) \right|^2\end{split} \]
This implies that 
\[ \| \Delta (k_{N,t} - k_t) \overline{k}_t \|_2^2 \leq C N^{-\delta} \exp (c_1 \exp (c_2 |t|))  \]
and thus, together with (\ref{eq:Delta1}), that
\[ \| \Delta (k_{N,t} \overline{k}_{N,t} - k_t \overline{k}_\ell ) \|_2 \leq C N^{-\delta} \exp (c_1 \exp (c_2 |t|)) \]
The last bound allows us to conclude that
\[ \begin{split} \| \Delta (p_{N,t} - p_t) \|_2 &\leq C N^{-\delta} \exp (c_1 \exp (c_2 |t|))  \\ \| \Delta (r_{N,t} - r_t) \|_2 &\leq C N^{-\delta} \exp (c_1 \exp (c_2 |t|)) \end{split} \]
\end{proof}


\begin{thebibliography}{}

\bibitem{AGT} R.~{Adami}, F.~{Golse} and A.~{Teta}. 
{Rigorous derivation of the cubic NLS in dimension one}. \emph{J. Stat. Phys.} \textbf{127} (2007), 1193-1220.


\bibitem{AFP} Z.~{Ammari}, M.~{Falconi} and B.~{Pawilowski}. {On the rate of convergence for the mean-field approximation of many-body quantum dynamics}. Preprint arXiv:1411.6284.

\bibitem{AN} Z.~{Ammari} and F.~{Nier}. {Mean-field limit for bosons and propagation of Wigner measures}. \emph{J. Math. Phys.} \textbf{50} (2009), 042107.

  
\bibitem{BGM}  C.~{Bardos}, F.~{Golse} and N.~J.~{Mauser}. {Weak coupling limit of the $N$-particle Schr\"odinger equation}. \emph{Methods Appl. Anal.} \textbf{2} (2000), 275--293.

\bibitem{BKS}
G.~{Ben Arous}, K.~{Kirkpatrick} and B.~{Schlein}. {A central limit
  theorem in many-body quantum dynamics}. \emph{Comm. Math. Phys.} \textbf{321} (2013), 371--417.

\bibitem{BdOS12}
N.~{Benedikter}, G.~{de Oliveira} and B.~{Schlein}. {Quantitative derivation of the Gross-Pitaevskii equation}. \emph{Comm. Pure Appl. Math.} \textbf{68}  (2015), no. 8, 1399–1482.




\bibitem{BSS} S.~{Buchholz}, C.~{Saffirio} and B.~{Schlein}. {Multivariate central limit theorem in quantum dynamics}. \emph{J. Stat. Phys.} \textbf{154} (2014), 113--152.

\bibitem{XC} X.~{Chen}. {Second order corrections to mean-field evolution for weakly interacting
bosons in the case of three-body interactions}. \emph{Archive for Rational Mechanics and
Analysis} \textbf{203} (2012), 455--497.

\bibitem{CH} X.~{Chen} and J.~{Holmer}. {Focusing quantum many-body dynamics: the rigorous derivation of the $1D$ focusing cubic nonlinear Schr\"odinger equation}. Preprint arXiv:1308.3895.

\bibitem{CH2} X.~{Chen} and J.~{Holmer}. {Correlation structures, many-body scattering processes and the derivation of the Gross-Pitaevskii hierarchy}. Preprint arXiv:1409.1425.

\bibitem{CLS}
L.~{Chen}, J.~O.~{Lee} and B.~{Schlein}. {Rate of convergence towards {H}artree dynamics}. \emph{J. Stat. Phys.} \textbf{144} (2011), no.~4, 872--903.


\bibitem{CP} T.~{Chen} and N.~{Pavlovi\'{c}}. {The quintic NLS as the mean-field limit of a boson gas with three-body interactions}. \emph{J. Funct. Anal.} \textbf{260} (2011), 959--997.

\bibitem{DN}
J. Derezinski and M. Napi{\'o}rkowski. Excitation spectrum of interacting bosons in the mean-field infinite-volume limit. {\it Annales Henri Poincar{\'e}} {\bf 15} (12), 2409--2439.



\bibitem{EMS}
L. Erd\H os, A. Michelangeli, and B. Schlein. Dynamical formation of correlations in a Bose-Einstein condensate. {\it Comm. Math. Phys.} {\bf 289} (2009), no. 3, 1171--1210.

\bibitem{ES} A.~{Elgart}, B.~{Schlein}.  Mean-field dynamics for boson
stars. \emph{Comm. Pure Applied Math.} \textbf{60} (2007), no. 4, 500--545.

\bibitem{ESY0}
L.~Erd\"os, B.~Schlein, and H.T.~Yau. Derivation of the {G}ross-{P}itaevskii hierarchy for the dynamics of {B}ose-{E}instein condensate. {\it
Comm. Pure Appl. Math}, {\bf 59} (12), 2006.


\bibitem{ESY1}
L.~{Erd\H{o}s}, B.~{Schlein} and H.-T.~{Yau}.
\newblock {Derivation of the cubic nonlinear Schr\"{o}dinger equation from
  quantum dynamics of many-body systems}.
\newblock \emph{Inv. Math.} \textbf{167} (2006), 515--614.

\bibitem{ESY2}
L.~{Erd{\H{o}}s}, B.~{Schlein} and H.-T.~{Yau}.
\newblock Derivation of the {G}ross-{P}itaevskii equation for the dynamics of
  {B}ose-{E}instein condensate,
\newblock \emph{Ann. of Math. (2)} \textbf{172} (2010), no. 1, 291--370.

\bibitem{ESY3}
L.~{Erd\H{o}s}, B.~{Schlein} and H.-T.~{Yau}.
\newblock {Rigorous derivation of the Gross-Pitaevskii equation with a large
  interaction potential}.
\newblock \emph{J. Amer. Math. Soc.} \textbf{22} (2009), 1099--1156.

\bibitem{ESY9}
L.~{Erd\H{o}s}, B.~{Schlein} and H.-T.~{Yau}.
\newblock {The ground state energy of a low density Bose gas: a second order upper bound}.
\newblock \emph{Phys.\ Rev.\ A} \textbf{78} (2008), 053627.

\bibitem{EY}
L.~{Erd{\H{o}}s} and H.-T.~{Yau}. {Derivation of the
  nonlinear {S}chr\"{o}dinger equation from a many-body {C}oulomb system}. \emph{Adv.
  Theor. Math. Phys.} \textbf{5} (2001), no.~6, 1169--1205.



\bibitem{FK}
J.~{Fr\"{o}hlich}, A.~{Knowles} and S.~{Schwarz}. {On the mean-field limit of bosons with {C}oulomb two-body interaction}. \emph{Comm. Math. Phys.} \textbf{288} (2009), no.~3, 1023--1059.

\bibitem{GV}
J.~Ginibre and G.~Velo. {The classical field limit of scattering theory
  for nonrelativistic many-boson systems. {I} and {II}}. \emph{Comm. Math. Phys.} \textbf{66}
  (1979), no.~1, 37--76, and \textbf{68} (1979), no.~1, 45--68.
  
\bibitem{GS}
Grech, R. Seiringer. The Excitation Spectrum for Weakly Interacting Bosons in a Trap. {\it Comm. Math. Phys.} {\bf 322} (2013), no. 2, 559--591.

\bibitem{GM}
M.~Grillakis, M.~Machedon. Pair excitations and the mean field approximation of interacting bosons, I.
{\it Comm. Math. Phys.} {\bf 324} (2013), 601-636.

\bibitem{GMM1}
M.~Grillakis, M.~Machedon and D.~Margetis. {Second-order corrections to 
mean-field evolution of weakly interacting bosons. {I}.} \emph{Comm. Math. Phys.} \textbf{294} 
(2010), no.~1, 273--301

\bibitem{GMM2}
M.~Grillakis, M.~Machedon and D.~Margetis. {Second-order corrections to mean-field evolution of weakly interacting bosons. {II}.} \emph{Adv. Math.} \textbf{228} (2011), no.~3, 1788--1815.

\bibitem{Hepp}
K.~{Hepp}. {The classical limit for quantum mechanical correlation
  functions}. \emph{Comm. Math. Phys.} \textbf{35} (1974), 265--277.


\bibitem{KSS}
K.~{Kirkpatrick}, B.~{Schlein} and G.~{Staffilani}. {Derivation of the two dimensional nonlinear
{S}chr\"odinger equation from many-body quantum dynamics}. \emph{Amer. J. Math.} \textbf{133} (2011), no.~1, 91-130. 


\bibitem{KP}
A.~{Knowles} and P.~{Pickl}. {Mean-field dynamics: singular potentials
  and rate of convergence}. \emph{Comm. Math. Phys.} \textbf{298} (2010), no.~1,
  101--138.


\bibitem{LNS} M.~{Lewin}, P.~T.~{Nam} and B.~{Schlein}. {Fluctuations around Hartree states in the mean-field regime}. Preprint arXiv:1307.0665.

\bibitem{LNSS} M.~Lewin, P.~T.~{Nam}, S.~{Serfaty}, J.P. {Solovej}. Bogoliubov spectrum of interacting Bose gases. Preprint arxiv:1211.2778. 

\bibitem{LS}
E.~H.~Lieb and R.~Seiringer.
\newblock Proof of {B}ose-{E}instein condensation for dilute trapped gases.
\newblock \emph{Phys. Rev. Lett.} \textbf{88} (2002), 170409.

\bibitem{LSY}
E.~H.~Lieb, R.~Seiringer, and J.~Yngvason.
\newblock Bosons in a trap: A rigorous derivation of the {G}ross-{P}itaevskii
  energy functional. \newblock \emph{Phys. Rev. A} 
\textbf{61} (2000), 043602.


\bibitem{P} P.~{Pickl}. {Derivation of the time dependent Gross Pitaevskii equation with external fields}. Preprint arXiv:1001.4894.

\bibitem{RS}
I.~{Rodnianski} and B.~{Schlein}. {Quantum fluctuations and rate of
  convergence towards mean-field dynamics}. \emph{Comm. Math. Phys.} \textbf{291}
  (2009), no.~1, 31--61.


\bibitem{S}
R. Seiringer. The Excitation Spectrum for Weakly Interacting Bosons. {\it Commun. Math. Phys.} {\bf 306} (2011), 565Ð578.


\bibitem{Spohn}
H.~{Spohn}. {Kinetic equations from {H}amiltonian dynamics: {M}arkovian
  limits}. \emph{Rev. Modern Phys.} \textbf{52} (1980), no.~3, 569--615.
  
\bibitem{YY}
H.-T.~{Yau} and J.~{Yin}. {{The Second Order Upper Bound for the Ground Energy of a Bose Gas}}, \emph{J.\ Stat.\ Phys.} \textbf{136}
  (2009), no.~3, 453--503.
\end{thebibliography}
\end{document}